\documentclass[hidelinks,sort&compress,3p,twoside,twocolumn]{article}

\usepackage[labelfont=bf]{caption}
\usepackage{tikz}
\usepackage{multirow}
\usepackage{natbib}
\usepackage{amsmath,epsfig,graphicx}

\usepackage{subfigure}

\usepackage{color}
\usepackage[affil-it]{authblk}
\usepackage{wrapfig}

\DeclareMathOperator*{\argmax}{argmax}
\DeclareMathOperator*{\argmin}{argmin}

\usepackage{accents}
\newlength{\dhatheight}

\title{Maximum likelihood recursive state estimation in state-space models: \\A new approach based on statistical analysis of incomplete data }

\author[]{Budhi Arta Surya\footnote{Email address: \texttt{budhi.surya@vuw.ac.nz}}}
\affil[]{School of Mathematics and Statistics, Victoria University of Wellington, \\Gate 6, Kelburn PDE, Wellington 6140, New Zealand}

\usetikzlibrary{automata,positioning}

\allowdisplaybreaks

\bibpunct{(}{)}{;}{a}{,}{,}

\usepackage{amsfonts,amssymb,amsthm,booktabs,color,epsfig,graphicx,hyperref,url}

\theoremstyle{plain}% Theorem-like structures provided by amsthm.sty
\newtheorem{theorem}{Theorem}
\newtheorem{proposition}{Proposition}
\newtheorem{lemma}{Lemma}
\newtheorem{corollary}{Corollary}
\newtheorem{remark}{Remark}
\newtheorem{example}{Example}
\newtheorem{assumption}{Assumption}

\begin{document}

\maketitle \pagestyle{myheadings} \markboth{B.A. Surya}{Maximum likelihood recursive state estimation in state-space models}

\begin{abstract}
This paper revisits the work of \cite{Rauch} and develops a novel method for recursive maximum likelihood particle filtering for general state-space models. The new method is based on statistical analysis of incomplete observations of the systems. Score function and conditional observed information of the incomplete observations/data are introduced and their distributional properties are discussed. Some identities concerning the score function and information matrices of the incomplete data are derived. Maximum likelihood estimation of state-vector is presented in terms of the score function and observed information matrices. In particular, to deal with nonlinear state-space, a sequential Monte Carlo method is developed. It is given recursively by an EM-gradient-particle filtering which extends the work of \cite{Lange} for state estimation. To derive covariance matrix of state-estimation errors, an explicit form of observed information matrix is proposed. It extends \cite{Louis} general formula for the same matrix to state-vector estimation. Under (Neumann) boundary conditions of state transition probability distribution, the inverse of this matrix coincides with the Cram\'er-Rao lower bound on the covariance matrix of estimation errors of unbiased state-estimator. In the case of linear models, the method shows that the Kalman filter is a fully efficient state estimator whose covariance matrix of estimation error coincides with  the Cram\'er-Rao lower bound. Some numerical examples are discussed to exemplify the main  results. 
\medskip


%\noindent \textbf{MSC 2020}: 60J20; 60J27; 60J28; 62N99; 62H05

\end{abstract}

%\tableofcontents

\section{Introduction\label{sec:1}}
The recursive filtering algorithm proposed by \cite{Kalman} has been an important tool for online estimation of (hidden) state variables, which are partially observed in noisy observations. It has found immense applications in various different fields spanning from engineering for tracking and navigation (see e.g., \cite{Bar-Shalom}, \cite{Musoff2005}, \cite{Schmidt66} and \cite{Smith62}), economics/time series analysis (see e.g., \cite{Durbin}, \cite{Harvey}, \cite{Kitagawa}), molecular and RNA diagnoses (\cite{Kelemen}, \cite{Vaziri}), just to mention a few. See relevant literature therein for more details. The key to the development of the Kalman filter lies on the state-space representation of an underlying physical/economical/biological systems under consideration. An estimate $\widehat{x}_{k\vert x}$ of state-vector $x_k$ is obtained as an orthogonal projection $\mathbb{E}[x_k\vert \underline{y}_k,\theta_k]$ of the state onto a sequence of past observations $\underline{y}_k:=\{y_0,y_1,\ldots,y_k\}$, for a given systems parameter $\theta_k$, as such that $\widehat{x}_{k\vert x}$ forms an unbiased estimator of $x_k$ with minimum mean-square error (MMSE). See \cite{Anderson}, \cite{Bar-Shalom} and the above references.

To deal with nonlinear state-space, the Kalman filter is implemented using localization method by linearizing the state-space around a current estimate of the state. This is called an extended Kalman filter, see e.g. \cite{Jazwinski}, \cite{Anderson},  \cite{Bar-Shalom}, and \cite{Musoff2005}. The localization method implies that the required probability distribution of the state is approximated by a Gaussian distribution which causes a distortion of the true distribution of the state resulting in filter divergence. Hence, the extended Kalman filter is suboptimal. To overcome this difficulty, \cite{Gordon} proposed a recursive Bayesian state estimation based on a sequential Monte Carlo (importance) sampling from posterior distribution $f(x_k\vert\underline{y}_k,\theta_k)$. This new approach is known as the particle filtering which has been successfully used in variety of applications in various fields. See \cite{Godsill} for a survey of recent developments of the particle filtering, see also \cite{Doucet,Doucet2000}, \cite{Kitagawa2021,Kitagawa,Kitagawa93}, \cite{Ristic} and \cite{Shephard}. The merit of the Bayesian approach is that an optimal state estimator $\widehat{x}_{k\vert k}=\mathbb{E}[x_k\vert \underline{y}_k,\theta_k]$ can be derived for non-Gaussian/nonlinear general state-space models, despite the approach lacks in providing an explicit estimator of the covariance matrix of estimation errors, which can be used to construct a $95\%$ confidence interval of the state estimates.

Maximum likelihood (ML) estimation of state was discussed (for linear state-space model) by \cite{Rauch}. In this approach, an estimator $\widehat{x}_k$ of state-vector $x_k$ is found as the maximizer of (the logarithm of) posterior distribution $f(x_k\vert\underline{y}_k,\theta_k)$, i.e., $\widehat{x}_k:=\argmax_{x_k} \log f(x_k\vert\underline{y}_k,\theta_k)$. In the case that the posterior distribution $f(x_k\vert\underline{y}_k,\theta_k)$ is unimodal and symmetry, the ML estimator $\widehat{x}_k$ and the MMSE $\widehat{x}_{k\vert k}$ coincide. Otherwise, they exhibit different unbiased estimators of the state-vector. For nonlinear state-space models, the ML estimation of \cite{Rauch} may not be applicable to derive an explicit estimator of the state. In their recent work, \cite{Ramadan} used the expectation-maximization (EM) algorithm to derive a recursive ML estimator $\{\widehat{x}_k^{(\ell)}\}_{\ell\geq 1}$ of the state $x_k$, obtained after the sequence converges. The update $\widehat{x}_k^{(\ell+1)}$ is defined as the maximizer of conditional expectation of log-posterior distribution $\log f(x_k\vert x_{k-1},\underline{y}_k,\theta_k)$ given past observations $\underline{y}_k$ and $\widehat{x}_k^{(\ell)}$. To be more precise, $\widehat{x}_k^{(\ell+1)}
:=\argmax_{x_k}\mathbb{E}[\log f(x_k\vert x_{k-1},\underline{y}_k,\theta_k)\vert \widehat{x}_k^{(\ell)}, \underline{y}_k,\theta_k]$. The main appealing feature of this approach is that it works for ML estimation of non-Gaussian/nonlinear state-space models. 

However, it is known in particular for parameter estimation that the EM-algorithm has (semi) linear rate of convergence compared to the EM-Gradient algorithm of \cite{Lange}, which achieves its convergence at quadratic rate, hence faster than the EM-algorithm. See also \cite{McLachlan}. This was demonstrated for e.g. in \cite{Surya2022} for parameter estimation of conditional Markov jump processes. Furthermore, the EM-algorithm most notably does not directly produce the score function nor the observed information matrix of the incomplete observations $(x_k,\underline{y}_k)$ of the state-space, both of which were absence in \cite{Rauch} and \cite{Ramadan}, needed to derive the ML estimator and the covariance matrix of estimation error, respectively. Although \cite{VanTrees68}, \cite{VanTrees}, \cite{Gill}, \cite{Tichavsky}, \cite{Bar-Shalom}, \cite{Bergman}, \cite{Ristic} and \cite{Zhao} provided Bayesian Cram\'er-Rao lower bounds for the covariance matrix of (random) parameter estimates, the results are not directly applicable for the ML estimator of state-vector. 

The main reasons that the Bayesian Cram\'er-Rao lower bounds being not directly applicable to ML state estimator are the following. First, the information matrix is assumed to be invertible at each time-step, which may not be the case in general. Secondly, the results are established using imposed boundary conditions on the prior distribution $f(x_k\vert\theta_k)$, that is $\lim\limits _{x_k\rightarrow \pm \infty} f(x_k\vert\theta_k)\mathbb{E}\big[\big(T(\underline{y}_k)-x_k)\vert x_k,\theta_k\big]=0$, with $T(\underline{y}_k)$ being an estimator of $x_k$, see eqns. (4.178) and (4.179), p. 261 in \cite{VanTrees}. Slightly differently, \cite{Gill} and \cite{Bergman} uses, among other conditions, $\lim\limits _{x_k\rightarrow \pm \infty} f(x_k\vert\theta_k)=0$. These conditions might be difficult to verify due to nonlinearity of the state-space. Thirdly, the lower bounds are given in terms of expected Fisher information matrix, which in general is also difficult to evaluate as it involves integration with respect to the prior distribution $f(x_k\vert\theta_k)$. Beside that, the above papers did not discuss how to use the results to evaluate the information matrix of the incomplete data $(x_k,\underline{y}_k)$. For some cases - such as linear state-space with Gaussian noises and zero process noise - a recursive equation for the information matrix was discussed in \cite{Tichavsky}, \cite{Bergman}, Chapter 4 of \cite{Ristic} and \cite{Zhao}. However, the information matrix discussed in these papers was concerned with the state-vector $\underline{x}_k$, not for $(x_k,\underline{y}_k)$. The valuation of expected information matrix becomes much harder if the true state-vector $x_k$ is replaced by a respective state estimator $T(\underline{y}_k)$ in the matrix. 

\subsection*{Research contributions}
This paper proposes a novel approach based on statistical analysis of incomplete data to derive recursive ML state estimator $\widehat{x}_k$. In this framework, we consider $(x_k,\underline{y}_k)$ as an incomplete observation/data of a (nonlinear) state-space systems represented by a complete data $(x_k,\underline{y}_k,\underline{x}_{k-1})$, with $\underline{z}_k:=\{z_0,z_1,\ldots,z_k\}$, $z_{\ell}\in\mathbb{R}^p$, $p\geq 1$. The respective score function of $(x_k,\underline{y}_k)$ is introduced along with the corresponding conditional observed information matrices. Recursive equations for ML state estimator $\widehat{x}_k$ are developed in terms of the score function and information matrices. Their availability in general form makes them available for ML state estimation in non-Gaussian/nonlinear state-space models. 

Some distributional properties and identities of the score function and observed information matrices of the incomplete data $(x_k,\underline{y}_k)$ are established. The latter have the same Loewner partial ordering properties as the expected information matrices do and are positive definite when evaluated at the ML state estimator $\widehat{x}_k$. An explicit form of the corresponding observed Fisher information (OFI) matrix of $(x_k,\underline{y}_k)$ is derived. It extends \cite{Louis} general formula of the same matrix for (non random) parameter to that of for state-vector. We show under Neumann boundary condition $\lim\limits _{x_k\rightarrow \pm \infty} \frac{\partial f(x_k\vert x_{k-1}, \theta_k)}{\partial x_k}=0$ and $\lim\limits _{x_k\rightarrow \pm \infty} f(x_k\vert x_{k-1}, \theta_k)=0$ that the covariance matrix of estimation error of the ML estimator is given by the inverse of expected OFI matrix. Furthermore, under the same boundary conditions, we prove that the expected OFI matrix serves as the Cram\'er-Rao lower bound on covariance matrix of estimation errors of any unbiased state estimator. 

As opposed to those discussed in  \cite{Tichavsky}, \cite{Bergman}, \cite{Ristic} and \cite{Zhao}, the Cram\'er-Rao lower bound is replaced by the inverse of sample average of the derived observed information matrix, which is evaluated using particle filtering of algorithm. In the case that OFI matrix is lack of sparsity, an alternative recursive equation is proposed to calculate the inverse. 

This paper is organized as follows. Section \ref{sec:sec2} overviews the Kalman filtering as the unbiased MMSE estimator. Also, this section discusses \cite{Rauch} model. The main results and contributions of the paper are presented in Section \ref{sec:sec3}. Applications to linear and nonlinear state-space models are discussed in Section \ref {sec:appl}. Section \ref{sec:numeric} presents numerical examples before the paper is concluded in Section \ref{sec:conclusion}. Some lengthy proofs are deferred to the Appendix.

\section{Overview of Kalman filtering}\label{sec:sec2}

\subsection{Minimum-variance unbiased estimator}\label{sec:mmse}

Suppose that a true controlled state vector $x_k\in\mathbb{R}^p$, $p\geq 1$, evolves according to a recursive linear state-space model
\begin{align}\label{eq:eq1a}
x_k=F_k x_{k-1} + G_k u_{k} + v_k, \tag{1a}
\end{align}
and is observed through noisy observation $y_k\in\mathbb{R}^q$, $q<p$,
\begin{align}\label{eq:eq1b}
y_k=H_k x_k + w_k,  \tag{1b}
\end{align}
where $v_k$ and $w_k$ are independent sequence of random variables with $v_k\sim N(0,Q_k)$ and $w_k\sim N(0,R_k)$. Assume that all matrices $F_k$, $G_k$, $H_k$, $Q_k$ and $R_k$ are conformable and prespecified, whereas $u_k\in\mathbb{R}^m$, $m\leq p$, forms a sequence of predetermined control variables. Also, we assume that the initial state $x_0\sim N(\mu,P_0)$. Let $\theta_k=(F_{\ell},G_{\ell},H_{\ell},Q_{\ell},R_{\ell}: \ell=1,\ldots,k)\cup (\mu,P_0)$ denote the parameters of the linear systems (1). The Kalman filtering is concerned with finding an estimator $\widehat{x}_{k\vert k}$ of the unknown state-vector $x_k$ which minimizes the estimation error $\mathbb{E}[\Vert x_k - \widehat{x}_{k\vert k} \Vert^2 \vert \theta_k ]$ based on a sequence of observations $\underline{y}_k=(y_{\ell}: \ell=1,\ldots, k)$, i.e.,
\setcounter{equation}{1}
\begin{align}\label{eq:objfunc}
\widehat{x}_{k\vert k}=\argmin_{\widehat{y}\in\mathcal{Y}}\mathbb{E}[\Vert x_k - \widehat{y} \Vert^2 \vert \theta_k ],
\end{align}
with $\mathcal{Y}$ representing a set of any possible values of $\underline{y}_k$. Hence, it is a minimum mean square error (MMSE) problem for which the solution is given by $\widehat{x}_{k\vert k}=\mathbb{E}[x_k\vert \underline{y}_k,\theta_k]$. See, e.g., \cite{Kalman}, \cite{Astrom} and \cite{Anderson}. This can be seen that the estimation error $\varepsilon_k=x_k- \widehat{x}_{k\vert k}$ is orthogonal to any vector $\widehat{y}\in\mathcal{Y}$. To be more precise, by iterated law of conditional expectation, $\mathbb{E}[\widehat{y}^{\top}\varepsilon_k\vert \theta_k]=\mathbb{E}\big[\mathbb{E}[\widehat{y}^{\top}\varepsilon_k\vert \underline{y}_k,\theta_k]\big\vert \theta_k\big]=\mathbb{E}\big[\widehat{y}^{\top}\mathbb{E}[\varepsilon_k\vert \underline{y}_k,\theta_k] \big\vert \theta_k\big]=0$. The latter is due to $\mathbb{E}[\varepsilon_k\vert \underline{y}_k,\theta_k]=\mathbb{E}[x_k\vert \underline{y}_k,\theta_k]-\widehat{x}_{k\vert k}=\underline{0}$. Hence, $\widehat{x}_{k\vert k}$ is an unbiased estimator of unknown state-vector $x_k$.

To derive recursive equations for $\widehat{x}_{k\vert k}$ (\ref{eq:objfunc}), the following observation is required. It is straightforward to check using induction argument that the solution of (\ref{eq:eq1a}) takes the form
\begin{equation}\label{eq:solofx}
\begin{split}
x_k=\Phi_{k+1,1}x_0 + \sum_{n=1}^k \Phi_{k+1,n+1} \big(G_n u_n + v_n\big),
\end{split}
\end{equation}  
where $\Phi_{k+\ell,k}$ is a conformable matrix defined by
\begin{align*}
\Phi(k+\ell,k)=F(k+\ell-1)\ldots F(k), \quad  \ell \geq 1,
\end{align*}
with $\Phi_{k,k}=\mathbf{I}$, an identity matrix. Thus, $x_k$ is a Gaussian random variable. It also follows from (1) and (\ref{eq:solofx}) that $x_k$ and $y_k$ are jointly Gaussian. Next, consider an innovation process $\overline{y}_k=y_k-\mathbb{E}[y_k\vert \underline{y}_{k-1},\theta_k]$, proposed by \cite{Kailath}. It is clear that $\overline{y}_k$ is orthogonal to any vector in the space spanned by the observations $\underline{y}_{k-1}$. By independence of $w_k$, it follows from (\ref{eq:eq1b}) that $\mathbb{E}[y_k\vert \underline{y}_{k-1},\theta_k]=H_k \widehat{x}_{k\vert k-1}$, with $\widehat{x}_{k\vert k-1}\equiv \mathbb{E}[x_k\vert \underline{y}_{k-1},\theta_k]$ denoting a priori estimate of $x_k$ based on $\underline{y}_{k-1}$. Following \cite{Kailath}, also by applying Theorem 3.3 in Ch. 7 of \cite{Astrom}, we obtain
\begin{align}\label{eq:id1}
\widehat{x}_{k\vert k}=&\mathbb{E}[x_k\vert \underline{y}_k,\theta_k]=\mathbb{E}[x_k\vert \overline{y}_k,\underline{y}_{k-1},\theta_k] \notag\\
=& \mathbb{E}[x_k\vert\underline{y}_{k-1},\theta_k] + \mathbb{E}[x_k\vert \overline{y}_k,\theta_k]  -\mathbb{E}[x_k\vert \theta_k].
\end{align}
Applying Theorem 3.2 of Ch.7 in \cite{Astrom}, 
\begin{align}\label{eq:id2}
 \mathbb{E}[x_k\vert \overline{y}_k,\theta_k] = \mathbb{E}[x_k\vert \theta_k] + \overline{R}_{x_k,\overline{y}_k}\overline{R}_{\overline{y}_k,\overline{y}_k}^{-1}\overline{y}_k,
\end{align}
where the covariance matrices $ \overline{R}_{x_k,\overline{y}_k}=\mathrm{Cov}(x_k,\overline{y}_k\vert \theta_k)$ and $\overline{R}_{\overline{y}_k,\overline{y}_k}=\mathrm{Cov}(\overline{y}_k,\overline{y}_k\vert \theta_k)$ are respectively defined by
\begin{align*}
 \overline{R}_{x_k,\overline{y}_k}=& P_{k\vert k-1} H_k^{\top},\\
 \overline{R}_{\overline{y}_k,\overline{y}_k}=& H_k P_{k\vert k-1} H_k^{\top} + R_k.
 \end{align*}
Recall that $P_{k\vert k-1}\equiv\mathbb{E}[(x_k - \widehat{x}_{k\vert k-1})(x_k - \widehat{x}_{k\vert k-1})^{\top}\vert \theta_k]$ denotes a priori estimate of covariance matrix of the estimation error $x_k -\widehat{x}_{k\vert k-1}$. Using the two identities (\ref{eq:id1}) and (\ref{eq:id2}), the solution $\widehat{x}_{k\vert k}$ of (\ref{eq:objfunc}) is given for $k\geq 0$ by
\begin{eqnarray}\label{eq:kalman}
\widehat{x}_{k\vert k}&=&\widehat{x}_{k\vert k-1} + K_k (y_k - H_k\widehat{x}_{k\vert k-1}),\\
K_k&=&P_{k\vert k-1} H_k^{\top}(H_k P_{k\vert k-1}H_k^{\top} + R_k)^{-1}. \nonumber
\end{eqnarray}
The (conformable) matrix $K_k$ is widely known as the Kalman gain. Equivalently, one can show that $K_k$ is the minimizer of the matrix function $J(K_k)=\mathrm{tr}(P_{k\vert k})$, where $P_{k\vert k}\equiv\mathbb{E}[(x_k - \widehat{x}_{k\vert k})(x_k - \widehat{x}_{k\vert k})^{\top}\vert \theta_k]$ is the a posteriori estimate of the covariance matrix, defined for $k\geq 0$ by 
\begin{align*}
P_{k\vert k}=& (I - K_k H_k) P_{k\vert k-1}  (I - K_k H_k)^{\top} + K_k R_k K_k^{\top}.
\end{align*}
Using explicit form of the matrix $K_k$, $P_{k\vert k}$ simplifies into
\begin{align}\label{eq:kalmancovmat}
P_{k\vert k}=(I- K_k H_k) P_{k\vert k-1}.
\end{align}
Since $P_{k\vert k}$ and $P_{k\vert k-1}$ are both positive definite, so is the matrix $ (I - K_k H_k) $. Also, $ K_k H_k$. Hence, $0< K_k H_k <I$ and therefore all eigenvalues of $I - K_k H_k$ are positive and strictly less than one. See Corollary 1.3.4 in \cite{Horn}. Thus, the posteriori estimate $P_{k\vert k}$ of covariance matrix is always less than the priori estimate $P_{k\vert k-1}$. Furthermore, by independence of $v_k$ we have for $k\geq 1$,
\begin{equation}\label{eq:sol2}
\begin{split}
\widehat{x}_{k\vert k-1}  =&F_k \widehat{x}_{k-1\vert k-1} + G_k u_k,\\
P_{k\vert k-1} =& F_k P_{k-1\vert k-1} F_k^{\top} + Q_k,
\end{split}
\end{equation}
The algorithm (\ref{eq:kalman})-(\ref{eq:sol2}) are run with initial conditions:
\begin{align*}
\widehat{x}_{0\vert -1}=\mu \quad \mathrm{and} \quad P_{0\vert -1}=P_0.
\end{align*}

To summarize, the MMSE estimator $\widehat{x}_{k\vert k}=\mathbb{E}[x_k\vert \underline{y}_k,\theta_k]$ is a minimum-variance unbiased (MVU) estimator of $x_k$.

\subsection{Maximum likelihood estimator}\label{sec:secMLE}

In their paper, \cite{Rauch} derived the recursive equations (\ref{eq:kalman})-(\ref{eq:sol2}) for $\widehat{x}_k$ as a maximum likelihood (ML) estimator of $x_k$. According to eqn. (3.1) in \cite{Rauch}, $\widehat{x}_k$ is defined as (global) maximizer of $\log f(x_k\vert \underline{y}_k,\theta_k)$, i.e.,
\begin{align*}
\widehat{x}_k=\argmax_{x_k} \log f(x_k \vert \underline{y}_k,\theta_k),
\end{align*}
where $ f(x_k \vert \underline{y}_k,\theta_k)$ defines the probability density function of $x_k$ given observations $\underline{y}_k$, characterized by the parameter $\theta_k$ of the linear system (1). Equivalently, $\widehat{x}_k$ solves
\begin{align}\label{eq:mle}
\frac{\partial \log  f(x_k \vert \underline{y}_k,\theta_k)}{\partial x_k}=0.
\end{align}

In the case the conditional distribution $f(x_k\vert \underline{y}_k,\theta_k)$ is unimodal and symmetry, the solution of (\ref{eq:mle}) coincides with (\ref{eq:objfunc}), i.e., $\widehat{x}_k=\mathbb{E}[x_k\vert \underline{y}_k,\theta_k]=\widehat{x}_{k\vert k}$. In other case, in order to evaluate the conditional expectation $\mathbb{E}[x_k\vert \underline{y}_k,\theta_k]$, it remains to specify the conditional distribution $f(x_k \vert \underline{y}_k,\theta_k)$. By Bayes' formula, $f(x_k \vert \underline{y}_k,\theta_k)$ can be rewritten as
\begin{align}\label{eq:bayes}
 f(x_k \vert \underline{y}_k,\theta_k)=\frac{f(x_k,\underline{y}_k\vert \theta_k)}{f(\underline{y}_k\vert \theta_k)}, 
\end{align}
where following (1) the numerator reads as
\begin{align*}
f(x_k,\underline{y}_k\vert \theta_k)=&f(y_k\vert x_k,\underline{y}_{k-1},\theta_k)f(x_k,\underline{y}_{k-1}\vert\theta_k)\\
=&f(y_k\vert x_k,\theta_k) f(x_k\vert \underline{y}_{k-1},\theta_k) f(\underline{y}_{k-1}\vert \theta_k).
\end{align*}
Although the likelihood function $f(y_k\vert x_k,\theta_k)$ is known from the description of the state-space, which for linear model (1) $f(y_k\vert x_k,\theta_k)=N\left(y_k\big\vert H_k x_k, R_k\right)$, the a priori distribution $f(x_k \vert \underline{y}_{k-1}\vert\theta_k)$ is not available explicitly. This posts difficulties in using the Bayesian approach, see for e.g. \cite{Trianta}. To deal with this issue, \cite{Rauch} considered the following a priori distribution,
\begin{equation}\label{eq:density}
\begin{split}
%y_k\vert x_k,\theta_k \sim& N\left(H_k x_k, R_k\right),\\
f(x_k\vert \underline{y}_{k-1},\theta_k)= N\left(x_k\big\vert \widehat{x}_{k\vert k-1}, P_{k\vert k-1}\right).
\end{split}
\end{equation}

Based on this consideration, \cite{Rauch} showed that the ML estimator $\widehat{x}_k$ coincides with the Kalman filter $\widehat{x}_{k\vert k}$ (\ref{eq:kalman}). Below we list some further concerns on (\ref{eq:mle}).

\begin{remark}\label{rem:rem1}
Some observations on estimation problem (\ref{eq:mle}). 
\begin{enumerate}
\item[(i)] From Bayes' formula (\ref{eq:bayes}) we see that the likelihood function $f(x_k,\underline{y}_k\vert \theta_k)$ of $x_k$ and $\underline{y}_k$ does not give complete information of linear system (1) in terms of the distributions of $\{v_k\}$ and $\{w_k\}$. That is, knowing $x_k$, $\underline{y}_k$ and control variables $\{u_k\}$ only yields $w_k$. But, it does not produce any information on $\underline{w}_{k-1}=(w_{\ell}: \ell\leq k-1)$ nor on $\underline{v}_k$, even though $x_k$ and $\underline{y}_{k}$ are provided. 

\item[(ii)] Thus, the a priori distribution $f(x_k \vert \underline{y}_{k-1}\vert\theta_k)$ is not available explicitly. For the linear state-space model (1), the incomplete observations of $\underline{w}_{k-1}$ and $\underline{v}_k$ are compensated by the use of normal distribution (\ref{eq:density}).

\item[(iii)] However, the normal distribution (\ref{eq:density}) was not verified in \cite{Rauch} and may not be necessary.

\item[(iv)] For general (nonlinear) state-space models, the score function $\frac{\partial \log f(x_k\vert \underline{y}_k,\theta_k)}{\partial x_k}$ might be difficult to evaluate.

\item[(v)] Hence, solving equations (\ref{eq:mle}) might as well be difficult. 

\item[(vi)] Deriving an explicit form of the covariance matrix of estimation error $x_k-\widehat{x}_k$ directly from the Bayes formula (\ref{eq:bayes}) is therefore a difficult problem and was not discussed explicitly in \cite{Rauch}. The problem has remained largely unexamined in literature. 
\end{enumerate}
\end{remark}

This paper attempts to generalize the problem (\ref{eq:mle}) as a maximum likelihood estimation of state-vector $x_k$ by considering $(x_k,\underline{y}_k)$ as an incomplete information of a state-space system and intends to resolve the points (i)-(vi).

\section{Main results and contributions}\label{sec:sec3}
In order to reformulate the problem (\ref{eq:mle}) as ML estimation from incomplete data, we consider the vector $(x_k,\underline{y}_k)$ as an incomplete information of a considered stochastic systems. The section below discusses the idea in greater details.

\subsection{Likelihood function of incomplete data}
%To deal with the points (i) and (ii), consider $(x_k,\underline{y}_k)$ as an incomplete observations of a given stochastic system. 
For this purpose, let $z_k=(x_k,\underline{y}_k,\underline{x}_{k-1})$ represent complete observations of state-vector $x_{\ell}$ and their measurements $y_{\ell}$, $\ell=1,\ldots,k$, up to time $k$, with $\underline{x}_k=\{x_0,x_1,\ldots, x_k\}$. In the maximum likelihood approach of \cite{Rauch}, the vector $z_k$ is only partially observed through incomplete data $\xi_k\equiv(x_k,\underline{y}_k)$. As before, we assume that $\underline{u}_k$ forms a sequence of predetermined control variables. Let for a fixed $k\geq 1$, $\mathcal{O}_k$ and $\mathcal{X}_{k-1}$ denote respectively the set of all possible values of the vectors $\xi_k$ and $\underline{x}_{k-1}$. Suppose that the random vector $z_k$ is defined on a measurable space $(\mathcal{Z}_k,\lambda)$, with $\mathcal{Z}_k=\mathcal{O}_k\times \mathcal{X}_{k-1}$, on which $z_k$ takes values, whereas $\lambda$ is a $\sigma-$finite measure on $\mathcal{Z}_k$. Let $T:\mathcal{Z}_k \rightarrow \mathcal{O}_k$ be many-to-one mapping from $\mathcal{Z}_k$ to $\mathcal{O}_k$. Denote by $f(z_k\vert \theta_k)$ the probability density (the Radon-Nikodym derivative w.r.t $\lambda$) of $z_k$. Define a subset $\mathcal{Z}(\xi_k)=\{z_k\in\mathcal{Z}_k: T(z_k)=\xi_k\}$. Following \cite{Halmos}, the marginal distribution $f(\xi_k\vert \theta_k)$ of the random vector $\xi_k$ is given by
\begin{align}\label{eq:halmos}
f(\xi_k\vert \theta_k)=\int_{\mathcal{Z}(\xi_k)}f(z_k\vert \theta_k) d \lambda(z_k).
\end{align}

Without loss of generality, we assume throughout the remaining of this paper that the probability density function $f(z_k\vert\theta_k)$ is twice continuously differentiable with respect to $x_k$ satisfying the integrability condition for $n\in\{1,2\}$,
\begin{align}\label{eq:ass}
\int_{\mathcal{Z}(\xi_k)} \Big\vert \frac{\partial^n \log f(z_k\vert \theta_k)}{\partial x_k^n}\Big\vert f(z_k\vert \xi_k,\theta_k) d\lambda (z_k) <\infty. \tag{A1}
\end{align}
The above integrability condition ensures the existence of conditional expectation $\mathbb{E}\Big[ \Big\vert \frac{\partial^n \log f(z_k\vert \theta_k)}{\partial x_k^n}\Big\vert \Big\vert \xi_k,\theta_k\Big]<\infty$.

In general, it is difficult to employ (\ref{eq:halmos}) in (\ref{eq:mle}) to derive the ML estimator $\widehat{x}_k$. This is mainly due to
\begin{align*}
\frac{\partial \log f(\xi_k\vert \theta_k)}{\partial x_k}=\frac{\frac{\partial }{\partial x_k} \int_{\mathcal{Z}(\xi_k)}f(z_k\vert \theta_k) d \lambda(z_k)}{\int_{\mathcal{Z}(\xi_k)}f(z_k\vert \theta_k) d \lambda(z_k)},
\end{align*}
and the difficulty in evaluating the integral (\ref{eq:halmos}), which might be attributed by the complexity of the likelihood function $f(z_k\vert \theta_k)$ of the complete data/information $z_k$, see for example (\ref{eq: complete}) for the linear system (1). It becomes a major concern if one wants to evaluate the observed information matrix $-\frac{\partial^2 \log f(\xi_k\vert \theta_k)}{\partial x_k \partial x_k^{\top}}$ of the incomplete data $\xi_k$. The difficulty of deriving ML estimator $\widehat{x}_k$ from (\ref{eq:mle}) is overcome by deriving distributional identities for the score function and the information matrix presented in the next section.

\subsection{Score functions and conditional observed information matrices of $(x_k,\underline{y}_k)$}

Distributional identities presented in this section play important roles in deriving ML estimator $\widehat{x}_k$ as the solution of (\ref{eq:mle}) and respective covariance matrix of estimation error. They are derived following similar approach used in \cite{Surya2022} for ML parameter estimation from incomplete data. 

\subsubsection{Score functions}

To start with, notice that the marginal distribution (\ref{eq:halmos}) of the vector $\xi_k$ can be equivalently rewritten as
\begin{align}\label{eq:int}
f(\xi_k\vert \theta_k)=\int_{\mathcal{X}_{k-1}} f(\xi_k,\underline{x}_{k-1}\vert \theta_k) d\lambda(\underline{x}_{k-1}),
\end{align}
which leads to the following identity for the conditional probability densities $f(z_k\vert \xi_k,\theta_k)$ and $f(\underline{x}_{k-1}\vert \xi_k,\theta_k)$. 
\begin{lemma}\label{eq:probidentity}
For given $\theta_k$ and observations $\xi_k$,
\begin{align}\label{eq:identity}
f(z_k\vert \xi_k,\theta_k)=\frac{f(z_k\vert \theta_k)}{f(\xi_k\vert \theta_k)}=f(\underline{x}_{k-1}\vert \xi_k,\theta_k).
\end{align}
\end{lemma}
\begin{proof}
It follows directly from the Bayes formula.
\end{proof}

Taking logarithm and derivative with respect to $x_k$ on both sides of (\ref{eq:identity}) results in deriving the relationship between the score functions of the data $\xi_k$ and $z_k$:
\begin{eqnarray}\label{eq:eq3}
\frac{\partial \log f(\xi_k\vert \theta_k)}{\partial x_k}=\frac{\partial \log f(z_k\vert \theta_k)}{\partial x_k} - \frac{\partial \log f(z_k\vert \xi_k,\theta_k)}{\partial x_k}.
\end{eqnarray}
\begin{proposition}\label{prop:prop1}
For given $\theta_k$ and observations $\xi_k$,
\begin{align}\label{eq:condscore}
\mathbb{E}\Big[\frac{\partial \log f(z_k\vert \xi_k,\theta_k)}{\partial x_k}\Big\vert \xi_k,\theta_k\Big]=0.
\end{align}
\end{proposition}
\begin{proof} Recall that $z_k=(\xi_k,\underline{x}_{k-1})$. By (\ref{eq:identity}), Bayes formula and dominated convergence theorem by condition (\ref{eq:ass}), 
\begin{align*}
&\mathbb{E}\Big[\frac{\partial \log f(z_k\vert \xi_k,\theta_k)}{\partial x_k}\Big\vert \xi_k,\theta_k\Big]\\
&\hspace{0.5cm}= \int_{\mathcal{Z}(\xi_k)} \frac{\partial \log f(z_k\vert \xi_k,\theta_k)}{\partial x_k}  f(z_k\vert \xi_k,\theta_k) d\lambda (z_k)\\
&\hspace{0.5cm}= \int_{\mathcal{X}_k} \frac{\partial \log f(z_k\vert \xi_k,\theta_k)}{\partial x_k}  f(\underline{x}_{k-1}\vert \xi_k,\theta_k) d\lambda (\underline{x}_{k-1})\\
&\hspace{0.5cm}= \int_{\mathcal{X}_k} \frac{\partial \log f(z_k\vert \xi_k,\theta_k)}{\partial x_k} f(z_k\vert \xi_k,\theta_k) d\lambda (\underline{x}_{k-1})\\
&\hspace{0.5cm}= \int_{\mathcal{X}_k} \frac{\partial f(z_k\vert \xi_k,\theta_k)}{\partial x_k}d\lambda (\underline{x}_{k-1}) \\
&\hspace{0.5cm}=\frac{\partial}{\partial x_k}\int_{\mathcal{X}_k} f(z_k\vert \xi_k,\theta_k) d\lambda (\underline{x}_{k-1}),
\end{align*}
which by (\ref{eq:int})-(\ref{eq:identity}) completes the proof. 
\end{proof}

The following identity gives an explicit form of the score function of $\xi_k$. It plays an important role in deriving ML estimator $\widehat{x}_k$ either explicitly for the linear state-space (1), or numerically using sequential Monte Carlo method for nonlinear state-space models. See Section \ref{sec:MCMC} for details.

\begin{theorem}\label{theo:main}
For any given $\theta_k$ and observations $\xi_k$,
\begin{align}\label{eq:mainidentity}
\mathbb{E}\Big[\frac{\partial \log f(z_k\vert \theta_k)}{\partial x_k}\Big\vert \xi_k,\theta_k\Big]=\frac{\partial \log f(\xi_k\vert \theta_k)}{\partial x_k}.
\end{align}
Furthermore, for any given state-vector $x_k$ and $\theta_k$,
\begin{align}\label{eq:mainidentityB}
\mathbb{E}\Big[\frac{\partial \log f(\xi_k\vert\theta_k)}{\partial x_k}\Big\vert x_k,\theta_k\Big]=\frac{\partial \log f(x_k\vert \theta_k)}{\partial x_k}.
\end{align}
\end{theorem}
\begin{proof} 
Taking conditional expectation $\mathbb{E}[\bullet \vert \xi_k,\theta_k]$ on both sides of the identity (\ref{eq:eq3}), we obtain
\begin{align*}
\frac{\partial \log f(\xi_k\vert \theta_k)}{\partial x_k}=&\mathbb{E}\Big[\frac{\partial \log f(z_k\vert \theta_k)}{\partial x_k}\Big\vert \xi_k,\theta_k\Big] \\
&- \mathbb{E}\Big[  \frac{\partial \log f(z_k\vert \xi_k,\theta_k)}{\partial x_k}     \Big\vert \xi_k,\theta_k\Big].
\end{align*}
The proof of (\ref{eq:mainidentity}) is complete using Proposition \ref{prop:prop1}. From the view point of incomplete information, the state-vector $x_k$ can be seen as an incomplete data of $\xi_k$. Hence, taking conditional expectation $\mathbb{E}[\bullet\vert x_k,\theta_k]$ in (\ref{eq:mainidentity}) yields  (\ref{eq:mainidentityB}).
\end{proof}

The identity (\ref{eq:mainidentity})-(\ref{eq:mainidentityB}) extend a similar form for ML parameter estimation from incomplete data, see e.g. \cite{Dempster}, \cite{McLachlan} and \cite{Little}, to the estimation of state-vector $x_k$. 
%Also, \cite{Surya2022}.

\subsubsection{Conditional observed information matrices}

The result below is derived as a corollary of (\ref{eq:condscore}).

\begin{corollary}\label{cor:cor1}
For given $\theta_k$ and observations $\xi_k$,
\begin{align*}
\mathrm{Cov}\left(\frac{\partial \log f(\xi_k\vert \theta_k)}{\partial x_k}, \frac{\partial \log f(z_k\vert \xi_k,\theta_k)}{\partial x_k}\Big\vert \xi_k,\theta_k\right)=\mathbf{0}.
\end{align*}
\end{corollary}
The above corollary yields some properties of the following conditional observed information matrices which will be used to derive ML estimator of $x_k$ and standard errors.
\begin{align*}
J_{\xi}(\xi_k\vert\theta_k)\equiv& - \frac{\partial^2 \log f(\xi_k\vert \theta_k)}{\partial x_k \partial x_k^{\top}},\\[5pt]
J_z(\xi_k\vert \theta_k)\equiv & \mathbb{E}\Big[- \frac{\partial^2 \log f(z_k\vert \theta_k)}{\partial x_k \partial x_k^{\top}} \Big\vert \xi_k,\theta_k\Big],\\[5pt]
J_{z\vert \xi}(\xi_k\vert \theta_k)\equiv & \mathbb{E}\Big[ -\frac{\partial^2 \log f(z_k\vert \xi_k,\theta_k)}{\partial x_k \partial x_k^{\top}} \Big\vert \xi_k,\theta_k\Big].
\end{align*}
It follows from (\ref{eq:eq3}) that the information matrices $J_{\xi}(\xi_k\vert\theta_k)$, $J_z(\xi_k\vert \theta_k)$ and $J_{z\vert \xi}(\xi_k\vert \theta_k)$ satisfy the equation
\begin{align}\label{eq:info}
J_{\xi}(\xi_k\vert\theta_k)=J_z(\xi_k\vert \theta_k) - J_{z\vert \xi}(\xi_k\vert \theta_k).
\end{align}

The result below shows that the information matrix $J_{z\vert \xi} (\xi_k\vert \theta_k)$ is positive definite. It will be used to show resulting information loss presented in incomplete data $\xi_k$.
\begin{proposition}\label{prop:theo2}
For any given $\theta_k$ and observations $\xi_k$,
\begin{align*}
J_{z\vert \xi}(\xi_k\vert \theta_k)=\mathbb{E}\Big[\frac{\partial \log f(z_k\vert \xi_k,\theta_k)}{\partial x_k}\frac{\partial \log f(z_k\vert \xi_k,\theta_k)}{\partial x_k^{\top}}\Big\vert \xi_k,\theta_k\Big].
\end{align*}
Hence, $J_{z\vert \xi}(\xi_k\vert \theta_k)$ is a positive definite information matrix.
\end{proposition}
\begin{proof}
To prove the identity, the following result is required. Using chain rule of derivatives, one can show that 
\begin{align*}
&-\frac{\partial^2\log f(z_k\vert \xi_k,\theta_k)}{\partial x_k \partial x_k^{\top}} f(z_k\vert \xi_k,\theta_k)\\
&\hspace{0.5cm}= -\frac{\partial }{\partial x_k}\Big[\frac{\partial \log f(z_k\vert \xi_k,\theta_k)}{\partial x_k^{\top}} f(z_k\vert \xi_k,\theta_k)\Big]\\
&\hspace{0.75cm}+\frac{\partial \log f(z_k\vert \xi_k,\theta_k)}{\partial x_k} \frac{\partial \log f(z_k\vert \xi_k,\theta_k)}{\partial x_k^{\top}} f(z_k\vert \xi_k,\theta_k).
\end{align*}
Integrating both sides over the set $\mathcal{Z}(\xi_k)$ with respect to measure $d\lambda (z_k)$ leads to a series of equalities:
\begin{align*}
J_{z\vert \xi}(\xi_k\vert \theta_k)=&-\int_{\mathcal{Z}(\xi_k)}\frac{\partial^2\log f(z_k\vert \xi_k,\theta_k)}{\partial x_k \partial x_k^{\top}} f(z_k\vert \xi_k,\theta_k) d\lambda (z_k)\\
&\hspace{-2.25cm}= -\int_{\mathcal{Z}(\xi_k)} \frac{\partial }{\partial x_k}\Big[\frac{\partial \log f(z_k\vert \xi_k,\theta_k)}{\partial x_k^{\top}} f(z_k\vert \xi_k,\theta_k)\Big] d\lambda (z_k)\\
&\hspace{-2.2cm}+ \int_{\mathcal{Z}(\xi_k)} \frac{\partial \log f(z_k\vert \xi_k,\theta_k)}{\partial x_k} \frac{\partial \log f(z_k\vert \xi_k,\theta_k)}{\partial x_k^{\top}} f(z_k\vert \xi_k,\theta_k) d\lambda(z_k)\\
&\hspace{-2.25cm}=  -\int_{\mathcal{X}_k} \frac{\partial }{\partial x_k}\Big[\frac{\partial \log f(z_k\vert \xi_k,\theta_k)}{\partial x_k^{\top}} f(\underline{x}_{k-1} \vert \xi_k,\theta_k)\Big] d\lambda (\underline{x}_{k-1})\\
&\hspace{-2.2cm}+ \int_{\mathcal{Z}(\xi_k)} \frac{\partial \log f(z_k\vert \xi_k,\theta_k)}{\partial x_k} \frac{\partial \log f(z_k\vert \xi_k,\theta_k)}{\partial x_k^{\top}} f(z_k\vert \xi_k,\theta_k) d\lambda(z_k)\\
&\hspace{-2.25cm}=  - \frac{\partial }{\partial x_k}\int_{\mathcal{X}_k} \frac{\partial \log f(z_k\vert \xi_k,\theta_k)}{\partial x_k^{\top}} f(\underline{x}_{k-1} \vert \xi_k,\theta_k) d\lambda (\underline{x}_{k-1})\\
&\hspace{-2.2cm}+ \int_{\mathcal{Z}(\xi_k)} \frac{\partial \log f(z_k\vert \xi_k,\theta_k)}{\partial x_k} \frac{\partial \log f(z_k\vert \xi_k,\theta_k)}{\partial x_k^{\top}} f(z_k\vert \xi_k,\theta_k) d\lambda(z_k)\\
&\hspace{-2.25cm}= - \frac{\partial }{\partial x_k}\mathbb{E}\Big[ \frac{\partial \log f(z_k\vert \xi_k,\theta_k)}{\partial x_k^{\top}}\Big\vert \xi_k,\theta_k\Big]\\
&\hspace{-2.2cm}+\mathbb{E}\Big[\frac{\partial \log f(z_k\vert \xi_k,\theta_k)}{\partial x_k}\frac{\partial \log f(z_k\vert \xi_k,\theta_k)}{\partial x_k^{\top}}\Big\vert \xi_k,\theta_k\Big],
\end{align*}
which by (\ref{eq:condscore}) completes the first claim. Positive definiteness of $J_{z\vert \xi}(\xi_k\vert \theta_k)$ follows directly from the first claim.
\end{proof}

The inequality below represents the resulting information loss presented in incomplete data $\xi_k$ (compared to $z_k$). 
\begin{theorem}[Resulting loss-of-information in incomplete data $\xi_k$] \label{theo:infoloss}
For any given observations $\xi_k$ and $\theta_k$,
\begin{align}\label{eq:infoloss}
J_{z}(\xi_k\vert \theta_k) > J_{\xi}(\xi_k\vert \theta_k).
\end{align}
In particular, at $\widehat{x}_k=\argmax\limits_{x_k} \log f(x_k,\underline{y}_k\vert\theta_k)$, it holds a.s.,
\begin{align}\label{eq:ordering}
J_z(\widehat{x}_k, \underline{y}_k \vert \theta_k) > J_{\xi}(\widehat{x}_k, \underline{y}_k\vert \theta_k)>0. 
\end{align}
\end{theorem}
\begin{proof}
The inequality (\ref{eq:infoloss}) is derived from Proposition \ref{prop:theo2} by replacing $J_{z\vert \xi}(\xi_k\vert\theta_k)$ with $J_z(\xi_k\vert \theta_k)- J_{\xi}(\xi_k\vert \theta_k)$, whilst (\ref{eq:ordering}) holds a.s. since $\widehat{x}_k$ is the maximizer of $\log f(x_k,\underline{y}_k\vert\theta_k)$.
\end{proof}

For parameter estimation, the corresponding inequality is established in \cite{Surya2022} in which the same inequality for expected information matrices, found in \cite{Orchard} (see also Theorem 2.86 in \cite{Schervish}), is generalized to observed information matrices. 

To get an explicit form of the observed information matrix $J_{\xi}(\xi_k\vert \theta_k)$, the following result is needed.

\begin{proposition}\label{prop:prop2}
For given $\theta_k$ and observations $\xi_k$,
\begin{equation}\label{eq:matrixJzxi}
\begin{split}
J_{z\vert \xi}(\xi_k\vert \theta_k)=&\mathbb{E}\Big[\frac{\partial \log f(z_k \vert \theta_k)}{\partial x_k} \frac{\partial \log f(z_k \vert \theta_k)}{\partial x_k^{\top}}    \Big\vert \xi_k,\theta_k \Big]\\
&\hspace{-1.75cm}-\mathbb{E}\Big[\frac{\partial \log f(z_k \vert \theta_k)}{\partial x_k} \Big\vert \xi_k,\theta_k \Big]\mathbb{E}\Big[\frac{\partial \log f(z_k \vert \theta_k)}{\partial x_k^{\top}} \Big\vert \xi_k,\theta_k \Big].
\end{split}
\end{equation}
\end{proposition}
\begin{proof}
Using (\ref{eq:eq3}), (\ref{eq:condscore}), Corollary \ref{cor:cor1} and the fact that $\mathrm{Cov}\left(\frac{\partial \log f(\xi_k\vert \theta_k)}{\partial x_k}, \frac{\partial \log f(\xi_k\vert \theta_k)}{\partial x_k}\Big\vert \xi_k,\theta_k\right)=0,$ we have
\begin{align*}
&\mathrm{Cov}\left(\frac{\partial \log f(z_k\vert \theta_k)}{\partial x_k}, \frac{\partial \log f(z_k\vert \theta_k)}{\partial x_k}\Big\vert \xi_k,\theta_k\right)\\
&\hspace{-0.5cm}=\mathrm{Cov}\left(\frac{\partial \log f(z_k\vert \xi_k,\theta_k)}{\partial x_k}, \frac{\partial \log f(z_k\vert \xi_k,\theta_k)}{\partial x_k}\Big\vert \xi_k,\theta_k\right)\\
&\hspace{-0.5cm}=\mathbb{E}\Big[\frac{\partial \log f(z_k\vert \xi_k,\theta_k)}{\partial x_k}\frac{\partial \log f(z_k\vert \xi_k,\theta_k)}{\partial x_k^{\top}}\Big\vert \xi_k,\theta_k\Big],
\end{align*}
which establishes the claim by Proposition \ref{prop:theo2}.
\end{proof}

The above result leads to an explicit form of $J_{\xi}(\xi_k\vert \theta_k)$.

\begin{theorem}
The information matrix $J_{\xi}(\xi_k\vert \theta_k)$ is given by
\begin{align}
J_{\xi}(\xi_k\vert \theta_k)=& \mathbb{E}\Big[- \frac{\partial^2 \log f(z_k\vert \theta_k)}{\partial x_k \partial x_k^{\top}} \Big\vert \xi_k,\theta_k\Big] \nonumber\\
&\hspace{-1cm}-\mathbb{E}\Big[\frac{\partial \log f(z_k \vert \theta_k)}{\partial x_k} \frac{\partial \log f(z_k \vert \theta_k)}{\partial x_k^{\top}}    \Big\vert \xi_k,\theta_k \Big]  \label{eq:fisher}\\
&\hspace{-1cm}+\mathbb{E}\Big[\frac{\partial \log f(z_k \vert \theta_k)}{\partial x_k} \Big\vert \xi_k,\theta_k \Big]\mathbb{E}\Big[\frac{\partial \log f(z_k \vert \theta_k)}{\partial x_k^{\top}} \Big\vert \xi_k,\theta_k \Big]. \nonumber
\end{align}
\end{theorem}
\begin{proof}
The proof follows from replacing the information matrix $J_{z\vert\xi}(\xi_k\vert \theta_k)$ by $J_z(\xi_k\vert\theta_k)- J_{\xi}(\xi_k\vert \theta_k)$ in (\ref{eq:matrixJzxi}). The expression (\ref{eq:fisher}) can also be obtained by taking taking derivative of (\ref{eq:mainidentity}). See Appendix A for details of derivation.
\end{proof}

Note that the observed information matrix $J_{\xi}(\xi_k\vert \theta_k)$ (\ref{eq:fisher}) extends the general matrix formula of \cite{Louis} for the calculation of covariance matrix of estimation error $\widehat{x}_k-x_k$.

\subsubsection{Recursive equation for the inverse of observed information matrix $J_{\xi}(\widehat{x}_k,\underline{y}_k\vert\theta_k)$}

The result below presents a recursive equation for calculating the inverse of observed information matrix $J_{\xi}(\widehat{x}_k,\underline{y}_k\vert\theta_k)$ under ML estimator $\widehat{x}_k$. The scheme may be used in the case $J_{\xi}(\widehat{x}_k,\underline{y}_k\vert\theta_k)$ is lack of sparsity and may be large due to the size of state-vector $x_k$. In such case $J_{\xi}(\widehat{x}_k,\underline{y}_k\vert\theta_k)$ is difficult to invert and may result in an incorrect inverse.
\begin{theorem}\label{theo:recOFI}
Let, for fixed $k$, $\{\Omega_k^{\ell}\}_{\ell\geq 0}$ be a sequence of $(p\times p)-$matrices with $\Omega_k^0=0$ satisfying the recursive equation 
\begin{align}\label{eq:recOFI}
\Omega_k^{\ell+1}=A(\widehat{x}_k)\Omega_k^{\ell} + B(\widehat{x}_k),
\end{align}
for $A(\widehat{x}_k),B(\widehat{x}_k)\in \mathbb{R}^{p\times p}$. Then, $\{\Omega_k^{\ell}\}_{\ell\geq 1}$ converges with root of convergence $\rho(A(\widehat{x}_k))$ to $\Omega_k=J_{\xi}^{-1}(\widehat{x}_k,\underline{y}_k\vert\theta_k)$ with
\begin{align*}
A(\widehat{x}_k)=& I - J_{z}^{-1}(\widehat{x}_k,\underline{y}_k\vert\theta_k)J_{\xi}(\widehat{x}_k,\underline{y}_k\vert\theta_k)\\
 B(\widehat{x}_k)=& J_{z}^{-1}(\widehat{x}_k,\underline{y}_k\vert\theta_k).
\end{align*}
Furthermore, the convergence is monotone in the sense 
\begin{align*}
\Omega_k^{\ell}<\Omega_k^{\ell+1}\leq \Omega_k, \quad \textrm{for}\;\; \ell=0,1,\ldots
\end{align*}
\end{theorem} 
\begin{proof}
The proof is based on applying inequality (\ref{eq:infoloss}) and established in similar approach to the proof of Theorem 4 in \cite{Surya2022}. See Appendix B for details of derivation.
\end{proof}

\subsection{Maximum likelihood estimation of state from incomplete data $(x_k,\underline{y}_k)$}
To formulate ML estimation, suppose that a true value $x_k^0$ of $x_k$ was sampled independently from $f(x_k\vert \theta_k)$ and is partially observed through noisy measurements $y_k$. Based on a series of observations $\underline{y}_k$, an estimator $\widehat{x}_k^0$ is found as the maximizer of the loglikelihood function $\log f(x_k, \underline{y}_k  \vert\theta_k)$, i.e., 
\begin{align*}
\widehat{x}_k^0=\argmax_{x_k}  \log f(x_k, \underline{y}_k  \vert\theta_k),
\end{align*}
or, equivalently, as the solution of system of equations 
\begin{align}\label{eq:mle2}
\frac{\partial \log f(\widehat{x}_k^0, \underline{y}_k  \vert\theta_k)}{\partial x_k}=0.
\end{align}
The theorem below shows that as the solution of (\ref{eq:mle2}), the estimator $\widehat{x}_k^0$ corresponds to the value of $x_k$ that gives the smallest distance between $\log f(x_k\vert \theta_k)$ and $\log f(x_k^0\vert\theta_k)$.
\begin{proposition}
As the solution of of (\ref{eq:mle2}), we have 
\begin{align}\label{eq:LR}
\widehat{x}_k^0=\argmin_{x_k} \left\vert \log f(x_k^0\vert\theta_k) -\log f(x_k \vert\theta_k)\right\vert^2.
\end{align}
\end{proposition}
As we can see, the nature of the estimation problem (\ref{eq:LR}) is quite different from the MMSE problem (\ref{eq:objfunc}). As opposed to (\ref{eq:objfunc}), the criterion (\ref{eq:LR}) gives an estimate $\widehat{x}_k^0$ of the unknown state $x_k^0$ with the most adherence of loglikelihood.   
\begin{proof}
Taking expectation $\mathbb{E}[\bullet\vert \widehat{x}_k^0,\theta_k]$ on both sides of (\ref{eq:mle2}), we deduce from criterion (\ref{eq:mle2}) and the identity (\ref{eq:mainidentityB}) that 
\begin{align}\label{eq:mle2B}
\mathbb{E}\Big[\frac{\partial \log f(\widehat{x}_k^0, \underline{y}_k  \vert\theta_k)}{\partial x_k}\Big\vert \widehat{x}_k^0,\theta_k]=\frac{\partial \log f(\widehat{x}_k^0\vert \theta_k)}{\partial x_k}=0,
\end{align}
leading to the equation
\begin{align*}
2\big[ \log f(x_k^0\vert\theta_k) -\log f(\widehat{x}_k^0 \vert\theta_k)\big]\frac{\partial \log f(\widehat{x}_k^0\vert \theta_k)}{\partial x_k}=0,
\end{align*}
from which the solution $\widehat{x}_k^0$ of (\ref{eq:mle2}) coincides with (\ref{eq:LR}).
\end{proof}

We further assume throughout the remaining of this paper that the state transition probability density function $f(x_k\vert x_{k-1}, \theta_k)$ has the following properties:
\begin{align}\label{eq:ass3}
f(x_k\vert x_{k-1},\theta_k)=0 \quad \textrm{and} \quad \frac{\partial f(x_k\vert x_{k-1},\theta_k)}{\partial x_k} =0, \tag{A2a}
\end{align}
for $x_k\in \partial \mathcal{X}_k$. Necessarily, for each $x_k\in \mathcal{X}_k$ and $n=0,1$,
\begin{align}\label{eq:ass3b}
\sup_{x_{k-1}}\Big\vert \frac{\partial^n f(x_k\vert x_{k-1},\theta_k)}{\partial x_k^n}\Big\vert <\infty, \tag{A2b}
\end{align}
where $\mathcal{X}_k$ denotes a set of all possible values of $x_k$. In fact, using Bayes formula, the above condition can as well be specified in terms of conditional distribution $f(x_k\vert x_{k-1},\theta_k)$. 
\begin{lemma}
For $x_k\in \partial \mathcal{X}_k$, the condition (A2) implies
\begin{align}\label{eq:ass3B}
f(x_k\vert \theta_k)=0 \quad \textrm{and} \quad \frac{\partial f(x_k\vert\theta_k)}{\partial x_k} =0. 
\end{align}
\end{lemma}
\begin{proof}
By (\ref{eq:ass3b}), the proof follows from Bayes formula, $f(x_k\vert \theta_k)=\int_{\mathcal{X}_k} f(x_k\vert x_{k-1},\theta_k) f(x_{k-1}\vert \theta_k) d\lambda(x_{k-1}).$
\end{proof}

\begin{example}
From the linear state-space (\ref{eq:eq1a}), we have $f(x_k\vert x_{k-1},\theta_k)=N(x_k\vert F_k x_{k-1}+ G_ku_k,Q_k)$. Whereas from (\ref{eq:solofx}), $x_k$ has probability density $f(x_k\vert\theta_k)$ of a multivariate normal distribution with mean $\mu_{x_k}:=\Phi_{k+1,1}\mu + \sum_{n=1}^k \Phi_{k+1,n+1} G_n u_n$ and covariance matrix $\Sigma_{x_k}:=\Phi_{k+1,1}P_0 \Phi_{k+1,1}^{\top} + \sum_{n=1}^k \Phi_{k+1,n+1}Q_n  \Phi_{k+1,n+1}^{\top}$.  Thus, the distribution of $x_k$ clearly satisfies (A2) and (\ref{eq:ass3B}).
\end{example}
\begin{proposition}\label{prop:propass2}
For any given $\theta_k$, (\ref{eq:ass3B}) implies that
\begin{align*}
&\hspace{2.5cm}\mathbb{E}\Big[\frac{\partial \log f(x_k^0\vert \theta_k)}{\partial x_k}\big\vert \theta_k\Big]=0,\\[4pt]
&\mathbb{E}\Big[-\frac{\partial^2 f(x_k^0\vert \theta_k)}{\partial x_k \partial x_k^{\top}}\big\vert \theta_k\Big]=\mathbb{E}\Big[\frac{\partial \log f(x_k^0\vert \theta_k)}{\partial x_k}\frac{\partial \log f(x_k^0\vert \theta_k)}{\partial x_k^{\top}}\big\vert \theta_k\Big].
\end{align*}
\end{proposition}
\begin{proof} See Appendix C for details of the proof.
\end{proof}
It is worth mentioning that for the score function and information matrix of parameter, we do not need to impose the condition (A2) to arrive at the corresponding results\footnote{$\mathbb{E}\big[\frac{\partial \log f(x_k^0\vert \theta_k)}{\partial \theta_k}\big\vert \theta_k\big]=0$ and $\mathbb{E}\big[-\frac{\partial^2 \log f(x_k^0\vert \theta_k)}{\partial \theta_k \partial \theta_k^{\top}}\big\vert \theta_k\big]=\mathbb{E}\big[\frac{\partial \log f(x_k^0\vert \theta_k)}{\partial \theta_k}\frac{\partial \log f(x_k^0\vert \theta_k)}{\partial \theta_k^{\top}}\big\vert \theta_k\big]$. See e.g., \cite{Schervish}.} for distribution parameter, see for e.g., \cite{Schervish}.

Applying a first-order Taylor expansion around the true value $x_k^0$ to the score function $\frac{\partial \log f(\widehat{x}_k^0, \underline{y}_k  \vert\theta_k)}{\partial x_k}$ and replacing the observed information matrix $J_{\xi}(x_k^0,\underline{y}_k\vert\theta_k]$ by expected information matrix $I_{\xi}(x_k^0,\underline{y}_k\vert\theta_k]=\mathbb{E}\big[J_{\xi}(x_k^0,\underline{y}_k\vert\theta_k]\vert\theta_k\big]$, similar to those employed in \cite{Freedman} for parameter estimation, the ML estimator $\widehat{x}_k^0$ (\ref{eq:mle2}) reads
\begin{align}\label{eq:euler}
\frac{\partial \log f(x_k^0,\underline{y}_k\vert \theta_k)}{\partial x_k}= I_{\xi}(x_k^0,\underline{y}_k\vert \theta_k) \big(\widehat{x}_k^0 -x_k^0\big).
\end{align}
Since $\widehat{x}_k^0$ is the maximizer of the loglikelihood function $\log f(x_k,\underline{y}_k\vert \theta_k)$, $\frac{\partial \log f(x_k^0,\underline{y}_k\vert \theta_k)}{\partial x_k}>0$ for $x_k^0<\widehat{x}_k^0$, and $\frac{\partial \log f(x_k^0,\underline{y}_k\vert \theta_k)}{\partial x_k}<0$ for $x_k^0>\widehat{x}_k^0$. We deduce from this observation and (\ref{eq:euler}) that $I_{\xi}(x_k^0,\underline{y}_k \vert \theta_k)>0$ at true state-vector $x_k^0$ and hence is invertible in view of  Theorem 7.2.1 on p.438 of \cite{Horn}. Thus, $\widehat{x}_k^0$ is given by
\begin{align}\label{eq:taylor}
\widehat{x}_k^0= x_k^0 + I_{\xi}^{-1}(x_k^0,\underline{y}_k\vert \theta_k) \mathcal{S}(x_k^0,\underline{y}_k\vert \theta_k),
\end{align}
where $\mathcal{S}(x_k^0,\underline{y}_l\vert \theta_k)\equiv \frac{\partial \log f(x_k^0,\underline{y}_k\vert \theta_k)}{\partial x_k}$ is evaluated using (\ref{eq:mainidentity}).

\subsection{Covariance matrix of estimation error and the Cram\'er-Rao lower bound}
This section discusses unbiased property of $\widehat{x}_k^0$, covariance matrix of estimation error $\widehat{x}_k^0-x_k^0$ and the corresponding Cram\'er-Rao lower bound for the covariance matrix.

\begin{theorem}\label{theo:unbiased}
Under (\ref{eq:ass3B}), $\widehat{x}_k^0$ is an unbiased estimator of $x_k^0$ with covariance matrix $P_k$ given by $I_{\xi}^{-1}(x_k^0,\underline{y}_k\vert \theta_k)$, i.e.,
\begin{align*}
\mathbb{E}[\widehat{x}_k^0 - x_k^0\big\vert \theta_k\big]=&0\\
\mathbb{E}[(\widehat{x}_k^0-x_k^0)(\widehat{x}_k^0-x_k^0)^{\top}\vert\theta_k]=&I_{\xi}^{-1}(x_k^0,\underline{y}_k\vert \theta_k).
\end{align*}
\end{theorem}
To prove the theorem, the following result is required.

\begin{proposition}\label{prop:fishermatrix}
Under (\ref{eq:ass3B}), it holds for any given $\theta_k$,
\begin{align}\label{eq:fishermatrix}
\mathbb{E}\big[\mathcal{S}(x_k^0,\underline{y}_k\vert \theta_k)\mathcal{S}^{\top}(x_k^0,\underline{y}_k\vert \theta_k)\big\vert \theta_k\big]=I_{\xi}(x_k^0,\underline{y}_k\vert \theta_k).
\end{align}
\end{proposition}
\begin{proof} See Appendix D for details of the proof.
\end{proof}
The identity (\ref{eq:fishermatrix}) justifies positive definiteness of the expected information matrix $I_{\xi}(x_k^0,\underline{y}_k\vert \theta_k)$ deduced from the above arguments on optimality of the estimator $\widehat{x}_k^0$. 

\medskip

\noindent \textit{Proof}[of Theorem \ref{theo:unbiased}]
The unbiased property of $\widehat{x}_k^0$ follows immediately from (\ref{eq:mainidentityB}), (\ref{eq:taylor}) and (\ref{eq:ass3B}). That is,
\begin{align*}
\mathbb{E}[\widehat{x}_k^0 -x_k^0\vert x_k^0,\theta_k]=&I_{\xi}^{-1}(x_k^0,\underline{y}_k\vert \theta_k)\mathbb{E}\big[\mathcal{S}(x_k^0,\underline{y}_k\vert \theta_k) \big\vert x_k^0,\theta_k\big]\\
=&I_{\xi}^{-1}(x_k^0,\underline{y}_k\vert \theta_k)\frac{\partial \log f(x_k^0\vert\theta_k)}{\partial x_k},
\end{align*}
which by Proposition \ref{prop:propass2} and tower property of expectation lead to the claim. Using (\ref{eq:taylor}) the covariance matrix is
\begin{align*}
&\mathbb{E}\big[(\widehat{x}_k^0-x_k^0)(\widehat{x}_k^0-x_k^0)^{\top}\big\vert \theta_k\big]=
 I_{\xi}^{-1}(x_k^0,\underline{y}_k\vert \theta_k)\\
 &\hspace{1cm}\times\mathbb{E}\big[\mathcal{S}(x_k^0,\underline{y}_k\vert\theta_k)\mathcal{S}^{\top}(x_k^0,\underline{y}_k\vert\theta_k)\vert\theta_k\big] I_{\xi}^{-1}(x_k^0,\underline{y}_k\vert \theta_k),
\end{align*}
which establishes the claim on account of (\ref{eq:fishermatrix}). $\square$

\begin{corollary}
Applying (\ref{eq:mainidentity}) to (\ref{eq:fisher}), (\ref{eq:fishermatrix}) results in
\begin{eqnarray}\label{eq:fishermatrix2}
\mathbb{E}\Big[\frac{\partial \log f(z_k \vert \theta_k)}{\partial x_k} \frac{\partial \log f(z_k \vert \theta_k)}{\partial x_k^{\top}}    \Big\vert \theta_k \Big]=\mathbb{E}\big[J_z(\xi_k\vert\theta_k)\big\vert\theta_k\big].
\end{eqnarray}
\end{corollary}
\begin{example}
See Section \ref{sec:application} for verification of the two matrix identities (\ref{eq:fishermatrix}) and (\ref{eq:fishermatrix2}) for linear state-space (1).
\end{example}

The theorem below gives a lower bound for the covariance matrix of estimation error. It extends the corresponding Cram\'er-Rao lower bound for covariance matrix of parameter estimates to that of for an estimator of state-vector.
\begin{theorem}\label{theo:cramer-rao}
Let $\widehat{x}_{k\vert k}^0$ be an estimator of $x_k^0$ under observations $\underline{y}_k$ with covariance matrix $P_{k\vert k}$. Define $\phi(x_k^0)=\mathbb{E}[\widehat{x}_{k\vert k}^0\vert x_k^0,\theta_k]$ and $b(x_k^0):=\phi(x_k^0)-x_k^0$. Under (\ref{eq:ass3B}),
\begin{equation}
\begin{split}
&\Big(\mathbb{E}\big[\frac{\partial \phi^{\top}(x_k^0)}{\partial x_k}\big\vert \theta_k\big] \Big)^{\top}I_{\xi}^{-1}(x_k^0,\underline{y}_k\vert \theta_k) \Big(\mathbb{E}\big[\frac{\partial \phi^{\top}(x_k^0)}{\partial x_k}\big\vert \theta_k\big] \Big)\\
&\hspace{2cm} + \mathbb{E}\big[b(x_k^0)b^{\top}(x_k^0)\vert\theta_k\big]\leq P_{k\vert k},
\end{split}
\end{equation}
with equality obtained for $\widehat{x}_{k\vert k}^0-\phi(x_k^0)=\mathcal{S}(x_k^0,\underline{y}_k\vert \theta_k)$ up to a multiplicative constant matrix (or Affine transform).
\end{theorem}
The theorem implies that for a MVU estimator $\widehat{x}_{k\vert k}^0:=\mathbb{E}[x_k^0\vert \underline{y}_k,\theta_k]$ for which case $\phi(x_k^0)=x_k^0$ and $\frac{\partial \phi^{\top}(x_k^0)}{\partial x_k}=I$, 
\begin{align}\label{eq:mvu}
P_{k\vert k} \geq P_k=I_{\xi}^{-1}(x_k^0,\underline{y}_k\vert \theta_k)>0.
\end{align}
Thus, the ML estimator $\widehat{x}_k^0$ has smaller standard error (covariance matrix) than that of MVU estimator $\widehat{x}_{k\vert k}^0$. Any MVU estimator $\widehat{x}_{k\vert k}^0$ with covariance matrix equal to $I_{\xi}^{-1}(x_k^0,\underline{y}_k\vert \theta_k)$ is called a \textit{fully efficient state estimator}.
\begin{proof}
Since the score function $\mathcal{S}(x_k^0,\underline{y}_k\vert\theta_k)$ gives complete description of the joint distribution of $x_k^0$ and $\underline{y}_k$ compared to $\widehat{x}_{k\vert k}^0-\phi(x_k^0)$, the proof follows from applying orthogonal projection of $\widehat{x}_{k\vert k}^0 -\phi(x_k^0)$ onto a space spanned by $\mathcal{S}(x_k^0,\underline{y}_k\vert\theta_k)$. On account of (\ref{eq:fishermatrix}), the matrix extension of the Cauchy-Schwarz inequality, see \cite{Tripathi}, yields
\begin{equation*}
\begin{split}
&\Big(\mathbb{E}\big[\big(\widehat{x}_{k\vert k}^0 -\phi(x_k^0)\big)\mathcal{S}^{\top}(x_k^0,\underline{y}_k\vert\theta_k)\vert\theta_k\big]\Big)I_{\xi}^{-1}(x_k^0,\underline{y}_k\vert \theta_k)\\
&\hspace{0.5cm}\times \Big(\mathbb{E}\big[\big(\widehat{x}_{k\vert k}^0 -\phi(x_k^0)\big)\mathcal{S}^{\top}(x_k^0,\underline{y}_k\vert\theta_k)\vert\theta_k\big]\Big)^{\top}\\
&\hspace{2cm} \leq \mathbb{E}\big[\big(\widehat{x}_{k\vert k}^0 -\phi(x_k^0)\big)\big(\widehat{x}_{k\vert k}^0 -\phi(x_k^0)\big)^{\top}\big\vert \theta_k\big].
\end{split}
\end{equation*}
It is straightforward to check that the equality is obtained for $\widehat{x}_{k\vert k}^0-\phi(x_k^0)=\mathcal{S}(x_k^0,\underline{y}_k\vert \theta_k)$ up to a multiplicative constant matrix (or Affine transform). Furthermore, 
\begin{align*}
&\mathbb{E}\big[\big(\widehat{x}_{k\vert k}^0 -\phi(x_k^0)\big)\big(\widehat{x}_{k\vert k}^0 -\phi(x_k^0)\big)^{\top}\big\vert \theta_k\big]\\
&\hspace{2cm}=P_{k\vert k} - \mathbb{E}\big[b(x_k^0)b^{\top}(x_k^0)\vert\theta_k\big].
\end{align*}
The proof is complete once we show that
\begin{align*}
\mathbb{E}\big[\big(\widehat{x}_{k\vert k}^0 -\phi(x_k^0)\big)\mathcal{S}^{\top}(x_k^0,\underline{y}_k\vert\theta_k)\vert\theta_k\big]=\mathbb{E}\Big[\frac{\partial \phi^{\top}(x_k^0)}{\partial x_k}\Big\vert \theta_k\Big],
\end{align*}
To establish this equality, we have by Bayes formula,
\begin{align*}
& \mathbb{E}\big[\widehat{x}_{k\vert k}^0 \mathcal{S}^{\top}(x_k^0,\underline{y}_k\vert\theta_k)\big\vert x_k^0, \theta_k\big]\\
&\hspace{0cm}=\int \widehat{x}_{k\vert k}^0 \frac{\partial \log f(x_k^0,\underline{y}_k\vert\theta_k)}{\partial x_k^{\top}} f(\underline{y}_k\vert x_k^0,\theta_k)d\lambda (\underline{y}_k)\\
&\hspace{0cm}=\frac{1}{f(x_k^0\vert\theta_k)} \int \widehat{x}_{k\vert k}^0 \frac{\partial f(x_k^0,\underline{y}_k\vert\theta_k)}{\partial x_k^{\top}}d\lambda (\underline{y}_k)\\
&\hspace{0cm}=\frac{1}{f(x_k^0\vert\theta_k)} \frac{\partial }{\partial x_k} \int f(x_k^0,\underline{y}_k\vert\theta_k) (\widehat{x}_{k\vert k}^0)^{\top} d\lambda (\underline{y}_k)\\
&\hspace{0cm}=\frac{1}{f(x_k^0\vert\theta_k)} \frac{\partial }{\partial x_k} \left( f(x_k^0\vert\theta_k) \int   f(\underline{y}_k\vert x_k^0,\theta_k)(\widehat{x}_{k\vert k}^0)^{\top}d\lambda (\underline{y}_k) \right)\\
&\hspace{0cm}=\frac{1}{f(x_k^0\vert\theta_k)} \frac{\partial }{\partial x_k} \left( f(x_k^0\vert\theta_k) \mathbb{E}\big[(\widehat{x}_{k\vert k}^0)^{\top} \big\vert x_k^0,\theta_k\big] \right)\\
&= \frac{\partial \log f(x_k^0\vert\theta_k) }{\partial x_k}\mathbb{E}\big[(\widehat{x}_{k\vert k}^0)^{\top} \big\vert x_k^0,\theta_k\big]+  \frac{\partial \mathbb{E}\big[(\widehat{x}_{k\vert k}^0)^{\top} \big\vert x_k^0,\theta_k\big]  }{\partial x_k} \\
&=\frac{\partial \log f(x_k^0\vert\theta_k) }{\partial x_k}\phi^{\top}(x_k^0)+ \frac{\partial \phi^{\top}(x_k^0) }{\partial x_k}.
\end{align*}
Furthermore, by (\ref{eq:ass3B}) and integration by part,
\begin{align*}
&\mathbb{E}\Big[\frac{\partial \log f(x_k^0\vert\theta_k) }{\partial x_k} \phi^{\top}(x_k^0) \Big\vert \theta_k\Big]\\
&=\int  \frac{\partial \log f(x_k^0\vert\theta_k) }{\partial x_k} \phi^{\top}(x_k^0) f(x_k^0\vert\theta_k) d x_k^0\\
&=\int  \frac{\partial f(x_k^0\vert\theta_k) }{\partial x_k} \phi^{\top}(x_k^0) d x_k^0 = -\int f(x_k^0\vert\theta_k)  \frac{\partial \phi^{\top}(x_k^0)}{\partial x_k} d x_k^0\\
&=-\mathbb{E}\Big[\frac{\partial \phi^{\top}(x_k^0)}{\partial x_k}\Big\vert \theta_k\Big],
\end{align*}
which by iterated law of conditional expectation leads to $
\mathbb{E}\big[\widehat{x}_{k\vert k}^0 \mathcal{S}^{\top}(x_k^0,\underline{y}_k\vert\theta_k)\big\vert \theta_k\big]=0$. Following the above, we have 
\begin{align*}
\mathbb{E}\big[\phi(x_k^0)\big)\mathcal{S}^{\top}(x_k^0,\underline{y}_k\vert\theta_k)\vert\theta_k\big]=-\left(\mathbb{E}\Big[\frac{\partial \phi^{\top}(x_k^0)}{\partial x_k}\Big\vert \theta_k\Big]\right)^{\top},
\end{align*}
which in turn completes the proof of the theorem.
\end{proof}

\subsection{Recursive ML estimation}

Using observed information matrices $J_{\xi}(\xi_k\vert \theta_k)$ and $J_z(\xi_k\vert \theta_k)$, recursive equations for $\widehat{x}_k^0$ can be derived. 

\subsubsection{Newton-Raphson iterative estimator}
Since the ML estimator $\widehat{x}_k^0$ is expressed in terms of the unknown true state-vector $x_k^0$ (and expected information matrix $I_{\xi}(x_k^0,\underline{y}_k\vert\theta_k)$), (\ref{eq:taylor}) can not be used directly. Hence, it should be used iteratively by replacing true state $x_k^0$ with $\widehat{x}_k^{(\ell)}$, an estimate of $x_k$ obtained after $\ell-$iteration, while $\widehat{x}_k^0$ replaced by $\widehat{x}_k^{(\ell +1)}$. By replacing $I_{\xi}(\widehat{x}_k^{(\ell)},\underline{y}_k\vert\theta_k)$ with $J_{\xi}(\widehat{x}_k^{(\ell)},\underline{y}_k\vert\theta_k)$, we derive a recursive equation for $\widehat{x}_k^0$ as
\begin{align}\label{eq:algo1}
\widehat{x}_k^{(\ell +1)}= \widehat{x}_k^{(\ell)} + J_{\xi}^{-1}(\widehat{x}_k^{(\ell)},\underline{y}_k\vert \theta_k) \mathcal{S}(\widehat{x}_k^{(\ell)},\underline{y}_k\vert \theta_k).
\end{align}
This iteration is run until a stopping criterion below is met,
\begin{align}\label{eq:stoppingrule}
\Vert \widehat{x}_k^{(\ell +1)} -  \widehat{x}_k^{(\ell)} \Vert^2 <\epsilon,
\end{align}
for $\epsilon>0$. At the convergence, the limiting estimator $\widehat{x}_k^{\infty}$ corresponds to the ML estimate $\widehat{x}_k^0$ as $ \mathcal{S}(\widehat{x}_k^{\infty},\underline{y}_k\vert \theta_k)=0$. 

To run iteration (\ref{eq:algo1}), we choose an initial condition:
\begin{align}\label{eq:initial1b}
\widehat{x}_k^{(0)}=\mathbb{E}[x_k\vert \underline{y}_k,\theta_k].
\end{align}
For a general state-space, this condition can be set up using particle filtering, see Section \ref{sec:MCMC}. The choice of initial condition (\ref{eq:initial1b}) would result in $\widehat{x}_k^0$ being refinement of the MVU estimator $\widehat{x}_{k\vert k}^0$ with smaller estimation error (\ref{eq:mvu}). 

\subsubsection{EM-Gradient iterative estimator}

To speed up the convergence of Newton-Raphson iteration (\ref{eq:algo1}), we may replace the information matrix $J_{\xi}(\widehat{x}_k^{(\ell)},\underline{y}_k\vert \theta_k)$ in (\ref{eq:algo1}) by $J_z(\widehat{x}_k^{(\ell)},\underline{y}_k\vert \theta_k)$. The reason is as follows. If $J_{\xi}(\widehat{x}_k^{(\ell)},\underline{y}_k \vert \theta_k)>0$, then by the resulting loss of information in incomplete data, (\ref{eq:infoloss}) implies that $J_z(\widehat{x}_k^{(\ell)},\underline{y}_k\vert \theta_k)>0$ and is invertible with $J_{\xi}^{-1}(\widehat{x}_k^{(\ell)},\underline{y}_k \vert \theta_k)>J_z^{-1}(\widehat{x}_k^{(\ell)},\underline{y}_k\vert \theta_k)>0$, by Theorem 7.2.1 on p.438 of \cite{Horn}, also Lemma 1.3, p.4 in \cite{Chui}. Replacing $J_{\xi}^{-1}(\widehat{x}_k^{(\ell)},\underline{y}_k\vert \theta_k)$ by $J_z^{-1}(\widehat{x}_k^{\ell},\underline{y}_k\vert \theta_k)$ in (\ref{eq:algo1}) yields
\begin{align}\label{eq:algo2}
\widehat{x}_k^{(\ell +1)}= \widehat{x}_k^{(\ell)} +J_z^{-1}(\widehat{x}_k^{(\ell)}, \underline{y}_k\vert \theta_k) \mathcal{S}(\widehat{x}_k^{(\ell)},\underline{y}_k\vert \theta_k),
\end{align}
a faster convergence of $\{\widehat{x}_k^{(\ell)}\}_{\ell \geq 0}$ than that of (\ref{eq:algo1}). The update $\widehat{x}_k^{(\ell +1)}$ (\ref{eq:algo2}) corresponds to the maximizer of quadratic function $\mathcal{Q}\big(x_k\vert \widehat{x}_k^{(\ell)}\big):=\mathcal{Q}(x_k\vert \widehat{x}_k^{(\ell)},\underline{y}_k,\theta_k)$ defined by 
\begin{align*}
\mathcal{Q}(x_k\vert \widehat{x}_k^{(\ell)})=&\mathcal{Q}(\widehat{x}_k^{(\ell)}\vert \widehat{x}_k^{(\ell)}) + \mathcal{S}(\widehat{x}_k^{(\ell)},\underline{y}_k\vert\theta_k)\big(x_k-\widehat{x}_k^{(\ell)}\big)\\
&-\frac{1}{2}\big(x_k-\widehat{x}_k^{(\ell)}\big)^{\top} J_z(\widehat{x}_k^{(\ell)},\underline{y}_k\vert\theta_k)\big(x_k-\widehat{x}_k^{(\ell)}\big),
\end{align*}
such that the update $\widehat{x}_k^{(\ell +1)}$ (\ref{eq:algo2}) is read as
\begin{align}\label{eq:EMfunc}
\widehat{x}_k^{(\ell+1)}=\argmax_{x_k} \mathcal{Q}(x_k\vert \widehat{x}_k^{(\ell)}).
\end{align}

When the stopping criterion (\ref{eq:stoppingrule}) is reached at the converging estimate $\widehat{x}_k^{\infty}$, $\mathcal{S}(\widehat{x}_k^{\infty},\underline{y}_k\vert \theta_k)=0$. Thus, $\widehat{x}_k^{\infty}$ corresponds to the ML estimator $\widehat{x}_k^0$ (\ref{eq:mle2}). To run the recursive equation (\ref{eq:algo2}), we use the same initial condition as in (\ref{eq:initial1b}). The recursive equation (\ref{eq:algo2}) extends further the EM-Gradient algorithm for ML parameter estimation of \cite{Lange} to a recursive ML estimation $\widehat{x}_k$ of state-vector $x_k$.

The result below shows that (\ref{eq:algo2}) is locally equivalent to EM algorithm of \cite{Dempster} for state vector.
\begin{proposition}[\textbf{EM-algorithm}]
The recursive equation (\ref{eq:algo2}) is locally equivalent to the EM algorithm. Namely,
\begin{eqnarray}\label{eq:EM}
\widehat{x}_k^{(\ell+1)}=\argmax_{x_k} \mathbb{E}\big[\log f(x_k,\underline{y}_k,\underline{x}_{k-1}\vert \theta_k) \big\vert \widehat{x}_k^{(\ell)}, \underline{y}_k, \theta_k\big].
\end{eqnarray}
\end{proposition}
The notion of being locally equivalent between the two algorithms (\ref{eq:EM}) and (\ref{eq:algo2}) refers to (\ref{eq:algo2}) having quadratic convergence compared to linear convergence of (\ref{eq:EM}). See Section 3.9 of \cite{McLachlan} and \cite{Lange} for details of discussions (for parameter estimation). 

\begin{proof}
In view of (\ref{eq:EMfunc}), we can identify the objective function $\mathcal{Q}(x_k\vert \widehat{x}_k^{(\ell)})= \mathbb{E}\big[\log f(x_k,\underline{y}_k,\underline{x}_{k-1}\vert \theta_k) \big\vert \widehat{x}_k^{(\ell)}, \underline{y}_k, \theta_k\big]$. 
Since $\widehat{x}_k^{(\ell +1)}$ is the maximizer of EM criterion (\ref{eq:EM}), the proof follows from applying first-order Taylor expansion around current estimate $\widehat{x}_k^{(\ell)}$ to work out the score function
\begin{align*}
\frac{\partial \log f(\widehat{x}_k^{(\ell+1)},\underline{y}_k,\underline{x}_{k-1}\vert \theta_k)}{\partial x_k}=&\frac{\partial \log f(\widehat{x}_k^{(\ell)},\underline{y}_k,\underline{x}_{k-1}\vert \theta_k)}{\partial x_k} \\
&\hspace{-2.25cm}+ \frac{\partial^2 \log f(\widehat{x}_k^{(\ell)},\underline{y}_k,\underline{x}_{k-1}\vert \theta_k)}{\partial x_k \partial x_k^{\top}}(\widehat{x}_k^{(\ell+1)} - \widehat{x}_k^{(\ell)}).
\end{align*}
Taking conditional expectation $\mathbb{E}[\bullet \vert \widehat{x}_k^{(\ell)},\underline{y}_k,\theta_k]$ on both sides, the proof is complete on account that $\mathbb{E}\big[\frac{\partial \log f(\widehat{x}_k^{(\ell+1)},\underline{y}_k,\underline{x}_{k-1}\vert \theta_k)}{\partial x_k}\big\vert \widehat{x}_k^{(\ell)},\underline{y}_k,\theta_k\big]=0$ by (\ref{eq:EM}), and $\mathbb{E}\big[\frac{\partial \log f(\widehat{x}_k^{(\ell)},\underline{y}_k,\underline{x}_{k-1}\vert \theta_k)}{\partial x_k}\big\vert \widehat{x}_k^{(\ell)},\underline{y}_k,\theta_k\big]=\frac{\partial \log f(\widehat{x}_k^{(\ell)},\underline{y}_k\vert \theta_k)}{\partial x_k}$.
\end{proof}

Observe that the EM criterion (\ref{eq:EM}) is slightly different from that of used in a recent work of \cite{Ramadan}. As the maximizer of (\ref{eq:EM}), $\widehat{x}_k^{(\ell+1)}$ increases loglikelihood function of incomplete information $(x_k,\underline{y}_k)$.

\begin{theorem}
The EM update $\widehat{x}_k^{(\ell+1)}$ (\ref{eq:EM}) increases the incomplete-data loglikelihood function $\log f(\xi_k\vert \theta_k)$, 
\begin{align}\label{eq:increase}
\log f(\widehat{x}_k^{(\ell+1)}, \underline{y}_k\vert \theta_k) > \log f(\widehat{x}_k^{(\ell)}, \underline{y}_k\vert \theta_k).
\end{align}
\end{theorem}
\begin{proof}
To prove the claim, it is necessary to show that the function $x_k\rightarrow \mathcal{H}(x_k\vert \widehat{x}_k^{(\ell)},\underline{y}_k,\theta)$ defined by
\begin{eqnarray*}
&\mathcal{H}(x_k\vert \widehat{x}_k^{(\ell)},\underline{y}_k,\theta_k) \\
&\hspace{-0.5cm}\equiv \mathbb{E}\Big[\log f(x_k,\underline{y}_k,\underline{x}_{k-1}\vert x_k,\underline{y}_k,\theta_k) \Big\vert \widehat{x}_k^{(\ell)},\underline{y}_k,\theta_k\Big],
\end{eqnarray*}
is decreasing in $x_k$. Indeed, this is the case since
\begin{align*}
&\mathcal{H}\big(\widehat{x}_k^{(\ell+1)}\vert \widehat{x}_k^{(\ell)},\underline{y}_k,\theta_k\big)
-\mathcal{H}\big(\widehat{x}_k^{(\ell)}\vert \widehat{x}_k^{(\ell)},\underline{y}_k,\theta_k\big)\\
&\hspace{0cm}=\mathbb{E}\left[\log \left( \frac{f(\widehat{x}_k^{(\ell +1)},\underline{y}_k,\underline{x}_{k-1}\vert \widehat{x}_k^{(\ell+1)},\underline{y}_k,\theta)}{ f(\widehat{x}_k^{(\ell)},\underline{y}_k,\underline{x}_{k-1}\vert \widehat{x}_k^{(\ell)},\underline{y}_k,\theta) }\right) \Big\vert \widehat{x}_k^{(\ell)},\underline{y}_k,\theta_k\right]\\
&\hspace{0cm}<\log \mathbb{E}\left[ \frac{f(\widehat{x}_k^{(\ell +1)},\underline{y}_k,\underline{x}_{k-1}\vert \widehat{x}_k^{(\ell+1)},\underline{y}_k,\theta)}{ f(\widehat{x}_k^{(\ell)},\underline{y}_k,\underline{x}_{k-1}\vert \widehat{x}_k^{(\ell)},\underline{y}_k,\theta) } \Big\vert \widehat{x}_k^{(\ell)},\underline{y}_k,\theta_k\right]=0,
\end{align*}
where the inequality is due to applying Jensen's inequality, whilst the last equality corresponds to $\frac{f(\widehat{x}_k^{(\ell +1)},\underline{y}_k,\underline{x}_{k-1}\vert \widehat{x}_k^{(\ell+1)},\underline{y}_k,\theta)}{ f(\widehat{x}_k^{(\ell)},\underline{y}_k,\underline{x}_{k-1}\vert \widehat{x}_k^{(\ell)},\underline{y}_k,\theta) } $ being the likelihood ratio (Radon-Nikodym derivative) of changing underlying probability distribution with density function $ f(\widehat{x}_k^{(\ell)},\underline{y}_k,\underline{x}_{k-1}\vert \widehat{x}_k^{(\ell)},\underline{y}_k,\theta)$ to that of with $f(\widehat{x}_k^{(\ell +1)},\underline{y}_k,\underline{x}_{k-1}\vert \widehat{x}_k^{(\ell+1)},\underline{y}_k,\theta)$. Finally, the proof is complete after taking expectation $\mathbb{E}[\bullet \vert \widehat{x}_k^{(\ell)},\underline{y}_k,\theta_k]$ on both sides of (\ref{eq:eq3}) and taking account of (\ref{eq:EM}).
\end{proof}

\begin{remark}
To deal with invertibility of $J_z(\xi_k\vert \theta_k)$ in (\ref{eq:algo2}), we may consider replacing $J_z(\xi_k\vert \theta_k)$ by the matrix
\begin{align}\label{eq:matM}
M_z(\xi_k\vert \theta_k)\equiv \mathbb{E}\Big[\frac{\partial \log f(z_k \vert \theta_k)}{\partial x_k} \frac{\partial \log f(z_k \vert \theta_k)}{\partial x_k^{\top}}    \Big\vert \xi_k,\theta_k \Big].
\end{align}
By doing so, the recursive equation (\ref{eq:algo2}) becomes
\begin{align}\label{eq:algo2b}
\widehat{x}_k^{(\ell +1)}= \widehat{x}_k^{(\ell)} +M_z^{-1}(\widehat{x}_k^{(\ell)}, \underline{y}_k\vert \theta_k) \mathcal{S}(\widehat{x}_k^{(\ell)},\underline{y}_k\vert \theta_k).
\end{align}
When iteration (\ref{eq:algo2b}) converges to $\widehat{x}_k^{\infty}$, it corresponds to the ML estimator $\widehat{x}_k$ (\ref{eq:mle2}) since it satisfies $\mathcal{S}(\widehat{x}_k^{\infty},\underline{y}_k\vert \theta_k)=0$. 

Thus, iteration (\ref{eq:algo2b}) extends further the BHHH algorithm of \cite{BHHH} and EM-Gradient algorithm of \cite{Lange} for recursive estimation of state-vector $x_k$.
\end{remark}

\subsection{Recursive Monte Carlo ML estimation}\label{sec:MCMC}
To deal with nonlinear state-space in general, we impose the following assumption on the distributions of $x_k$ and $y_k$.
\begin{assumption}
For a given $x_k$, the observation random vector $y_k$ is independent of $x_{k-1}$ and $\underline{y}_{k-1}$. Furthermore, for a given $x_{k-1}$, $x_k$ is independent of $\underline{y}_{k-1}$. That is,
\begin{align}\label{eq:ass2}
y_k\vert x_k  \perp (x_{k-1},\underline{y}_{k-1}) \quad \mathrm{and}\quad x_k\vert x_{k-1} \perp \underline{y}_{k-1}.\tag{A3} 
\end{align}
\end{assumption}

\subsubsection{Bayes formula for $f(x_{k-1}\vert x_k,\underline{y}_k,\theta_k)$}

The following results simplify the conditional probability density and play important roles for numerical simulation and to solve estimation problem (\ref{eq:mle2}) for linear state-space.
\begin{proposition}
For given $x_k$ and $\theta_k$, under (\ref{eq:ass2}),
\begin{align}\label{eq:bayes2}
f(x_{k-1}\vert x_k,\underline{y}_k,\theta_k)=f(x_{k-1}\vert x_k,\underline{y}_{k-1},\theta_k),
\end{align}
and
\begin{equation}\label{eq:bayes3}
\begin{split}
&f(x_{k-1}\vert x_k,\underline{y}_{k-1},\theta_k)\\
&\hspace{1cm}=\frac{1}{\mathcal{C}}f(x_k\vert x_{k-1},\theta_k)f(x_{k-1}\vert \underline{y}_{k-1},\theta_k),
\end{split}
\end{equation}
where the constant $\mathcal{C}$ is defined by $f(x_k\vert \underline{y}_{k-1},\theta_k)$.
\end{proposition}
\begin{proof}
By applying the Bayes formula, we obtain,
\begin{align*}
&f(x_{k-1}\vert x_k,\underline{y}_k,\theta_k)=f(x_{k-1}\vert x_k,y_k,\underline{y}_{k-1},\theta_k)\\
&\hspace{0.5cm}=\frac{f(x_{k-1},y_k\vert x_k, \underline{y}_{k-1},\theta_k)}{f(y_k\vert x_k,\underline{y}_{k-1},\theta_k)}\\
&\hspace{0.5cm}=\frac{f(x_{k-1}\vert x_k,\underline{y}_{k-1},\theta_k)f(y_k\vert x_k, x_{k-1},\underline{y}_{k-1},\theta_k)}{f(y_k\vert x_k,\underline{y}_{k-1},\theta_k)},
\end{align*}
which by (\ref{eq:ass2}) leads to (\ref{eq:bayes2}). By similar arguments, 
\begin{align*}
&f(x_{k-1}\vert x_k,\underline{y}_{k-1},\theta_k)=\frac{f(x_{k-1},x_k\vert \underline{y}_{k-1},\theta_k)}{f(x_k\vert \underline{y}_{k-1},\theta_k)}\\
&\hspace{0.5cm}=\frac{f(x_{k-1}\vert \underline{y}_{k-1},\theta_k)f(x_k\vert x_{k-1},\underline{y}_{k-1},\theta_k)}{f(x_k\vert \underline{y}_{k-1},\theta_k)},
\end{align*}
which completes the proof of the two identities.
\end{proof}

It follows from identities (\ref{eq:bayes2}) and (\ref{eq:bayes3}) that since the normalized distribution $\mathcal{C}:=f(x_k\vert \underline{y}_{k-1},\theta_k)$ does not contain information on random state-vector $x_{k-1}$, generating $x_{k-1}$ therefore only requires the knowledge of $f(x_{k-1}\vert \underline{y}_{k-1},\theta_k)$. 
\begin{remark}
Recall that since the probability density function $f(x_{k-1}\vert \underline{y}_{k-1},\theta_k)$ does not depend on $x_k$, we do not need to sample $x_{k-1}$ repeatedly at each iteration step $\ell$. 
\end{remark}

Suppose that a sequence of random state vectors $x_{k}:=\{x_{k}^{n}: n=1,\ldots,N\}$, $N\in\mathbb{N}$, have been generated. Let $\Delta x_{k}:=\big(\max\{x_{k}^{n}\}-\min\{x_{k}^{n}\})/N$. The function $f(x_{k-1}\vert x_k,\underline{y}_{k-1},\theta_k)$ is evaluated using formula (\ref{eq:bayes3}). While the state transition distribution $f(x_k\vert x_{k-1},\theta_k)$ is specified, the probability $\alpha_k^n:=f(x_k^n\vert \underline{y}_k,\theta_k)\Delta x_{k}$ is computed as 
\begin{align*}
&\alpha_k^n=\lim_{\Delta x_k\rightarrow 0}\mathbb{P}\{x_k^n\leq x_k<x_k^n+\Delta x_k \vert y_k,\underline{y}_{k-1},\theta_k\}\\
&\hspace{-0.75cm}=\lim_{\Delta x_k,\Delta y\rightarrow 0}\frac{\mathbb{P}\{x_k^n\leq x_k<x_k^n+\Delta x_k, y_k\leq y<y_k+\Delta y \vert \underline{y}_{k-1},\theta_k\}}{\mathbb{P}\{y_k\leq y<y_k+\Delta y \vert \underline{y}_{k-1},\theta_k\}}\\
=& \lim_{\Delta x_k\rightarrow 0} \frac{f(y_k\vert x_k^n,\theta_k)\mathbb{P}\{x_k^n \leq x_k< x_k^n +\Delta x_k \vert \underline{y}_{k-1},\theta_k\}}{\sum_{n=1}^N f(y_k\vert x_k^n,\theta_k) \mathbb{P}\{x_k^n \leq x_k< x_k^n +\Delta x_k \vert \underline{y}_{k-1},\theta_k\}}.
\end{align*}
Since $\lim_{\Delta x_k\rightarrow 0}\mathbb{P}\{x_k^n \leq x_k< x_k^n +\Delta x_k \vert \underline{y}_{k-1},\theta_k\}$ corresponds to the probability of randomly selecting $x_k=x_k^n$ out of $N$ possible realizations $\{x_{k}^{n}: n=1,\ldots,N\}$, we set $\lim_{\Delta x_k\rightarrow 0}\mathbb{P}\{x_k^n \leq x_k< x_k^n +\Delta x_k \vert \underline{y}_{k-1},\theta_k\}=1/N$. Hence, 
\begin{align}\label{eq:alphakn}
\alpha_k^n= \frac{f(y_k\vert x_k^n,\theta_k)}{\sum_{n=1}^N f(y_k\vert x_k^n,\theta_k) }.
\end{align}
The probabilities $\{\alpha_k^n: n=1,\ldots,N\}$ will be used to generate the state-vector $\{x_k^n:n=1,\ldots,N\}$, see Section \ref{sec:kitagawa}. 

In particle filtering literature, see e.g., \cite{Doucet,Doucet2000} and \cite{Kitagawa,Kitagawa2021}, the posterior distribution function $f(x_k\vert \underline{y}_k,\theta_k)$ is usually presented as
\begin{align}\label{eq:dirac}
 f(x_k\vert \underline{y}_k,\theta_k)=\sum_{n=1}^N \alpha_k^n \delta_{x_{k}^n}(x_k),
\end{align}
where $\delta(x)$ is the Dirac delta function. By this construction,
\begin{align}\label{eq:pdf}
f(x_k\vert\underline{y}_k,\theta_k)=f(y_k\vert x_k,\theta_k)\sum_{n=1}^N \alpha_{k-1}^n f(x_k\vert x_{k-1}^n,\theta_k),
\end{align}
up to a positive multiplicative constant $1/f(y_k\vert \underline{y}_{k-1},\theta_k)$. 
\begin{remark}
We see that attempting to solve (\ref{eq:mle}) using (\ref{eq:pdf}) would complicate the estimation problem for state-vector.  
\end{remark}

Subsequently, $f(x_{k-1}\vert x_k,\underline{y}_k,\theta_k)$ (\ref{eq:bayes2}) is evaluated as
\begin{eqnarray}\label{eq:measure}
f(x_{k-1}\vert x_k,\underline{y}_k,\theta_k)=\frac{1}{\mathcal{C}}\sum_{n=1}^N \alpha_{k-1}^n f(x_k\vert x_{k-1},\theta_k) \delta_{x_{k-1}^n}(x_{k-1}).
\end{eqnarray}
This distribution is used throughout the rest of this paper. 

\subsubsection*{Computing the constant $\mathcal{C}$}
In the case where the constant $\mathcal{C}:=f(x_k\vert \underline{y}_{k-1},\theta_k)$ (\ref{eq:bayes3}) needs to be computed, for instance in numerically evaluating the information matrix $J_{\xi}(x_k,\underline{y}_k\vert\theta_k)$, we have by (\ref{eq:dirac})
\begin{align}\label{eq:constant}
\mathcal{C}=&f(x_k\vert \underline{y}_{k-1},\theta_k) \nonumber\\
=&\int_{\mathcal{X}_k} f(x_k\vert x_{k-1},\theta_k)  f(x_{k-1}\vert \underline{y}_{k-1},\theta_k) d \lambda(x_{k-1})  \nonumber\\
=&\sum_{n=1}^N \alpha_{k-1}^n \int_{\mathcal{X}_k} f(x_k\vert x_{k-1},\theta_k) \delta_{x_{k-1}^n}(x_{k-1}) d \lambda(x_{k-1})   \nonumber\\
=&\sum_{n=1}^N \alpha_{k-1}^n  f(x_k\vert x_{k-1}^n,\theta_k).
\end{align}
For application purposes, see Section \ref{sec:appl} for details, we denote by $w_{k-1}^n(x_k\vert \theta_k)$ the weight function defined as
\begin{align*}
w_{k-1}^n(x_k\vert\ x_{k-1}^n,\theta_k)= \frac{\alpha_{k-1}^n f(x_k\vert x_{k-1}^n,\theta_k)}{\sum_{n=1}^N \alpha_{k-1}^n f(x_k\vert x_{k-1}^n,\theta_k)}.
\end{align*}

\subsubsection{Sampling from $f(x_k\vert\underline{y}_k,\theta_k)$}\label{sec:kitagawa}
To sample state-vector $x_{k}$ from $f(x_{k}\vert \underline{y}_{k},\theta_k)$, we use the method of \cite{Kitagawa2021,Kitagawa,Kitagawa93}. The method can be used to derive $\widehat{x}_{k\vert k}=\mathbb{E}[x_k\vert \underline{y}_k,\theta_k]$, the MMSE estimator (\ref {eq:objfunc}) such as discussed in Section \ref{sec:mmse}. To this end, consider nonlinear state-space $x_k=F(x_{k-1},u_k,v_k)$, with $x_0\sim q(x)$, where $\{u_k\}$ is a given sequence of control variables, whilst $v_k$ is a measurement noise with known distribution $g(v)$. Assume that $x_k$ is partially observed through observation vector $y_k$ with known probability density function $f(y_k\vert x_k,\theta_k)$. Suppose that a random sample $\{x_{k-1}^n: n=1,\ldots,N\}$ have been drawn from $f(x_{k-1}\vert \underline{y}_{k-1},\theta_k)$. Independently of $\{x_{k-1}^n\}_{n\geq 1}$, a set of measurement noises $\{v_k\}_{k\geq 1}$ are generated from the distribution $g(v)$, i.e., 
\begin{align*}
x_{k-1}^n \sim f(x_{k-1}\vert \underline{y}_{k-1},\theta_k), \quad v_k^n\sim g(v).
\end{align*}
Having obtained $x_{k-1}^n$, $u_k$ and $v_k^n$, set
\begin{equation}\label{eq:kitagawa2}
\begin{split}
\gamma_k^n=&F(x_{k-1}^n,u_k,v_k^n),\\[4pt]
\alpha_k^n=& \frac{f(y_k\vert \gamma_k^n,\theta_k)}{\sum_{n=1}^N f(y_k\vert \gamma_k^n,\theta_k) }.
\end{split}
\end{equation}
The state vectors $\{x_k^1,\ldots, x_k^N\}$ is obtained by resampling of $\{\gamma_k^1,\ldots, \gamma_k^N\}$ with probabilities $\{\alpha_k^1,\ldots,\alpha_k^N\}$. That is,
\begin{align}\label{eq:kitagawa3}
x_k^n =
\begin{cases}
\gamma_k^1, &\textrm{with probability $\alpha_k^1$},\\
\;\vdots & \quad \quad \quad \quad\vdots\\
\gamma_k^N, &\textrm{with probability $\alpha_k^N$},
\end{cases}
\end{align}
which is a realization of $x_k^n$ from posterior distribution (\ref{eq:dirac}).

To apply (\ref{eq:measure}) for evaluating $\mathbb{E}[F_k(x_{k-1},u_k)\vert x_k,\underline{y}_k,\theta_k]$ for a given sequence of control variables $\{u_k\}_{k\geq 1}$, see Section \ref{sec:appl}, an algorithm for sampling state-vector $x_{k}^n$ from the distribution $f(x_k\vert \underline{y}_k,\theta_k)$ is required. It is described below

\medskip

\textbf{Particle filtering algorithm}
\begin{enumerate}
\item[(i)] \textbf{Step 0}. Generate an initial $p-$dimensional state-vector $x_0^n\sim q(x)$, $n\in\{1,\ldots,N\}$.

\item[(ii)] \textbf{Step k}. Suppose that at each step $k=1,\ldots,K$ we have obtained $u_k$, $x_{k-1}^n$ and $v_k^n\sim g(v)$, $n\in\{1,\ldots,N\}$.
\begin{enumerate}
\item Compute $\gamma_k^n$ and $\alpha_k^n$, $n\in\{1,\ldots,N\}$, using (\ref{eq:kitagawa2}).
\item Generate $x_k^n$, $n\in\{1,\ldots,N\}$, according to (\ref{eq:kitagawa3}).
\end{enumerate}
\end{enumerate}

\subsubsection{Monte Carlo evaluation of score function and observed information matrices}\label{sec:filterS}

To examine recursive equations (\ref{eq:algo1}) and (\ref{eq:algo2}) for nonlinear state-space models, a recursive Monte Carlo based ML estimation is proposed on account of the following remark.
\begin{remark}
Notice that since we are primarily interested in estimating $x_k$, it is in principle sufficient to consider $(x_k,\underline{y}_k, x_{k-1})$ as ''complete information'' by which the value of Gaussian noises $v_k$ and $w_k$ are identifiable, given respectively by $v_k=x_k-F_k x_{k-1} -G_ku_k$ and $w_k= y_k -H_k x_k$.
\end{remark}

Using the probability density function (\ref{eq:measure}), the score function $\mathcal{S}(x_k,\underline{y}_k\vert\theta_k)$ and observed information matrix $J_{\xi}(x_k,\underline{y}_k\vert\theta_k)$ can be evaluated numerically as follows.
\begin{align}
&\mathcal{S}(x_k,\underline{y}_k\vert\theta_k)=\mathbb{E}\Big[\frac{\partial \log f(x_k,\underline{y}_k,x_{k-1}\vert\theta_k)}{\partial x_k}\big\vert x_k,\underline{y}_k,\theta_k\Big] \nonumber\\
&=\int_{\mathcal{X}_k} f(x_{k-1}\vert x_k,\underline{y}_k,\theta_k)\frac{\partial \log f(x_k,\underline{y}_k,x_{k-1}\vert\theta_k)}{\partial x_k} d\lambda(x_{k-1})  \nonumber\\
&=\frac{1}{\mathcal{C}}\sum_{n=1}^N \alpha_{k-1}^n \int_{\mathcal{X}_k}  f(x_k\vert x_{k-1},\theta_k)\delta_{x_{k-1}^n}(x_{k-1})  \nonumber\\
&\hspace{2.5cm}\times  \frac{\partial \log f(x_k,\underline{y}_k,x_{k-1}\vert\theta_k)}{\partial x_k} d\lambda(x_{k-1})  \nonumber\\
&=\sum_{n=1}^N \frac{\alpha_{k-1}^n   f(x_k\vert x_{k-1}^n,\theta_k)}{\mathcal{C}}\frac{\partial \log f(x_k,\underline{y}_k,x_{k-1}^n\vert\theta_k)}{\partial x_k} \nonumber
\end{align}

Thus, the score function $\mathcal{S}(x_k,\underline{y}_k\vert\theta_k)$ is determined by
\begin{equation}
\begin{split}
&\mathcal{S}(x_k,\underline{y}_k\vert\theta_k)
=\sum_{n=1}^N w_{k-1}^n(x_k\vert x_{k-1}^n,\theta_k) \\
&\hspace{3cm}\times \frac{\partial \log f(x_k,\underline{y}_k,x_{k-1}^n\vert\theta_k)}{\partial x_k}. \label{eq:filterS}
\end{split}
\end{equation}
Similarly defined, the matrix $J_{z}(x_k,\underline{y}_k\vert\theta_k)$ is evaluated as
\begin{align*}
&J_{z}(x_k,\underline{y}_k\vert\theta_k)= -\sum_{n=1}^N w_{k-1}^n(x_k\vert x_{k-1}^n,\theta_k)\\
&\hspace{3cm}\times \frac{\partial^2 \log f(x_k,\underline{y}_k,x_{k-1}^n\vert\theta_k)}{\partial x_k \partial x_k^{\top}},
\end{align*}
whilst the matrix $M_z(x_k,\underline{y}_k\vert\theta_k)$ (\ref{eq:matM}) is represented by
\begin{align*}
&M_z(x_k,\underline{y}_k\vert\theta_k)=
 \sum_{n=1}^N w_{k-1}^n(x_k\vert x_{k-1}^n,\theta_k) \\
 &\hspace{1cm}\times \frac{\partial \log f(x_k,\underline{y}_k,x_{k-1}^n\vert\theta_k)}{\partial x_k} \frac{\partial \log f(x_k,\underline{y}_k,x_{k-1}^n\vert\theta_k)}{\partial x_k^{\top}}.
 \end{align*}
Hence, the information matrix $J_{\xi}(x_k,\underline{y}_k\vert\theta_k)$ is given by
\begin{align*}
J_{\xi}(x_k,\underline{y}_k\vert\theta_k)=&J_{z}(x_k,\underline{y}_k\vert\theta_k) - M_z(x_k,\underline{y}_k\vert\theta_k)\\
&+ \mathcal{S}(x_k,\underline{y}_k\vert\theta_k) \mathcal{S}^{\top}(x_k,\underline{y}_k\vert\theta_k).
\end{align*}

\begin{remark}
Notice that since we are concerned with taking derivative with respect to $x_k$ in the score function and information matrices, it is sufficient to consider $ f(x_k,\underline{y}_k,x_{k-1}^n\vert\theta_k) \sim f(x_k\vert x_{k-1}^n,\theta_k) f(y_k\vert x_k,\theta_k)$.
\end{remark}

\subsubsection{EM-Gradient-Particle filter}\label{sec:mcmc1}
The method is based on employing the particle filtering algorithm for evaluation of  the score function $S(\xi_k\vert\theta_k)$ and information matrices $J_{\xi}(\xi_k\vert \theta_k)$ and $J_z(\xi_k\vert\theta_k)$ as described above. Thus, the recursive equation (\ref{eq:algo2}) can be written as 
\begin{align}\label{eq:algo3}
\widehat{x}_k^{(\ell +1)}=&\widehat{x}_k^{(\ell)} + J_z^{-1}(\widehat{x}_k^{(\ell)},\underline{y}_k\vert\theta_k)\mathcal{S}(\widehat{x}_k^{(\ell)},\underline{y}_k\vert\theta_k).
\end{align}
Run iteration using initial condition (\ref{eq:initial1b}) until the stopping rule (\ref{eq:stoppingrule}) is met and the sequence $\{\widehat{x}_k^{(\ell)}\}_{\ell}$ converges to $\widehat{x}_k^{\infty}$. 

The recursive equations (\ref{eq:algo1}) and (\ref{eq:algo2b}) can be reformulated in similar way as (\ref{eq:algo3}), although the derived equation for (\ref{eq:algo1}) might be slightly complex in view of formula (\ref{eq:fisher}).

\subsubsection{Covariance matrix of estimation error}\label{sec:mcmc2}
To derive an estimate of the covariance matrix of estimation error $\widehat{x}_k^{\infty}-x_k$, we use repeated sampling method. To be more details, draw $M$ random samples of size $N$ each of $x_{k-1}$ using the Monte Carlo method described above. To each subsample $\{x_{k-1,m}^{n}: n=1,\ldots,N\}$, $m\in \{1,\ldots,M\}$, apply recursive equation (\ref{eq:algo3}) until convergence to $\widehat{x}_{k,m}^{\infty}$ at which $\mathcal{S}(\widehat{x}_{k,m}^{\infty},\underline{y}_k,\theta_k)= \sum_{n=1}^N  w_{k-1}^n (\widehat{x}_{k,m}^{\infty}\vert x_{k-1,m}^n,\theta_k) \frac{\partial f(\widehat{x}_{k,m}^{\infty},\underline{y}_k,x_{k-1,m}^n\vert \theta_k)}{\partial x_k}$ is (close to) zero. Then, for each $\widehat{x}_{k,m}^{\infty}$ evaluate the matrices
\begin{align*}
&M_z^{(m)}(\widehat{x}_{k,m}^{\infty},\underline{y}_k\vert\theta_k)\\
%&:=\mathbb{E}\Big[\frac{\partial f(\widehat{x}_{k,m}^{\infty},\underline{y}_k,x_{k-1}\vert \theta_k)}{\partial x_k}\frac{\partial f(\widehat{x}_{k,m}^{\infty},\underline{y}_k,x_{k-1}\vert \theta_k)}{\partial x_k^{\top}}\Big\vert \widehat{x}_k^{\infty},\underline{y}_k,\theta_k\Big] \\
&\hspace{-0.5cm}=  \sum_{n=1}^N w_{k-1}^n (\widehat{x}_{k,m}^{\infty}\vert x_{k-1,m}^n,\theta_k)   \
\frac{\partial \log f(\widehat{x}_{k,m}^{\infty},\underline{y}_k,x_{k-1,m}^n\vert \theta_k)}{\partial x_k} \\
&\hspace{2cm}\times \frac{\partial \log f(\widehat{x}_{k,m}^{\infty},\underline{y}_k,x_{k-1,m}^n\vert \theta_k)}{\partial x_k^{\top}},\\[3pt]
&J_{z}^{(m)}(\widehat{x}_{k,m}^{\infty},\underline{y}_k\vert \theta_k)\\
%:=\mathbb{E}\Big[-\frac{\partial^2 f(\hathat{x}_k^{\infty},\underline{y}_k, x_{k-1}\vert,\theta_k)}{\partial x_k \partial x_k^{\top}} \Big\vert \hathat{x}_k^{\infty},\underline{y}_k,\theta_k\Big]\\
&\hspace{-0.5cm}= - \sum_{n=1}^N   w_{k-1}^n (\widehat{x}_{k,m}^{\infty}\vert x_{k-1,m}^n,\theta_k)  \
\frac{\partial^2 \log f(\widehat{x}_{k,m}^{\infty},\underline{y}_k,x_{k-1,m}^n\vert \theta_k)}{\partial x_k \partial x_k^{\top}},
\end{align*}
and the observed information matrix 
\begin{eqnarray*}
J_{\xi}^{(m)}(\widehat{x}_{k,m}^{\infty},\underline{y}_k \vert \theta_k)=J_{z}^{(m)}(\widehat{x}_{k,m}^{\infty},\underline{y}_k\vert \theta_k) -M_z^{(m)}(\widehat{x}_{k,m}^{\infty},\underline{y}_k \vert \theta_k).
\end{eqnarray*}
After all $M$ samples have been used, we have $M$ estimators $\{\widehat{x}_{k,m}^{\infty}: m=1,\ldots,M\}$ of $x_k^0$ and $M$ estimates $J_{\xi}^{(m)}(\widehat{x}_{k,m}^{\infty},\underline{y}_k\vert \theta_k)$ of the information matrix $J_{\xi}(x_k^0,\underline{y}_k\vert\theta_k)$. The estimators of $x_k^0$ and $I_{\xi}(x_k^0,\underline{y}_k\vert\theta_k)$ are defined by 
\begin{align}
\widehat{x}_k^0:=&\frac{1}{M}\sum_{m=1}^M \widehat{x}_{k,m}^{\infty}, \label{eq:mcmcmle}\\[5pt]
\widehat{I}_{\xi}(\widehat{x}_k^0,\underline{y}_k\vert\theta_k):=&\frac{1}{M}\sum_{m=1}^M J_{\xi}^{(m)}(\widehat{x}_{k,m}^{\infty},\underline{y}_k\vert \theta_k)>0, \nonumber
\end{align}
respectively. Thus, an estimate $\widehat{P}_k$ of the covariance matrix of $\widehat{x}_k^0-x_k^0$ is given by the inverse of $\widehat{I}_{\xi}(\widehat{x}_k^0,\underline{y}_k\vert\theta_k)$, i.e., 
\begin{eqnarray}\label{eq:estcovmat}
\widehat{P}_k:=\widehat{I}_{\xi}^{-1}(\widehat{x}_k^0,\underline{y}_k\vert\theta_k).
\end{eqnarray}

\section{Application of the main results}\label{sec:appl}

\subsection{Linear state-space models}\label{sec:application}

This section applies the main results of Section \ref{sec:sec3} to the linear state-space (1). In particular, we show that the ML estimator $\widehat{x}_k$ (\ref{eq:mle2}) corresponds to the Kalman filter $\widehat{x}_{k\vert k}$ (\ref{eq:kalman}) and that it is a fully efficient estimator of state vector $x_k$. 

 \subsubsection{Loglikelihood of complete data}

Following linear state-space (1), the observations $z_k$ gives a complete information of the systems in the sense of observing the sequence $\underline{v}_k$ and $\underline{w}_k$ of Gaussian random vectors as well as initial state $x_0$. By independence of $\underline{v}_k$, $\underline{w}_k$ and $x_0$, the likelihood function $f(z_k \vert \theta_k)$ of complete data $\{z_k\}$ is
\begin{eqnarray}
&&\hspace{-1cm}f(z_k \vert \theta_k)=\frac{1}{\sqrt{2\pi\vert P_0\vert}}  \exp\Big[-\frac{1}{2}(x_0-\mu)^{\top}P_0^{-1}(x_0-\mu)\Big] \nonumber \\ 
&&\hspace{-1cm}\times
 \prod_{\ell =1}^k \frac{1}{ 2\pi \vert Q_{\ell} \vert \vert R_{\ell} \vert }\exp\Big[-\frac{1}{2} \big( v_{\ell}^{\top} Q_{\ell}^{-1} v_{\ell} + w_{\ell}^{\top} R_{\ell}^{-1} w_{\ell}\big)\Big], \label{eq: complete}
\end{eqnarray}
where for each $k$, $v_k$ and $w_k$ are subject to the constraint described by the linear dynamics (1). Replacing $v_k=x_k - F_k x_{k-1} - G_k u_k$ and $w_k=y_k-H_k x_k$ in the likelihood function (\ref{eq: complete}) and taking logarithm, the loglikelihood function $\log f(z_k \vert \theta_k)$ of complete data $\{z_k\}$ becomes
\begin{equation*}
\begin{split}
\log f(z_k \vert \theta_k)=-&\frac{1}{2}\log\big(2 \pi \vert P_0\vert\big) -\frac{1}{2}(x_0-\mu)^{\top}P_0^{-1}(x_0-\mu)\\
&\hspace{-2.5cm} -\frac{1}{2} \sum_{\ell=1}^k\big(x_{\ell}-F_{\ell} x_{\ell-1}- G_{\ell} u_{\ell}\big)^{\top}Q_{\ell}^{-1} \big(x_{\ell}-F_{\ell} x_{\ell-1}- G_{\ell} u_{\ell}\big)\\
&\hspace{-2.5cm} -\frac{1}{2} \sum_{\ell=1}^k\big(y_{\ell}-H_{\ell} x_{\ell} \big)^{\top}R_{\ell}^{-1} \big(y_{\ell}-H_{\ell} x_{\ell}\big).
\end{split}
\end{equation*}

\subsubsection{Score functions and information matrices}

By applying basic rules of vector derivatives, see for e.g. \cite{Petersen}, the complete-data score function of the state-vector $x_k$ is specified explicitly by  
\begin{equation*}
\begin{split}
\frac{\partial \log f(z_k\vert \theta_k)}{\partial x_k}=&-H_k^{\top}R_k^{-1}\big(H_k x_k -y_k\big)\\
&-Q_k^{-1}\big(x_k-F_k x_{k-1} -G_k u_k\big),
\end{split}
\end{equation*}
with the complete-data information matrix given by
\begin{align*}
-\frac{\partial^2 \log f(z_k\vert \theta_k)}{\partial x_k \partial x_k^{\top}}=H_k^{\top} R_k^{-1} H_k + Q_k^{-1},
\end{align*}
Hence, the observed information matrix $J_z(x_k,\underline{y}_k\vert \theta_k)$ is 
\begin{align}\label{eq:infomat}
J_z(x_k,\underline{y}_k\vert \theta_k)=&\mathbb{E}\Big[-\frac{\partial^2 \log f(z_k\vert \theta_k)}{\partial x_k \partial x_k^{\top}}\Big\vert x_k,\underline{y}_k,\theta_k\Big] \nonumber\\
=&H_k^{\top} R_k^{-1} H_k + Q_k^{-1},
\end{align}
which is positive definite. Therefore, the expected information matrix is given by $I_z(x_k,\underline{y}_k\vert \theta_k)=H_k^{\top} R_k^{-1} H_k + Q_k^{-1}$. 

To derive incomplete-data score function $\mathcal{S}(\xi_k\vert\theta_k)$ and information matrix $J_{\xi}(\xi_k\vert\theta_k)$, the result below is needed.
\begin{lemma}\label{lem:smoothing}
For any given vector $(x_k,\underline{y}_k)$ and $\theta_k$,
\begin{align*}
\mathbb{E}[x_{k-1}\vert x_k, \underline{y}_{k},\theta_k]=&\widehat{x}_{k-1\vert k-1} + \Sigma_{k\vert k-1}P_{k\vert k-1}^{-1}(x_k-\widehat{x}_{k\vert k-1}),
\end{align*}
where $ \Sigma_{k\vert k-1} \equiv \mathbb{E}\big[(x_{k-1}-\widehat{x}_{k-1\vert k-1})(x_k - \widehat{x}_{k\vert k-1}) \vert \theta_k\big]$ with 
\begin{align*}
\Sigma_{k+1\vert k}=&P_{k\vert k-1}\big(F_{k+1} - F_{k+1}K_k  H_k\big)^{\top},\\
F_k\Sigma_{k\vert k-1}=& P_{k\vert k-1}-Q_k.
\end{align*}
\end{lemma}
This result is employed to replace the posterior distribution (\ref{eq:density}) and one-lag smoother $\widehat{x}_{k-1\vert k}=\mathbb{E}[x_{k-1}\vert \underline{y}_k,\theta_k]$, see e.g. \cite{Anderson}, in deriving the MLE $\widehat{x}_k$. 
\begin{proof}
See Appendix E for details of the proof.
\end{proof}
The incomplete-data score function $\mathcal{S}(x_k,\underline{y}_k\vert\theta_k)$ and observed information matrix $J_{\xi}(x_k,\underline{y}_k\vert\theta_k)$ are given below.
\begin{proposition}\label{prop:propscore}
$\mathcal{S}(x_k,\underline{y}_k\vert\theta_k)$ and $J_{\xi}(x_k,\underline{y}_k\vert\theta_k)$ are resp.,
\begin{align*}
\mathcal{S}(x_k,\underline{y}_k\vert\theta_k)=&-(H_k^{\top}R_k^{-1}H_k + P_{k\vert k-1}^{-1})(x_k -\widehat{x}_{k\vert k-1}) \\
&+ H_k^{\top}R_k^{-1}\big(y_k - H_k \widehat{x}_{k\vert k-1}\big),\\[4pt]
J_{\xi}(x_k,\underline{y}_k\vert\theta_k)=&(H_k^{\top}R_k^{-1}H_k + P_{k\vert k-1}^{-1}).
\end{align*}
\end{proposition}  
\begin{proof}
See Appendix F for details of the proof.
\end{proof}
From (\ref{eq:sol2}), it follows that $P_{k\vert k-1}> Q_k>0$, or equivalently $Q_k^{-1}>P_{k\vert k-1}^{-1}>0$ which by (\ref{eq:infomat}) results in $J_{z}(x_k,\underline{y}_k\vert\theta_k)>J_{\xi}(x_k,\underline{y}_k\vert\theta_k)$. The latter verifies the loss-of-information inequality (\ref{eq:infoloss}), and certainly the inequality $I_{z}(x_k,\underline{y}_k\vert\theta_k)>I_{\xi}(x_k,\underline{y}_k\vert\theta_k)>0.$ Furthermore, we can show that the two matrix identities (\ref{eq:fishermatrix}) and (\ref{eq:fishermatrix2}) hold.
\begin{lemma}\label{lem:fishermatrix2}
For a given system parameter $\theta_k$, 
\begin{align*}
\mathbb{E}\left[\mathcal{S}(x_k,\underline{y}_k\vert\theta_k)\mathcal{S}^{\top}(x_k,\underline{y}_k\vert\theta_k)\big\vert\theta_k\right]=&H_k^{\top}R_k^{-1}H_k + P_{k\vert k-1}^{-1},\\[4pt]
\mathbb{E}\left[\frac{\partial \log f(z_k\vert\theta_k)}{\partial x_k}\frac{\partial \log f(z_k\vert\theta_k)}{\partial x_k^{\top}}\big\vert\theta_k\right]=&H_k^{\top} R_k^{-1} H_k + Q_k^{-1}.
\end{align*}
\end{lemma}
\begin{proof}
See Appendix G for details of the proof.
\end{proof}

\subsubsection{Kalman filter as fully efficient ML estimator}
Setting the score function $\mathcal{S}(x_k,\underline{y}_k \vert \theta_k)$ to zero followed by premultiplying both sides by $\big(H_k^{\top} R_k^{-1}H_k + P_{k\vert k-1}^{-1}\big)^{-1}$ yields
\begin{align*}
\widehat{x}_k=\widehat{x}_{k\vert k-1} + \big(H_k^{\top} R_k^{-1}H_k + P_{k\vert k-1}^{-1}\big)^{-1}H_k^{\top}R_k^{-1}\overline{y}_k,
\end{align*}
which by the result below leads to the claim that $\widehat{x}_k=\widehat{x}_{k\vert k}$.

\begin{lemma}\label{lem:inverse}
If the matrix $H_k P_{k\vert k-1} H_k^{\top}+R_k$ is invertible,
\begin{align}\label{eq:woodbury}
&\big(H_k^{\top} R_k^{-1}H_k + P_{k\vert k-1}^{-1}\big)^{-1} \\
&\hspace{0cm}=
P_{k\vert k-1} - P_{k\vert k-1} H_k^{\top}\big(H_k P_{k\vert k-1} H_k^{\top} + R_k\big)^{-1} H_k P_{k\vert k-1} \nonumber,
\end{align}
from which the following identity holds
\begin{align*}
&\big(H_k^{\top} R_k^{-1}H_k + P_{k\vert k-1}^{-1}\big)^{-1}H_k^{\top} R_k^{-1} \\
&\hspace{1.5cm}= P_{k\vert k-1} H_k^{\top}\big(H_k P_{k\vert k-1} H_k^{\top} + R_k\big)^{-1} \nonumber,
\end{align*}
\end{lemma}
\begin{proof} See Appendix H for details of the proof.
\end{proof}

From the covariance matrix $P_{k\vert k}$ (\ref{eq:kalmancovmat}) of Kalman filter $\widehat{x}_{k\vert k}$, we notice from (\ref{eq:woodbury}) and Proposition \ref{prop:propscore} that $\widehat{x}_{k\vert k}$ is a fully efficient state estimator since $P_{k\vert k}=I_{\xi}^{-1}(x_l,\underline{y}_k\vert\theta_k)$. This conclusion is in line with the statement of Theorem \ref{theo:cramer-rao} as
\begin{align*}
\widehat{x}_{k\vert k}- x_k=\big(H_k^{\top} R_k^{-1} H_k + P_{k\vert k-1}^{-1}\big)^{-1}\mathcal{S}(x_k,\underline{y}_k\vert\theta_k);
\end{align*}
that is the estimation error $\widehat{x}_{k\vert k}- x_k$ is given by a matrix multiplicative constant of the score function $\mathcal{S}(x_k,\underline{y}_k\vert\theta_k)$.

\subsubsection{EM-Gradient-Particle filter}
Applying recursive equation (\ref{eq:algo3}) to linear state-space (1), 
\begin{align*}
&\widehat{x}_k^{(\ell+1)}=\widehat{x}_k^{(\ell)} + \left(H_k^{\top}R_k^{-1} H_k + Q_k^{-1}\right)^{-1}\\&\times \sum_{n=1}^N w_{k-1}^n(\widehat{x}_k^{(\ell)}\vert x_{k-1}^n,\theta_k)\Big[-\big(H_k^{\top}R_k^{-1} H_k + Q_k^{-1}\big)x_k^{(\ell)} \\
&+ H_k^{\top} R_k^{-1} y_k + Q_k^{-1}\big(F_k x_{k-1}^n + G_ku_k\big)  \Big] \\
=&  \left(H_k^{\top}R_k^{-1} H_k + Q_k^{-1}\right)^{-1}\Big[H_k^{\top}R_k^{-1} y_k \\
&+ Q_k^{-1}\Big(G_ku_k+ F_k \sum_{n=1}^n  x_{k-1}^ n w_{k-1}^n(\widehat{x}_k^{(\ell)}\vert x_{k-1}^n,\theta_k)\Big)\Big].
\end{align*}
Notice that $f(x_k\vert x_{k-1}^n,\theta_k)$ in $w_{k-1}^n(x_k\vert\theta_k)$ denotes the probability density function of multivariate normal with mean $F_k x_{k-1}^n + G_ku_k$ and covariance matrix $Q_k$. Using matrix inversion similar to (\ref{eq:woodbury}) simplifies the estimate $\widehat{x}_k^{(\ell+1)}$ to,
\begin{align}\label{eq:recmcmc}
&\widehat{x}_k^{(\ell+1)}=F_k \sum_{n=1}^N x_{k-1}^n w_{k-1}^n(\widehat{x}_k^{(\ell)}\vert x_{k-1}^n,\theta_k) \nonumber\\&\hspace{1cm}+ G_ku_k  +Q_k H_k^{\top} \big(H_kQ_k H_k^{\top} + R_k\big)^{-1}\\
\times &\Big[y_k - H_k \Big( F_k \sum_{n=1}^N x_{k-1}^nw_{k-1}^n(\widehat{x}_k^{(\ell)}\vert x_{k-1}^n,\theta_k) + G_ku_k\Big) \Big]. \nonumber
\end{align}
To summarize the application of main results to the linear state-space (1), the Kalman filter (\ref{eq:kalman}) takes slightly different form than the recursive EM-Gradient-Particle estimator (\ref{eq:recmcmc}), although at convergence the two estimators coincide. 

\subsection{Nonlinear state-space models}\label{sec:application2}

This section generalizes iterative scheme (\ref{eq:recmcmc}) for a nonlinear Gaussian state-space models. 

\subsubsection{EM-Gradient-Particle filter}

\begin{theorem}
Consider nonlinear state-space 
\begin{align}\label{eq:nls1}
x_k=F_k(x_{k-1},u_k) + v_k, 
\end{align}
where $v_k\sim N(0,Q_k)$, with linear measurement
\begin{align}\label{eq:nls2}
y_k=H_k x_k + w_k,
\end{align}
and $w_k\sim N(0,R_k)$. Based on this information, we have
\begin{align}\label{eq:recmcmc2}
\widehat{x}_k^{(\ell+1)}=&\sum_{n=1}^N w_{k-1}^n(\widehat{x}_k^{(\ell)}\vert x_{k-1}^n, \theta_k) F_k(x_{k-1}^n,u_k) \nonumber \\&\hspace{0.5cm}+ Q_k H_k^{\top}\left(H_k Q_k H_k^{\top} + R_k\right)^{-1}\\
&\times\Big[ y_k -H_k \sum_{n=1}^N w_{k-1}^n(\widehat{x}_k^{(\ell)}\vert\theta_k) F_k(x_{k-1}^n,u_k)\Big].\nonumber
\end{align}
\end{theorem}
\begin{remark}
Notice that the EM-Gradient-Particle filter (\ref{eq:recmcmc}) and (\ref{eq:recmcmc2}) respectively provide an improved EM-algorithm update $\widehat{x}_k^{(\ell+1)}$ for linear state-space (1) and the nonlinear models (\ref{eq:nls1})-(\ref{eq:nls2}) of those presented in eqns (35) and (14) of Theorem 3 in \cite{Ramadan}. 
\end{remark}

\begin{proof}
To derive the recursive estimator (\ref{eq:recmcmc2}), recall that
\begin{align*}
&\log f(x_k,\underline{y}_k,x_{k-1}\vert\theta_k)=-\log(2\pi \vert Q_k\vert) -\log(2\pi \vert R_k\vert)\\
&-\frac{1}{2}\big(x_k-F_k(x_{k-1},u_k)\big)^{\top}Q_k^{-1}\big(x_k-F_k(x_{k-1},u_k)\big)\\
&-\frac{1}{2}\big(y_k-H_kx_k\big)^{\top}R_k^{-1}\big(y_k-H_k x_k\big).
\end{align*}
Using basic rules of matrix derivatives, we obtain
\begin{align*}
\frac{\partial \log f(x_k,\underline{y}_k,x_{k-1}\vert\theta_k)}{\partial x_k}=&-(H_k^{\top}R_k^{-1}H_k+Q_k^{-1})x_k\\
&\hspace{-2cm}+Q_k^{-1}F_k(x_{k-1},u_k) + H_k^{\top}R_k^{-1} y_k,\\[4pt]
-\frac{\partial^2 \log f(x_k,\underline{y}_k,x_{k-1}\vert\theta_k)}{\partial x_k \partial x_k^{\top}}=&(H_k^{\top}R_k^{-1}H_k+Q_k^{-1}).
\end{align*}
Applying (\ref{eq:filterS}), the incomplete-data score function is 
\begin{align*}
&\mathcal{S}(x_k,\underline{y}_k\vert\theta_k)=-(H_k^{\top}R_k^{-1}H_k+Q_k^{-1})x_k+ \Big[H_k^{\top}R_k^{-1} y_k \\
&\hspace{1cm}+ Q_k^{-1}\sum_{n=1}^N w_{k-1}^n(x_k\vert x_{k-1}^n,\theta_k)F_k(x_{k-1}^n,u_k)\Big],
\end{align*}
where we have used the distribution (\ref{eq:measure}) to re-derive an explicit formula of $\mathcal{S}(x_k,\underline{y}_k\vert\theta_k)$ discussed in Section \ref{sec:mcmc1}. Similarly, the observed information matrix $J_z(x_k,\underline{y}_k\vert \theta_k)$ is 
\begin{align*}
J_z(x_k,\underline{y}_k\vert \theta_k) 
=H_k^{\top}R_k^{-1}H_k + Q_k^{-1}.
\end{align*}
Notice that, there is no need to use Monte Carlo particle filtering for the evaluation of information matrix $J_z(x_k,\underline{y}_k\vert \theta_k)$, which is in contrary to $J_{\xi}(x_k,\underline{y}_k\vert\theta_k)$ as we can see below. Thus, the recursive equation (\ref{eq:algo3}) is read as
\begin{align*}
&\widehat{x}_k^{(\ell+1)}=\widehat{x}_k^{(\ell)}+J_z^{-1}(\widehat{x}_k^{(\ell)},\underline{y}_k\vert \theta_k)\mathcal{S}(\widehat{x}_k^{(\ell)},\underline{y}_k\vert\theta_k)\\
=&\left(H_k^{\top}R_k^{-1}H_k + Q_k^{-1}\right)^{-1}\Big[H_k^{\top}R_k^{-1} y_k\\
&+Q_k^{-1}\sum_{n=1}^N w_{k-1}^n(\widehat{x}_k^{(\ell)}\vert x_{k-1}^n,\theta_k)F_k(x_{k-1}^n,u_k)\Big],
\end{align*}
which using matrix inversion formula leads to (\ref{eq:recmcmc2}).
\end{proof}

\subsubsection{Covariance matrix of estimation error}
Following the identity (\ref{eq:fisher}) and the expressions of Section \ref{sec:filterS}, the observed information matrix $J_{\xi}(x_k,\underline{y}_k\vert\theta_k)$ is
\begin{align*}
&J_{\xi}(x_k,\underline{y}_k\vert\theta_k)=(H_k^{\top} R_k^{-1}H_k + Q_k^{-1})\\
&\hspace{1cm}-\Big(\sum_{n=1}^N \frac{\partial w_{k-1}^n (x_k\vert x_{k-1}^n,\theta_k)}{\partial x_k}\Big) y_k^{\top} R_k^{-1}H_k\\
&\hspace{1cm}- \sum_{n=1}^N \frac{\partial w_{k-1}^n (x_k\vert x_{k-1}^n,\theta_k)}{\partial x_k} F_k^{\top}(x_{k-1}^n,u_k)Q_k^{-1},
\end{align*}
where partial derivative $ \frac{\partial w_{k-1}^n (x_k\vert x_{k-1}^n,\theta_k)}{\partial x_k}$ is given by
\begin{align*}
& \frac{\partial w_{k-1}^n (x_k\vert x_{k-1}^n,\theta_k)}{\partial x_k}\\
&\hspace{1cm}=w_{k-1}^n (x_k\vert x_{k-1}^n,\theta_k)\Big[\frac{\partial \log f(x_k\vert x_{k-1}^n,\theta_k)}{\partial x_k}\\
&\hspace{1.5cm}-\sum_{m=1}^N w_{k-1}^m (x_k\vert x_{k-1}^m,\theta_k)\frac{\partial \log f(x_k\vert x_{k-1}^m,\theta_k)}{\partial x_k}\Big].
 \end{align*}
Given that $\frac{\partial \log f(x_k\vert x_{k-1}^n,\theta_k)}{\partial x_k}=-Q_k^{-1}\big(x_k-F_k(x_{k-1},u_k)\big)$, the information matrix $J_{\xi}(x_k,\underline{y}_k\vert\theta_k)$ simplifies to
\begin{align*}
&J_{\xi}(x_k,\underline{y}_k\vert\theta_k)=(H_k^{\top} R_k^{-1}H_k + Q_k^{-1})\\
& - Q_k^{-1}\Big(\sum_{n=1}^N w_{k-1}^n(x_k\vert x_{k-1}^n,\theta_k)\big[x_k -F_k(x_{k-1}^n,u_k)\big]\\
&\hspace{2cm}\times \big[x_k^{\top} -F_k^{\top}(x_{k-1}^n,u_k)\big]\Big)Q_k^{-1}\\
&+Q_k^{-1}\Big(\sum_{n=1}^N w_{k-1}^n(x_k\vert x_{k-1}^n,\theta_k)\big[x_k -F_k(x_{k-1}^n,u_k)\big]\Big)\\
&\times \Big(\sum_{m=1}^N w_{k-1}^m(x_k\vert x_{k-1}^m,\theta_k)\big[x_k^{\top} - F_k^{\top}(x_{k-1}^m,u_k)\big]\Big)Q_k^{-1}.
\end{align*}

\section{Numerical examples}\label{sec:numeric}

\subsection{Linear state-space model}

Consider the linear state-space (1) with the following parameter values, used in \cite{Ramadan}:
\begin{align*}
x_k=&\left(
\begin{array}{c}
x_{k,1}\\
x_{k,2}\\
x_{k,3}
\end{array}\right),
\;\;
F_k=\left(
\begin{array}{ccc}
0.66 & -1.31 & -1.11\\
0.07 & 0.73 & -0.06\\
0.00 & 0.08 & 0.80
\end{array}\right),\\
H_k=&\left(
\begin{array}{ccc}
0 & 1 & 1
\end{array}\right),
\;\;
Q_k=\left(
\begin{array}{ccc}
0.2 & 0 & 0\\
0  & 0.3 & 0\\
0 & 0 & 0.5
\end{array}\right),\\
R_k=&0.1, \;\;
\mu=\left(
\begin{array}{c}
0\\
0\\
0
\end{array}\right),
\;\;
P_0=\left(
\begin{array}{ccc}
0.3 & 0 & 0\\
0  & 0.3 & 0\\
0 & 0 & 0.3
\end{array}\right),
\end{align*}
in the absence of control variable $u_k$, i.e., $G_k=\mathbf{0}$. From the structure of observation matrix $H_k$, the state-vector $x_k$, started at a random initial $x_0\sim N(\mu,P_0)$, is only partially observed through $y_k$ in terms of superposition (sum) of $x_{k,2}$ and $x_{k,3}$ subject to a measurement error $w_k\sim N(0,R_k)$. Thus, the degree of uncertainty in estimating the state-vector $x_k$ is relatively high. To estimate the state $x_k$, the EM-Gradient-Particle filter (\ref{eq:recmcmc}) is employed. The result is compared to the Kalman filter (\ref{eq:kalman}) and $\mathbb{E}[x_k\vert \underline{y}_k,\theta_k]$ using the particle filtering algorithm of \cite{Kitagawa93,Kitagawa}. In all numerical computations, \textbf{R} language (2013) was used. 

To get the results, a simulation over $K=100$ time-steps, with $N=2000$ particles, is  performed. They are presented in Figure \ref{fig:kalman}. Despite relatively high degree of uncertainty in the observation of state-vector $x_k$, the estimation results $\widehat{x}_k$ exhibit similar behavior of true state $x_k$ with all the three state estimators (the Kalman filter $\widehat{x}_{k\vert k}$, the particle filtering $\mathbb{E}[x_k\vert\underline{y}_k,\theta_k]$ and the ML estimate $\widehat{x}_k$) coincide.

This observation is in line with the fact that for the linear Gaussian state-space model (1), the Kalman filter $\widehat{x}_{k\vert k}$ is a fully efficient ML estimator for which case $\widehat{x}_{k\vert k}=\widehat{x}_k$.

To compute the covariance matrix $P_k$, the simulation is repeated $M=250$ times (at each time step $k$) and uses $\widehat{I}_{\xi}^{-1}(\widehat{x}_k,\underline{y}_k,\theta_k)$ (\ref{eq:estcovmat}) to derive $\widehat{P}_k$. The matrix estimate $\widehat{P}_k$ (\ref{eq:estcovmat}) is compared to the covariance matrix $P_{k\vert k}$ (\ref{eq:kalmancovmat}) and the recursive covariance matrix estimator $\Omega_k$ (\ref{eq:recOFI}), which is run at 50 iterations for each time-step $k$. Below are the results for two (randomly chosen) time steps $k=\{6,83\}$.
\begin{figure}[tp!]
\centering
\includegraphics[width=.965\linewidth]{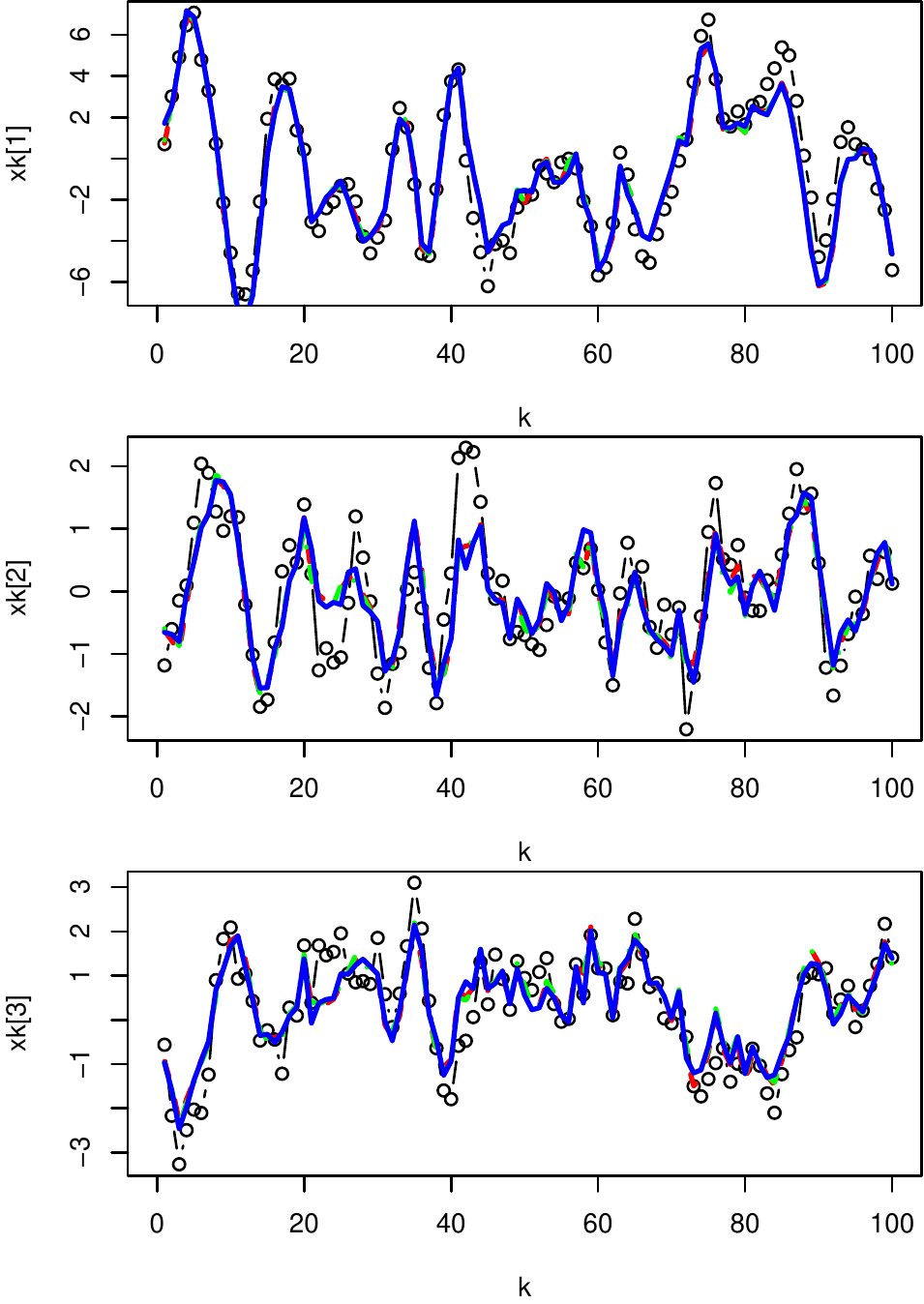}
\caption{Plot of ML estimate $\widehat{x}_k$ based on Kalman filter (\ref{eq:kalman}), particle filter (\ref{eq:kitagawa3}) and EM-Gradient-Particle filter (\ref{eq:recmcmc}). True state $x_k$ is denoted by $-o-$, Kalman filter $\widehat{x}_{k\vert k}$ (dashed), particle filter $\mathbb{E}[x_k\vert \underline{y}_l,\theta_k]$ (dotdashed) and ML estimator $\widehat{x}_k$ (solid). The top graph is for $x_{k,1}$, the second for $x_{k,2}$ and the bottom for $x_{k,3}$. All three estimates coincide. } \label{fig:kalman}
\end{figure}
\begin{enumerate}
\item[(i)] covariance matrices $P_6$, $\Omega_6$ and $P_{6\vert 6}$ are respectively
\begin{align*}
\widehat{P}_6=&\left(
\begin{array}{ccc}
 0.6442 & -0.0720  & 0.0658 \\
 -0.0720 & 0.4443 & -0.4088 \\
 0.0658 & -0.4088 & 0.4629
\end{array}\right),\\[4pt]
\Omega_{6}=&\left(
\begin{array}{ccc}
 0.6442 & -0.0720  & 0.0658\\
 -0.0720 & 0.4443 & -0.4088\\
 0.0658 & -0.4088  &0.4629
\end{array}\right),\\[4pt]
P_{6\vert 6}=&\left(
\begin{array}{ccc}
 0.6448 & -0.0778  & 0.0712 \\
 -0.0778 & 0.4458  & -0.4103 \\
  0.0712 & -0.4103 & 0.4644
\end{array}\right),
\end{align*}
\item[(ii)] covariance matrices $P_{83}$, $\Omega_{83}$, and $P_{83\vert 83}$ are given by
\begin{align*}
\widehat{P}_{83}=&\left(
\begin{array}{ccc}
 0.6591 & -0.0885 & 0.0817\\
 -0.0885 & 0.4531 & -0.4177\\
 0.0817 & -0.4177  & 0.4718
\end{array}\right),\\[4pt]
\Omega_{83}=&\left(
\begin{array}{ccc}
 0.6591 & -0.0885  & 0.0817 \\
 -0.0885 & 0.4531 & -0.4177 \\
 0.0817 & -0.4177 & 0.4718
\end{array}\right),\\[4pt]
P_{83\vert 83}=&\left(
\begin{array}{ccc}
0.6601 & -0.0867 & 0.0801 \\
 -0.0867 & 0.4530 & -0.4175 \\
 0.0801 &-0.4175  &0.4716
\end{array}\right).
\end{align*}
\end{enumerate}

\begin{figure}[tp!]
\centering
\includegraphics[width=.865\linewidth]{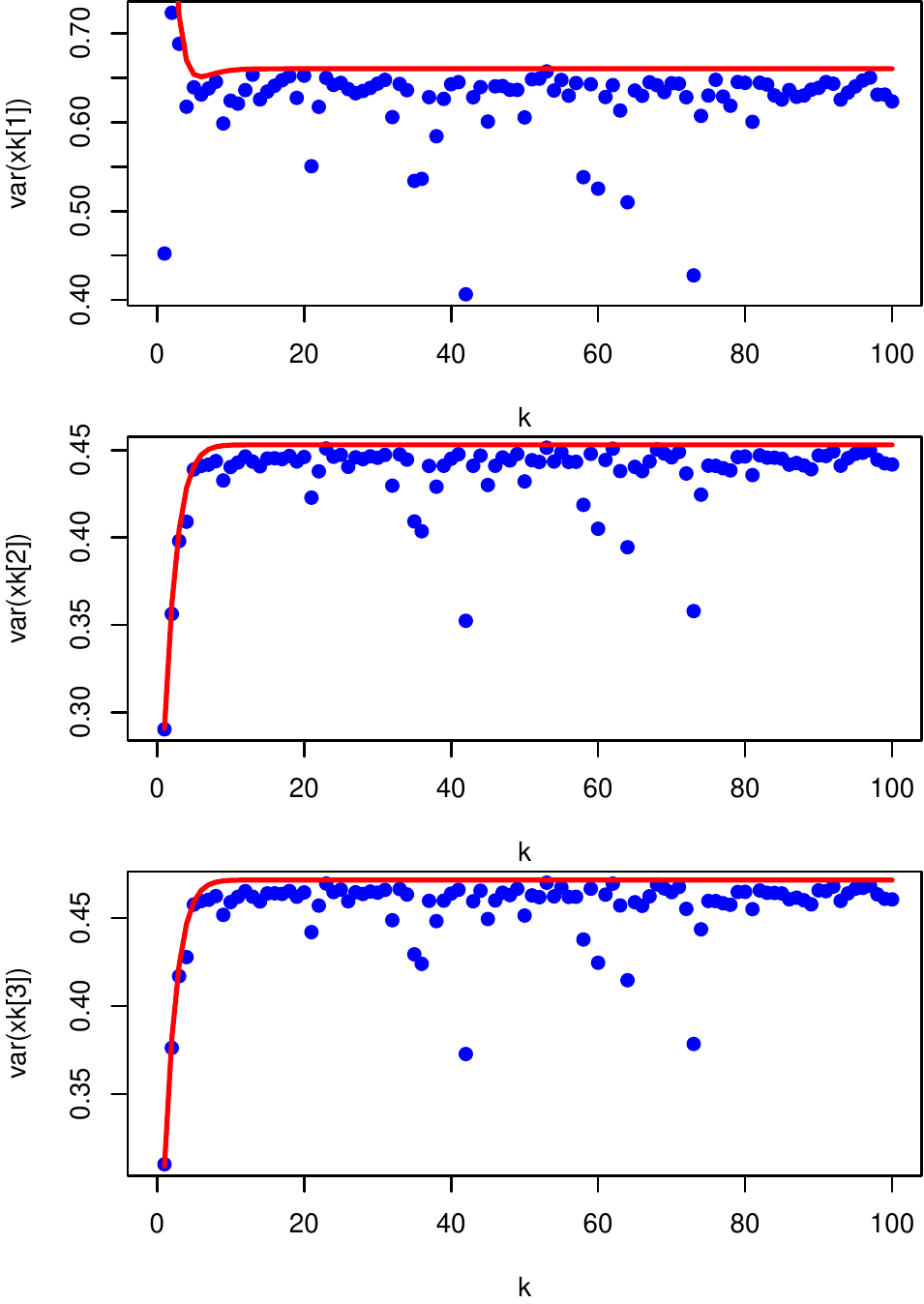}
\caption{Plots of estimated variance $\widehat{\textrm{var}}(\widehat{x}_{k})$, the diagonal element of the covariance matrix $\widehat{P}_k=\widehat{I}_{\xi}(\widehat{x}_k,\underline{y}_k\vert\theta_k)$, and that of obtained from the posterior variance $P_{k\vert k}$ (\ref{eq:kalmancovmat}). } \label{fig:stdev}
\end{figure}

From the results, we notice that recursive covariance matrix estimator $\Omega_k^{\ell}$ (\ref{eq:recOFI}) converges to observed information matrix $\widehat{P}_k$. However, there are some slight variations between two matrices $\widehat{P}_k$ and $P_{k\vert k}$, which may be attributed by a small sampling size $M=250$ for each time step $k$. 

Figures \ref{fig:stdev} and \ref{fig:stdeviter} display plots of estimated variance $\widehat{\textrm{var}}(\widehat{x}_k)$ of the ML estimator $\widehat{x}_k$ given by the diagonal element of the covariance matrices $\widehat{P}_k=\widehat{I}_{\xi}^{-1}(\widehat{x}_k,\underline{y}_k\vert\theta_k)$ and $\Omega_k$ (\ref{eq:recOFI}), respectively. That is $\widehat{\textrm{var}}(\widehat{x}_{k,1})$ is given by the first diagonal element of $\widehat{P}_k$, similarly for  $\widehat{\textrm{var}}(\widehat{x}_{k,2})$ and  $\widehat{\textrm{var}}(\widehat{x}_{k,3})$. These variances are compared with the corresponding element of $P_{k\vert k}$. From the two plots we observe that there are some apparent variations in the variances. But overall, the estimated variance (diagonal element of) $\widehat{P}_k$ shows the convergence to $P_{k\vert k}$ for each $k$, the upper bound of $\widehat{P}_k$.
\begin{figure}[tp!]
\centering
\includegraphics[width=.865\linewidth]{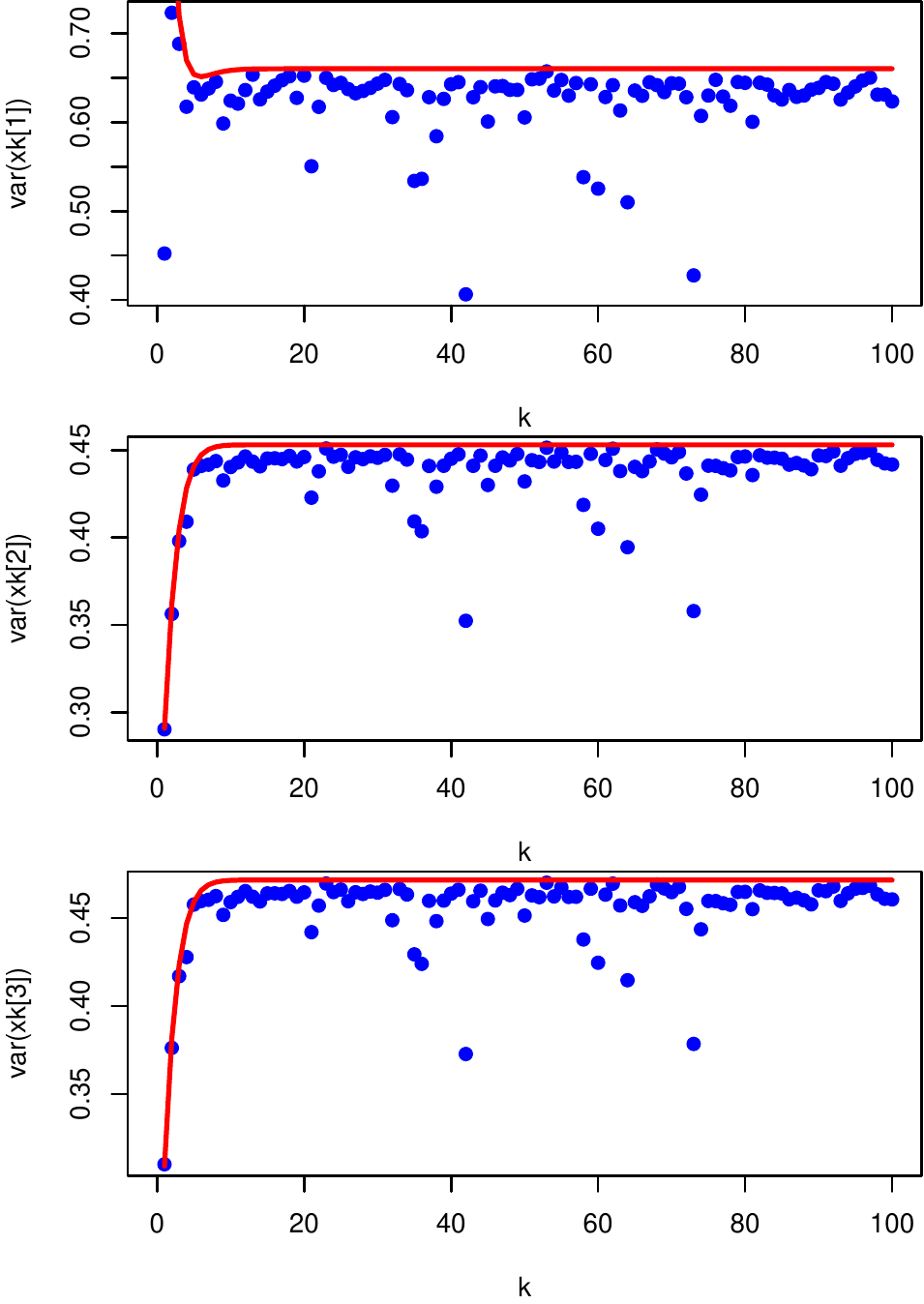}
\caption{Plots of estimated variance $\widehat{\textrm{var}}(\widehat{x}_{k})$, the diagonal element of the recursive covariance estimator $\Omega_k$ (\ref{eq:recOFI}), and that of obtained from the posterior variance $P_{k\vert k}$ (\ref{eq:kalmancovmat}). } \label{fig:stdeviter}
\end{figure}

\subsection{Nonlinear state-space model}

In this example, the recursive EM-Gradient-Particle filter (\ref{eq:recmcmc2}) is applied to a nonlinear state-space model
\begin{equation}\label{eq:NLSP}
\begin{split}
x_k=& f_k \tanh(\pi x_{k-1}) + v_k \\
y_k=& \frac{1}{2} x_k + w_k
\end{split}
\end{equation}
with $f_k=(1+0.5\sin(2\pi k/20))$, $v_k\sim N(0,0.2)$, $w_k\sim N(0,1)$ and $x_0\sim N(0,1)$. This is the same model considered in \cite{Ramadan}. By (\ref{eq:recmcmc2}), we have
\begin{align*}
&\widehat{x}_k^{(\ell+1)}= \frac{2}{21}\big[ y_k + 10\sum_{n=1}^N w_{k-1}^n(\widehat{x}_k^{(\ell)}\vert x_{k-1}^n,\theta_k) f_k \tanh(\pi x_{k-1}^n)\big].\nonumber
\end{align*}
%
%\begin{figure}[tp!]
%\centering
%\includegraphics[width=.975\linewidth]{Rplot04.pdf}
%\caption{ The top graph displays plots of estimates in the nonlinear state-space (\ref{eq:NLSP}) based on particle filter (\ref{eq:kitagawa3}) and EM-Gradient-Particle filter (\ref{eq:recmcmc}). True state $x_k$ is denoted by $-o-$, particle filter $\mathbb{E}[x_k\vert \underline{y}_l,\theta_k]$ (dashed) and ML estimator $\widehat{x}_k$ (solid). The dotted line represents the $95\%$ confidence interval.  The bottom graph exhibits the plots of sample variance $s_{\widehat{x}_k}^2=\frac{1}{M}\sum_{m=1}^M (\widehat{x}_{k,m}^{\infty}-x_k)^2$ (denoted by $-o-$) against the inverse of $\widehat{I}_{\xi}^{-1}(\widehat{x}_k^0,\underline{y}_k\vert\theta_k)$ (\ref{eq:estcovmat}) (solid line).  } \label{fig:NLfilter}
%\end{figure}

\begin{figure}[tp!]
\centering
\subfigure[MLE $\widehat{x}_k^{\infty}$ (and $95\%$ confidence interval) and particle filter $\widehat{x}_{k\vert k}$]{\includegraphics[width=0.495\textwidth]{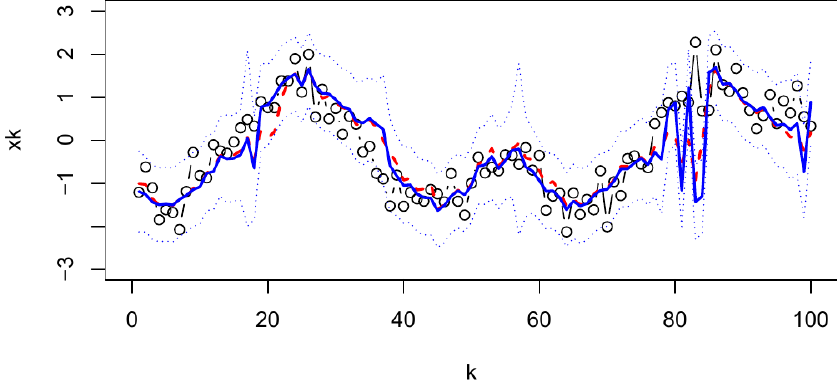}} \\
\subfigure[Sample variance and estimated variance of MLE $\widehat{x}_k^{\infty}$]{\includegraphics[width=0.45\textwidth]{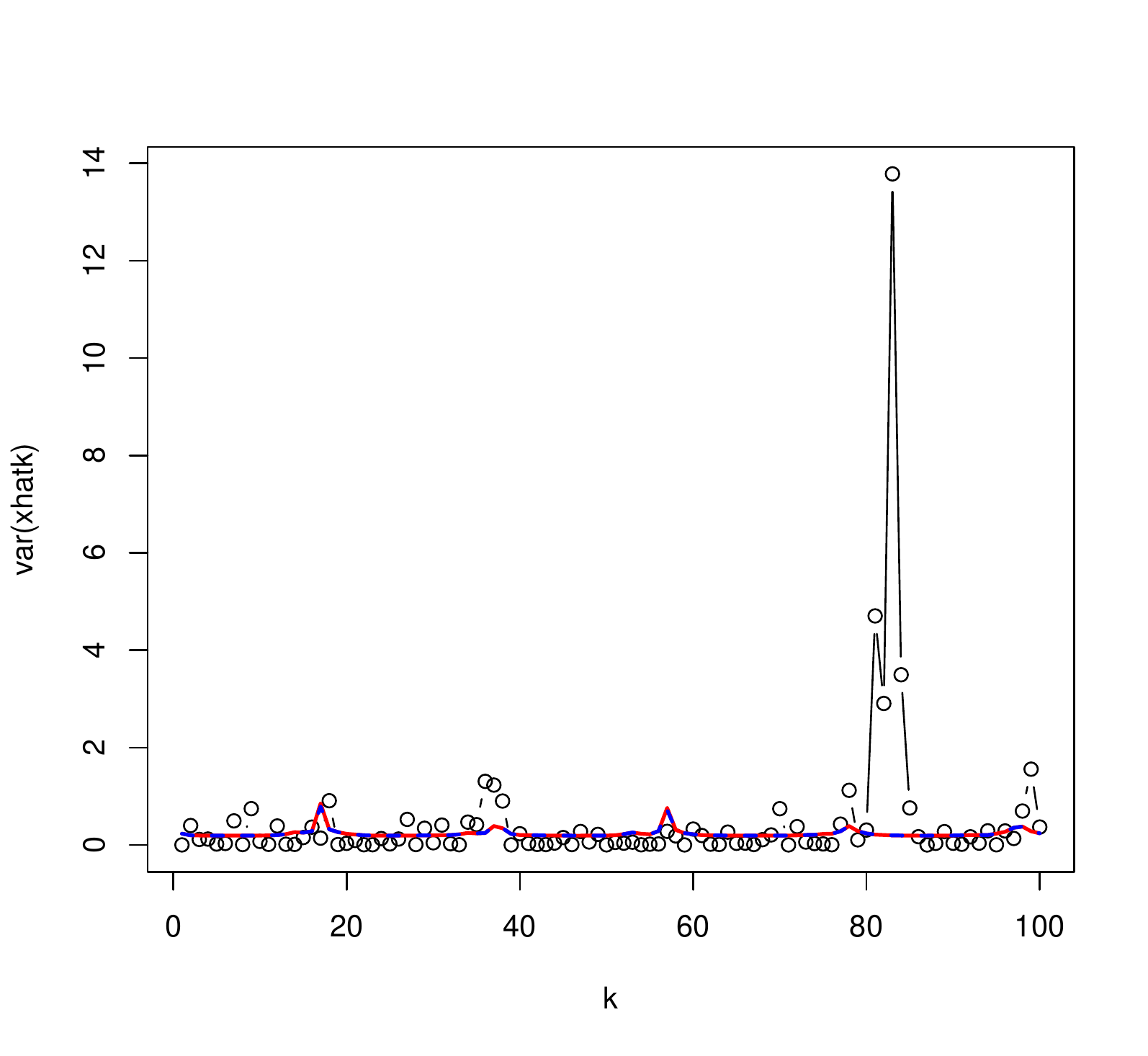}}
\caption{ The top graph displays plots of estimates in the nonlinear state-space (\ref{eq:NLSP}) based on particle filter (\ref{eq:kitagawa3}) and EM-Gradient-Particle filter (\ref{eq:recmcmc}). True state $x_k$ is denoted by $-o-$, particle filter $\mathbb{E}[x_k\vert \underline{y}_l,\theta_k]$ (dashed) and ML estimator $\widehat{x}_k$ (solid). The dotted line represents the $95\%$ confidence interval.  The bottom graph exhibits the plots of sample variance $s_{\widehat{x}_k}^2=\frac{1}{M}\sum_{m=1}^M (\widehat{x}_{k,m}^{\infty}-x_k)^2$ (denoted by $-o-$) against the inverse of $\widehat{I}_{\xi}^{-1}(\widehat{x}_k^0,\underline{y}_k\vert\theta_k)$ (\ref{eq:estcovmat}) (solid line) and recursive variance estimator $\Omega_k$ (dashed) using (\ref{eq:recOFI}).  } \label{fig:NLfilter}
\end{figure}

\vspace{-1.1cm}

Notice that the obtained equation is slightly different from that of presented in eqn.(37) in \cite{Ramadan}. For the estimation of state $x_k$, the above recursive equation is employed on the basis of $K=100$ time steps and $N=2000$ sampling particles. The estimation results are presented in Figure \ref{fig:NLfilter}. From the top graph we observe that the ML estimator $\widehat{x}_k$ has more (likelihood) adherence to the true state $x_k$ compared to the particle filter $\widehat{x}_{k\vert k}=\mathbb{E}[x_k\vert \underline{y}_k,\theta_k]$. Despite some irregularities in sample variance $s_{\widehat{x}_k}^2=\frac{1}{M}\sum_{m=1}^M (\widehat{x}_{k,m}^{\infty}-x_k)^2$, overall estimates $\widehat{\sigma}_{\widehat{x}_k}^2$ of the variance derived from the inverse of information $\widehat{P}_k=\widehat{I}_{\xi}^{-1}(\widehat{x}_k^0,\underline{y}_k\vert\theta_k)$ and the recursive variance estimator $\Omega_k$ (\ref{eq:recOFI}) are relatively close to the sample variance $s_{\widehat{x}_k}^2$. 

Furthermore, Figure \ref{fig:NLfilter} shows that the $95\%$ confidence interval $(\widehat{x}_k\pm 1.96\times \widehat{\sigma}_{\widehat{x}_k})$ of the state estimation contains all true states $x_k$, with only exception at time-step $k=83$.

\section{Concluding remarks}\label{sec:conclusion}

This paper developed a new approach for maximum likelihood recursive state estimation. The novelty of the approach is central to formulating $(x_k,\underline{y}_k)$ as an incomplete information of the systems, which is slightly different from the existing methods. Such formulation allows one applying statistical analysis of incomplete data, developed primarily for parameter estimation (see e.g.,\cite{Dempster} and \cite{Little}), to obtain maximum likelihood estimate of the state-vector $x_k$. The ML estimate is given recursively by the EM-Gradient-Particle filtering algorithm in terms of the score function and observed information matrices of the incomplete observations $(x_k,\underline{y}_k)$. The score function and information matrices are evaluated using the particle filtering developed in \cite{Kitagawa,Kitagawa93}. This approach can be used to numerically evaluate the Cram\'er-Rao lower bound. The latter is needed to derive an estimate of covariance matrix of estimation errors. Numerical study verifies the main results of this paper. The new approach can be applied for maximum likelihood recursive state estimation in general state-space models. 

%\section*{Acknowledgements}
%
%%The authors are grateful to an anonymous referee  for very helpful comments.
%He acknowledges the Faculty Strategic Research
%Grant No. 20859 of Victoria University of Wellington and hospitality
%provided by the hosts of his visits.

%{\it Conflict of Interest}: None declared.

\section{Appendix}

\subsection*{\textbf{A}: Derivation of information matrix (\ref{eq:fisher})}
By taking derivative w.r.t $x_k$ directly to identity (\ref{eq:mainidentity}),
\begin{align*}
&\frac{\partial^2 \log f(\xi_k\vert\theta_k)}{\partial x_k \partial x_k^{\top}}=\frac{\partial}{\partial x_k}\Big(\frac{\partial \log f(x_k,\underline{y}_k\vert\theta_k)}{\partial x_k^{\top}}\Big)\\
=&\frac{\partial}{\partial x_k}\mathbb{E}\Big[\frac{\partial \log f(x_k,\underline{y}_k,x_{k-1}\vert\theta_k)}{\partial x_k^{\top}}\Big\vert x_k,\underline{y}_k,\theta_k\Big]\\
=&\frac{\partial}{\partial x_k}\int_{\mathcal{X}_k} \frac{\partial \log f(x_k,\underline{y}_k,x_{k-1}\vert\theta_k)}{\partial x_k^{\top}} \\
&\hspace{2cm}\times f(x_{k-1}\vert x_k,\underline{y}_k,\theta_k) d\lambda(x_{k-1})\\
=&\int_{\mathcal{X}_k} \frac{\partial^2 \log f(x_k,\underline{y}_k,x_{k-1}\vert\theta_k)}{\partial x_k \partial x_k^{\top}} \\
&\hspace{2cm}\times f(x_{k-1}\vert x_k,\underline{y}_k,\theta_k) d\lambda(x_{k-1})\\
&+ \int_{\mathcal{X}_k} \frac{\partial \log f(x_{k-1} \vert x_k,\underline{y}_k,\theta_k)}{\partial x_k} \frac{\partial \log f(x_k,\underline{y}_k,x_{k-1}\vert\theta_k)}{\partial x_k^{\top}}\\
&\hspace{2cm}\times f(x_{k-1}\vert x_k,\underline{y}_k,\theta_k) d\lambda(x_{k-1})\\
=&\mathbb{E}\Big[ \frac{\partial^2 \log f(x_k,\underline{y}_k,x_{k-1}\vert\theta_k)}{\partial x_k \partial x_k^{\top}} \Big\vert x_k,\underline{y}_k,\theta_k\Big]\\
&+ \int_{\mathcal{X}_k} \frac{\partial \log f(x_k,\underline{y}_k,x_{k-1}\vert \theta_k)}{\partial x_k} \frac{\partial \log f(x_k,\underline{y}_k,x_{k-1}\vert\theta_k)}{\partial x_k^{\top}}\\
&\hspace{2cm}\times f(x_{k-1}\vert x_k,\underline{y}_k,\theta_k) d\lambda(x_{k-1})\\
&- \frac{\partial \log f(x_k,\underline{y}_k\vert \theta_k)}{\partial x_k}\int_{\mathcal{X}_k}  \frac{\partial \log f(x_k,\underline{y}_k,x_{k-1}\vert\theta_k)}{\partial x_k^{\top}}\\
&\hspace{2cm}\times f(x_{k-1}\vert x_k,\underline{y}_k,\theta_k) d\lambda(x_{k-1}),
\end{align*}
where in the last equality we have used the fact that 
\begin{align*}
 \frac{\partial \log f(x_{k-1} \vert x_k,\underline{y}_k, \theta_k)}{\partial x_k}=& \frac{\partial \log f(x_k,\underline{y}_k,x_{k-1}\vert \theta_k)}{\partial x_k}\\
 &-\frac{\partial \log f(x_k,\underline{y}_k\vert \theta_k)}{\partial x_k}.
\end{align*}
The proof is complete on account of identity (\ref{eq:mainidentity}). $\square$

\subsection*{\textbf{B}: Proof of Theorem \ref{theo:recOFI}}

Following Theorem \ref{theo:infoloss}, the observed information matrices $J_{\xi}(\widehat{x}_k,\underline{y}_k\vert \theta_k)$ and $J_{z}(\widehat{x}_k,\underline{y}_k\vert \theta_k)$ are positive definite with $J_{z}(\widehat{x}_k,\underline{y}_k\vert \theta_k)>J_{\xi}(\widehat{x}_k,\underline{y}_k\vert \theta_k)>0$ which by Theorem 7.2.1 on p.438 of \cite{Horn}, $I> J_{z}^{-1}(\widehat{x}_k,\underline{y}_k\vert \theta_k)J_{\xi}(\widehat{x}_k,\underline{y}_k\vert \theta_k)>0$. Therefore, all eigenvalues of $I-J_{z}^{-1}(\widehat{x}_k,\underline{y}_k\vert \theta_k)J_{\xi}(\widehat{x}_k,\underline{y}_k\vert \theta_k)$ are positive and strictly less than one. See Corollary 1.3.4 in \cite{Horn}. By Corollary 5.6.16 of \cite{Horn},
\begin{align*}
\Omega_k=&\big[I-J_{z}^{-1}(\widehat{x}_k,\underline{y}_k\vert \theta_k)\big(J_{z}(\widehat{x}_k,\underline{y}_k\vert \theta_k)-J_{\xi}(\widehat{x}_k,\underline{y}_k\vert \theta_k)\big)\big]^{-1}\\
&\hspace{0.5cm}\times J_{z}^{-1}(\widehat{x}_k,\underline{y}_k\vert \theta_k)\\
&\hspace{-1cm}=\left(\sum_{\ell=0}^{\infty}\big[I- J_{z}^{-1}(\widehat{x}_k,\underline{y}_k\vert \theta_k)J_{\xi}(\widehat{x}_k,\underline{y}_k\vert \theta_k)\big]^{\ell}\right)^{-1}J_{z}^{-1}(\widehat{x}_k,\underline{y}_k\vert \theta_k).
\end{align*}
Since all eigenvalues of $I-J_{z}^{-1}(\widehat{x}_k,\underline{y}_k\vert \theta_k)J_{\xi}(\widehat{x}_k,\underline{y}_k\vert \theta_k)$ are positive and strictly less than one, it follows that
\begin{align*}
\Omega^{\ell+1} - \Omega=&\big[I- J_{z}^{-1}(\widehat{x}_k,\underline{y}_k\vert \theta_k)J_{\xi}(\widehat{x}_k,\underline{y}_k\vert \theta_k)\big]\Omega^{\ell}\\
&+J_{z}^{-1}(\widehat{x}_k,\underline{y}_k\vert \theta_k)-\Omega\\
=&\big[I- J_{z}^{-1}(\widehat{x}_k,\underline{y}_k\vert \theta_k)J_{\xi}(\widehat{x}_k,\underline{y}_k\vert \theta_k)\big]\big(\Omega^{\ell}-\Omega\big),
\end{align*}
which converges to (matrix) zero as $\ell\rightarrow \infty$ with root of convergence factor given by the maximum absolute eigenvalues of $I- J_{z}^{-1}(\widehat{x}_k,\underline{y}_k\vert \theta_k)J_{\xi}(\widehat{x}_k,\underline{y}_k\vert \theta_k)$. Thus, the matrix sequence $\{\Omega^{\ell}\}_{\ell\geq 1}$ converges to $\Omega=J_{\xi}^{-1}(\widehat{x}_k,\underline{y}_k\vert \theta_k)$ as $\ell\rightarrow \infty$. In fact the convergence is monotone since
\begin{align*}
\Omega^{\ell+1}-\Omega^{\ell}=&\big[I- J_{z}^{-1}(\widehat{x}_k,\underline{y}_k\vert \theta_k)J_{\xi}(\widehat{x}_k,\underline{y}_k\vert \theta_k)\big]\big(\Omega^{\ell}-\Omega^{\ell-1}\big)\\
&\hspace{-1cm}=\big[I- J_{z}^{-1}(\widehat{x}_k,\underline{y}_k\vert \theta_k)J_{\xi}(\widehat{x}_k,\underline{y}_k\vert \theta_k)\big]^{\ell} J_{z}^{-1}(\widehat{x}_k,\underline{y}_k\vert \theta_k),
\end{align*}
which is positive definite by Corollary 7.7.4(a) of \cite{Horn} and the information-loss inequality (\ref{eq:ordering}). $\square$

\subsection*{\textbf{C}: Proof of Proposition \ref{prop:propass2}}
The proof for score function is established directly as
\begin{align*}
&\mathbb{E}\Big[\frac{\partial \log f(x_k^0\vert\theta_k)}{\partial x_k}\big\vert \theta_k\Big]
%=\int_{\mathcal{X}_k} \frac{\partial \log f(x_k^0\vert\theta_k)}{\partial x_k}  f(x_k^0\vert\theta_k) dx_k^0 
=\int_{\mathcal{X}_k} \frac{\partial f(x_k^0\vert\theta_k)}{\partial x_k} dx_k^0:=\Delta f(x_k^0)\big\vert_{x_k^0\in \partial\mathcal{X}_k} \\
&=\left(
\begin{array}{c}
\int f(x_{k,1}^0=+\infty, x_{k,2}^0,\ldots,x_{k,p}^0) dx_{k,2}^0\ldots dx_{k,p}^0\\
\int f(x_{k,1}^0, x_{k,2}^0=+\infty,\ldots,x_{k,p}^0) dx_{k,1}^0dx_{k,3}^0\ldots dx_{k,p}^0\\
\vdots\\
\int f(x_{k,1}^0, x_{k,2}^0,\ldots,x_{k,p}^0=+\infty)dx_{k,1}^0\dots x_{k,p-1}^0
\end{array}\right)\\[3pt]
&-\left(
\begin{array}{c}
\int f(x_{k,1}^0=-\infty, x_{k,2}^0,\ldots,x_{k,p}^0) dx_{k,2}^0\ldots dx_{k,p}^0\\
\int f(x_{k,1}^0, x_{k,2}^0=-\infty,\ldots,x_{k,p}^0) dx_{k,1}^0dx_{k,3}^0\ldots dx_{k,p}^0\\
\vdots\\
\int f(x_{k,1}^0, x_{k,2}^0,\ldots,x_{k,p}^0=-\infty)dx_{k,1}^0\dots x_{k,p-1}^0
\end{array}\right),
\end{align*}
which completes the proof of the first claim by (\ref{eq:ass3B}). To prove the second claim, the following identity is required
\begin{align*}
&\frac{\partial^2 \log f(x_k\vert\theta_k)}{\partial x_k \partial x_k^{\top}}=\frac{1}{f(x_k\vert\theta_k)}\frac{\partial^2 f(x_k\vert\theta_k)}{\partial x_k \partial x_k^{\top}} \\
&\hspace{2cm}-\frac{\partial \log f(x_k\vert \theta_k)}{\partial x_k}\frac{\partial \log f(x_k\vert \theta_k)}{\partial x_k^{\top}},
\end{align*}
which is derived using the chain rule. Thus,
\begin{align*}
&\mathbb{E}\Big[\frac{\partial^2 \log f(x_k^0\vert\theta_k)}{\partial x_k \partial x_k^{\top}}\big\vert \theta_k\Big]=\int_{\mathcal{X}_k} \frac{\partial^2 \log f(x_k^0\vert\theta_k)}{\partial x_k \partial x_k^{\top}} f(x_k^0\vert \theta_k) dx_k^0\\
&=\int_{\mathcal{X}_k} \frac{\partial^2 f(x_k\vert\theta_k)}{\partial x_k \partial x_k^{\top}} dx_k^0 - \mathbb{E}\Big[\frac{\partial \log f(x_k\vert \theta_k)}{\partial x_k}\frac{\partial \log f(x_k\vert \theta_k)}{\partial x_k^{\top}}\big\vert \theta_k\Big]\\
&=\Delta \frac{\partial f(x_k\vert\theta_k)}{\partial x_k}\Big\vert _{x_k\in\partial \mathcal{X}_k} - \mathbb{E}\Big[\frac{\partial \log f(x_k\vert \theta_k)}{\partial x_k}\frac{\partial \log f(x_k\vert \theta_k)}{\partial x_k^{\top}}\big\vert \theta_k\Big],
\end{align*}
where $\Delta \frac{\partial f(x_k\vert\theta_k)}{\partial x_k}\Big\vert _{x_k\in\partial \mathcal{X}_k} $ is defined similarly to that of $\Delta f(x_k^0\vert\theta_k)\in \partial \mathcal{X}_k$, which by (\ref{eq:ass3B}) completes the proof. $\square$

\subsection*{\textbf{D}: Proof of Proposition \ref{prop:fishermatrix}}
By taking derivative w.r.t $x_k^0$ on (\ref{eq:ass3B}), we obtain
\begin{align*}
&\frac{\partial^2 \log f(x_k^0\vert \theta_k)}{\partial x_k \partial x_k^{\top}}=\frac{\partial}{\partial x_k^0} \mathbb{E}\Big[\frac{\partial \log f(x_k^0, \underline{y}_k  \vert\theta_k)}{\partial x_k^{\top}}\Big\vert x_k^0,\theta_k\Big]\\
=&\frac{\partial}{\partial x_k^0} \int  \frac{\partial \log f(x_k^0, \underline{y}_k  \vert\theta_k)}{\partial x_k^{\top}} f(\underline{y}_k\vert x_k^0,\theta_k) d\lambda(\underline{y}_k)\\
=& \int  \frac{\partial^2 \log f(x_k^0, \underline{y}_k  \vert\theta_k)}{\partial x_k \partial x_k^{\top}} f(\underline{y}_k\vert x_k^0,\theta_k) d\lambda(\underline{y}_k)\\
&+  \int  \frac{\partial f(\underline{y}_k\vert x_k^0,\theta_k)}{\partial x_k^0}\frac{\partial \log f(x_k^0, \underline{y}_k  \vert\theta_k)}{\partial x_k^{\top}}  d\lambda(\underline{y}_k)\\
=& \int  \frac{\partial^2 \log f(x_k^0, \underline{y}_k  \vert\theta_k)}{\partial x_k \partial x_k^{\top}} f(\underline{y}_k\vert x_k^0,\theta_k) d\lambda(\underline{y}_k)\\
&+  \int \Big( \frac{\partial \log f(x_k^0,\underline{y}_k\vert \theta_k)}{\partial x_k} - \frac{\partial \log f(x_k^0\vert \theta_k)}{\partial x_k}\Big) \\
&\hspace{1cm}\times \frac{ \partial \log f(x_k^0, \underline{y}_k  \vert\theta_k)}{\partial x_k^{\top}} f(\underline{y}_k\vert x_k^0,\theta_k)  d\lambda(\underline{y}_k)\\
=& \int  \frac{\partial^2 \log f(x_k^0, \underline{y}_k  \vert\theta_k)}{\partial x_k \partial x_k^{\top}} f(\underline{y}_k\vert x_k^0,\theta_k) d\lambda(\underline{y}_k)\\
&+  \int \Big( \frac{\partial \log f(x_k^0,\underline{y}_k\vert \theta_k)}{\partial x_k} - \frac{\partial \log f(x_k^0\vert \theta_k)}{\partial x_k}\Big) \\
&\hspace{1cm}\times \frac{ \partial\log f(x_k^0, \underline{y}_k  \vert\theta_k)}{\partial x_k^{\top}} f(\underline{y}_k\vert x_k^0,\theta_k)  d\lambda(\underline{y}_k)\\
=& \int  \frac{\partial^2 \log f(x_k^0, \underline{y}_k  \vert\theta_k)}{\partial x_k \partial x_k^{\top}} f(\underline{y}_k\vert x_k^0,\theta_k) d\lambda(\underline{y}_k)\\
&\hspace{-0.5cm}+  \int \frac{\partial \log f(x_k^0,\underline{y}_k\vert \theta_k)}{\partial x_k} \frac{ \partial\log f(x_k^0, \underline{y}_k  \vert\theta_k)}{\partial x_k^{\top}} f(\underline{y}_k\vert x_k^0,\theta_k)  d\lambda(\underline{y}_k)\\
& \hspace{-0.5cm}- \frac{\partial \log f(x_k^0\vert \theta_k)}{\partial x_k}  \int  \frac{ \partial\log f(x_k^0, \underline{y}_k  \vert\theta_k)}{\partial x_k^{\top}} f(\underline{y}_k\vert x_k^0,\theta_k)  d\lambda(\underline{y}_k)\\
=&\mathbb{E}\Big[ \frac{\partial^2 \log f(x_k^0, \underline{y}_k  \vert\theta_k)}{\partial x_k \partial x_k^{\top}}\Big\vert x_k^0,\theta_k\Big]\\
&+\mathbb{E}\Big[   \frac{\partial \log f(x_k^0,\underline{y}_k\vert \theta_k)}{\partial x_k} \frac{ \partial\log f(x_k^0, \underline{y}_k  \vert\theta_k)}{\partial x_k^{\top}}   \Big\vert x_k^0,\theta_k\Big]\\
&- \frac{\partial \log f(x_k^0\vert \theta_k)}{\partial x_k}\mathbb{E}\Big[ \frac{ \partial\log f(x_k^0, \underline{y}_k  \vert\theta_k)}{\partial x_k^{\top}}\Big\vert x_k^0,\theta_k\Big],
\end{align*}
leading to (\ref{eq:fishermatrix}) on account of the identity (\ref{eq:mainidentityB}), Proposition \ref{prop:propass2} and by taking expectation $\mathbb{E}[\bullet\vert\theta_k]$ on both sides. $\square$

\subsection*{\textbf{E}: Proof of Lemma \ref{lem:smoothing}}

To establish the claim, recall following identity (\ref{eq:bayes2}) that 
\begin{align*}
\mathbb{E}\big[x_{k-1}\vert x_k,\underline{y}_k,\theta_k\big]=&\int x_{k-1} f(x_{k-1}\vert x_k,\underline{y}_k,\theta_k)d\lambda (x_{k-1})\\
=&\int x_{k-1} f(x_{k-1}\vert x_k,\underline{y}_{k-1},\theta_k)d\lambda (x_{k-1})\\
=&\mathbb{E}\big[x_{k-1}\vert x_k,\underline{y}_{k-1},\theta_k\big]
\end{align*}
The proof follows on account of (\ref{eq:solofx}) that $x_{k-1}$, $x_k$ and $\underline{y}_{k-1}$ are jointly Gaussian. Using innovation approach of \cite{Kailath} and Theorem 3.3 in Ch. 7 of \cite{Astrom},
\begin{align*}
&\mathbb{E}\big[x_{k-1} \big\vert x_k,\underline{y}_{k-1},\theta_k\big]=\mathbb{E}\big[x_{k-1} \big\vert x_k-\widehat{x}_{k\vert k-1}, \underline{y}_{k-1},\theta_k\big]\\
&= \mathbb{E}[x_{k-1}\vert \underline{y}_{k-1},\theta_k] +  \mathbb{E}[x_{k-1}\vert x_k-\widehat{x}_{k\vert k-1},\theta_k] -\mathbb{E}[x_{k-1}\vert \theta_k]\\
&= \widehat{x}_{k-1\vert k-1} +  \mathbb{E}[x_{k-1}\vert x_k-\widehat{x}_{k\vert k-1},\theta_k] -\mathbb{E}[x_{k-1}\vert \theta_k].
\end{align*}
By applying Theorem 3.2 in Ch.7 of \cite{Astrom}, 
\begin{align*}
 &\mathbb{E}[x_{k-1}\vert x_k-\widehat{x}_{k\vert k-1},\theta_k] =\mathbb{E}[x_{k-1}\vert\theta_k] \\
&\hspace{0.5cm}+ \mathrm{Cov}(x_{k-1}, x_k-\widehat{x}_{k\vert k-1}\vert \theta_k)\mathrm{Var}^{-1} \big(x_k-\widehat{x}_{k\vert k-1}\vert \theta_k\big)\\
&\hspace{1.5cm}\times (x_k-\widehat{x}_{k\vert k-1}),
\end{align*}
where the covariance matrix simplifies to
\begin{align*}
& \mathrm{Cov}(x_{k-1}, x_k-\widehat{x}_{k\vert k-1}\vert \theta_k)=\mathbb{E}\big[x_{k-1}\big(x_k-\widehat{x}_{k\vert k-1}\big)^{\top}\vert \theta_k\big]\\
&=\mathbb{E}\big[\big(x_{k-1}-\widehat{x}_{k-1\vert k-1}\big)\big(x_k-\widehat{x}_{k\vert k-1}\big)^{\top}\vert \theta_k\big]\\
&\hspace{1cm}+ \mathbb{E}\big[\widehat{x}_{k-1\vert k-1}\big(x_k-\widehat{x}_{k\vert k-1}\big)^{\top}\vert \theta_k\big]=\Sigma_{k\vert k-1},
\end{align*}
since $ \mathbb{E}\big[\widehat{x}_{k-1\vert k-1}\big(x_k-\widehat{x}_{k\vert k-1}\big)^{\top}\vert \theta_k\big]=0$. Thus,
\begin{align*}
&\mathbb{E}[x_{k-1}\vert x_k, \underline{y}_{k-1},\theta_k]\\
&\hspace{1cm}=\widehat{x}_{k-1\vert k-1} + \Sigma_{k\vert k-1} P_{k\vert k-1}^{-1}(x_k-\widehat{x}_{k\vert k-1}).
\end{align*}
It remains to derive the recursive equation for the matrix $$\Sigma_{k+1\vert k}\equiv\mathbb{E}\big[(x_k - \widehat{x}_{k\vert k})(x_{k+1}-\widehat{x}_{k+1\vert k})^{\top}\big\vert \theta_k\big].$$ For this purpose, one can show using (1) that
\begin{align*}
x_k -\widehat{x}_{k\vert k}=&F_k(x_{k-1} - \widehat{x}_{k-1\vert k-1}) + v_k - K_k \overline{y}_k.
\end{align*}
see (\ref{eq:kalman}) and (\ref{eq:eq1a}). Also, by similar arguments used before, we derive using Theorems 3.2 and 3.3 in \cite{Astrom}, 
\begin{align*}
\widehat{x}_{k+1\vert k}\equiv& \mathbb{E}[ x_{k+1}\vert \underline{y}_k,\theta_k]\\=&F_{k+1} \widehat{x}_{k\vert k-1} + G_{k+1} u_{k+1}\\
+&F_{k+1} P_{k\vert k-1} H_k^{\top}(H_k P_{k\vert k-1} H_k^{\top} + R_k )^{-1}\overline{y}_k.
\end{align*}
Thus, a priori estimation error at $k+1$ reads as
\begin{align*}
x_{k+1}-\widehat{x}_{k+1\vert k}=F_{k+1}(x_k-\widehat{x}_{k\vert k-1}) - F_{k+1}K_k \overline{y}_k + v_{k+1}.
\end{align*}
Furthermore, after some calculations one can show that 
\begin{align*}
\mathbb{E}\big[(x_{k-1} - \widehat{x}_{k\vert k-1})\overline{y}_k^{\top}\vert \theta_k\big]=& \Sigma_{k\vert k-1} H_k^{\top},\\
\mathbb{E}\big[\overline{y}_k (x_k -\widehat{x}_{k\vert k-1})^{\top} \vert \theta_k\big]=&H_k P_{k\vert k-1},\\
\mathbb{E}\big[v_k(x_k-\widehat{x}_{k\vert k-1})^{\top}\vert \theta_k\big]=& Q_k,\\
\mathbb{E}\big[v_k \overline{y}_k^{\top}  \vert \theta_k\big]=&Q_k H_k^{\top}.
\end{align*}
Using the above results, we obtain by independence of $\{v_k\}$
\begin{align*}
\Sigma_{k+1\vert k}=&F_k \mathbb{E}\big[(x_{k-1}-\widehat{x}_{k-1\vert k-1})(x_k-\widehat{x}_{k\vert k-1})^{\top}\vert \theta_k\big] F_{k+1}^{\top}\\
&- F_k \mathbb{E}\big[(x_{k-1}-\widehat{x}_{k-1\vert k-1})\overline{y}_k^{\top}\vert \theta_k\big]K_k^{\top} F_{k+1}^{\top}\\
& -K_k \mathbb{E}\big[\overline{y}_k (x_k -\widehat{x}_{k\vert k-1})^{\top} \vert \theta_k\big]F_{k+1}^{\top}\\
&+ K_k \mathbb{E}\big[\overline{y}_k \overline{y}_k^{\top} \vert \theta_k\big] K_k^{\top} F_{k+1}^{\top}\\
&+\mathbb{E}\big[v_k(x_k- \widehat{x}_{k\vert k-1})^{\top}\vert \theta_k\big] F_{k+1}^{\top}\\
&-\mathbb{E}\big[v_k \overline{y}_k^{\top}  \vert \theta_k\big]K_k^{\top} F_{k+1}^{\top}.
\end{align*}
The proof is complete on account that $\mathbb{E}\big[\overline{y}_k \overline{y}_k^{\top} \vert \theta_k\big] = H_k P_{k\vert k-1} H_k^{\top} + R_k$ and $F_k \Sigma_{k\vert k-1} = P_{k\vert k-1} - Q_k$. $\square$

\subsection*{\textbf{F}: Proof of Proposition \ref{prop:propscore}}
Using the identity (\ref{eq:mainidentity}) and the result of Lemma \ref{lem:smoothing},
\begin{align*}
&\frac{\partial f(x_k,\underline{y}_k\vert \theta_k)}{\partial x_k}=\mathbb{E}\Big[\frac{\partial \log f(z_k\vert \theta_k)}{\partial x_k}\Big\vert x_k,\underline{y}_k,\theta_k\Big]\\
&=-H_k^{\top}R_k^{-1}\big(H_k x_k -y_k\big)\\
&\hspace{0cm}-Q_k^{-1}\big(x_k-F_k \mathbb{E}\big[x_{k-1}\big\vert x_k,\underline{y}_k,\theta_k\big] -G_k u_k\big)\\
&=-H_k^{\top}R_k^{-1}\big(H_k x_k -y_k\big)-Q_k^{-1}\big(x_k-G_k u_k\big) \\
&\hspace{0cm}+Q_k^{-1}\big[F_k \widehat{x}_{k-1\vert k-1} +F_k\Sigma_{k\vert k-1}P_{k\vert k-1}^{-1}(x_k-\widehat{x}_{k\vert k-1})\big]\\
&=-(H_k^{\top}R_k^{-1}H_k + P_{k\vert k-1}^{-1})(x_k -\widehat{x}_{k\vert k-1}) + H_k^{\top} R_k^{-1} \overline{y}_k,
\end{align*}
where the last equality was obtained on account that
\begin{align*}
F_k\widehat{x}_{k-1\vert k-1}=& \widehat{x}_{k\vert k-1}- G_k u_k\\
F_k\Sigma_{k\vert k-1}=& P_{k\vert k-1}- Q_k.
\end{align*}

The observed information matrix $J_{\xi}(x_k,\underline{y}_k\vert\theta_k)$ is derived using (\ref{eq:fisher}). After a rather long calculations, we obtain
\begin{align*}
&\mathcal{S}(x_k,\underline{y}_k\vert\theta_k)\mathcal{S}^{\top}(x_k,\underline{y}_k\vert\theta_k) \\
&\hspace{1cm}-\mathbb{E}\Big[\frac{\partial \log f(z_k \vert \theta_k)}{\partial x_k} \frac{\partial \log f(z_k \vert \theta_k)}{\partial x_k^{\top}}    \Big\vert \xi_k,\theta_k \Big] \\
&\hspace{1cm}=(P_{k\vert k-1}^{-1} -Q_k^{-1})(x_k-\widehat{x}_{k\vert k-1})\\
&\hspace{2cm}\times (x_k-\widehat{x}_{k\vert k-1})^{\top}(P_{k\vert k-1}^{-1} -Q_k^{-1})\\
&\hspace{1cm}-Q_k^{-1}F_k\mathbb{E}\Big[(x_{k-1}-\widehat{x}_{k-1\vert k-1})\\
&\hspace{2cm}\times(x_{k-1}-\widehat{x}_{k-1\vert k-1})^{\top}\big\vert x_k,\underline{y}_k,\theta_k\Big]F_k^{\top}Q_k^{-1}.
\end{align*}
It remains to evaluate the conditional expectation. However, as it involves only the random variable $x_{k-1}$, then following the Bayes formula (\ref{eq:bayes2}) the expectation is equal after by adding and subtracting $\mathbb{E}[x_{k-1}\vert x_k,\underline{y}_{k-1},\theta_k]$ to
\begin{align*}
&\mathbb{E}\Big[(x_{k-1}-\widehat{x}_{k-1\vert k-1})(x_{k-1}-\widehat{x}_{k-1\vert k-1})^{\top}\big\vert x_k,\underline{y}_{k-1},\theta_k\Big]\\
&=\mathbb{E}\Big[\big(x_{k-1}-\mathbb{E}[x_{k-1}\vert x_k,\underline{y}_{k-1},\theta_k]\big)\\
&\hspace{1cm}\times \big(x_{k-1}-\mathbb{E}[x_{k-1}\vert x_k,\underline{y}_{k-1},\theta_k]\big)^{\top}\big\vert x_k,\underline{y}_{k-1},\theta_k\Big]\\
&+\mathbb{E}\Big[\big(\mathbb{E}[x_{k-1}\vert x_k,\underline{y}_{k-1},\theta_k] - \widehat{x}_{k-1\vert k-1}\big)\\
&\hspace{1cm}\times \big(\mathbb{E}[x_{k-1}\vert x_k,\underline{y}_{k-1},\theta_k] - \widehat{x}_{k-1\vert k-1}\big)^{\top}\big\vert x_k,\underline{y}_{k-1},\theta_k\Big]\\
&=\mathbb{E}\Big[\big(x_{k-1}-\mathbb{E}[x_{k-1}\vert x_k,\underline{y}_{k-1},\theta_k]\big)\\
&\hspace{1cm}\times \big(x_{k-1}-\mathbb{E}[x_{k-1}\vert x_k,\underline{y}_{k-1},\theta_k]\big)^{\top}\big\vert x_k,\underline{y}_{k-1},\theta_k\Big]\\
&+\Sigma_{k\vert k-1}P_{k\vert k-1}^{-1}\big(x_k-\widehat{x}_{k\vert k-1}\big)\big(x_k-\widehat{x}_{k\vert k-1}\big)^{\top}P_{k\vert k-1} \Sigma_{k\vert k-1}^{\top},
\end{align*}
where the last equality was due to Lemma \ref{lem:smoothing}. Therefore,
\begin{align*}
&Q_k^{-1}F_k\mathbb{E}\Big[(x_{k-1}-\widehat{x}_{k-1\vert k-1})\\
&\hspace{1.5cm}\times(x_{k-1}-\widehat{x}_{k-1\vert k-1})^{\top}\big\vert x_k,\underline{y}_k,\theta_k\Big]F_k^{\top}Q_k^{-1}\\
&=Q_k^{-1}F_k\mathbb{E}\Big[\big(x_{k-1}-\mathbb{E}[x_{k-1}\vert x_k,\underline{y}_{k-1},\theta_k]\big)\\
&\hspace{0.5cm}\times \big(x_{k-1}-\mathbb{E}[x_{k-1}\vert x_k,\underline{y}_{k-1},\theta_k]\big)^{\top}\big\vert x_k,\underline{y}_{k-1},\theta_k\Big]F_k^{\top}Q_k^{-1}\\
&+ Q_k^{-1}F_k\Sigma_{k\vert k-1}P_{k\vert k-1}^{-1}\big(x_k-\widehat{x}_{k\vert k-1}\big)\\
&\hspace{1cm}\times \big(x_k-\widehat{x}_{k\vert k-1}\big)^{\top}P_{k\vert k-1} \Sigma_{k\vert k-1}^{\top}F_k^{\top}Q_k^{-1}\\
&=Q_k^{-1}F_k\mathbb{E}\Big[\big(x_{k-1}-\mathbb{E}[x_{k-1}\vert x_k,\underline{y}_{k-1},\theta_k]\big)\\
&\hspace{0.5cm}\times \big(x_{k-1}-\mathbb{E}[x_{k-1}\vert x_k,\underline{y}_{k-1},\theta_k]\big)^{\top}\big\vert x_k,\underline{y}_{k-1},\theta_k\Big]F_k^{\top}Q_k^{-1}\\
&\hspace{1cm}+(P_{k\vert k-1}^{-1}-Q_k^{-1})\big(x_k-\widehat{x}_{k\vert k-1}\big)\\
&\hspace{3cm}\times \big(x_k-\widehat{x}_{k\vert k-1}\big)^{\top}(P_{k\vert k-1}^{-1}-Q_k^{-1}).
\end{align*}
Thus, incorporating this result we have
\begin{align*}
&\mathcal{S}(x_k,\underline{y}_k\vert\theta_k)\mathcal{S}^{\top}(x_k,\underline{y}_k\vert\theta_k) \\
&\hspace{1cm}-\mathbb{E}\Big[\frac{\partial \log f(z_k \vert \theta_k)}{\partial x_k} \frac{\partial \log f(z_k \vert \theta_k)}{\partial x_k^{\top}}    \Big\vert \xi_k,\theta_k \Big] \\
&=-Q_k^{-1}F_k\mathbb{E}\Big[\big(x_{k-1}-\mathbb{E}[x_{k-1}\vert x_k,\underline{y}_{k-1},\theta_k]\big)\\
&\hspace{0.25cm}\times \big(x_{k-1}-\mathbb{E}[x_{k-1}\vert x_k,\underline{y}_{k-1},\theta_k]\big)^{\top}\big\vert x_k,\underline{y}_{k-1},\theta_k\Big]F_k^{\top}Q_k^{-1}\\
&=-Q_k^{-1}F_k\mathbb{E}\Big[\big(x_{k-1}-\mathbb{E}[x_{k-1}\vert x_k,\underline{y}_{k-1},\theta_k]\big)\\
&\hspace{0.25cm}\times \big(x_{k-1}-\mathbb{E}[x_{k-1}\vert x_k,\underline{y}_{k-1},\theta_k]\big)^{\top}\big\vert\theta_k\Big]F_k^{\top}Q_k^{-1},
\end{align*}
where in the last equality we used the fact that the vector $x_{k-1}-\mathbb{E}[x_{k-1}\vert x_k,\underline{y}_{k-1},\theta_k]$ is orthogonal to any vector in the space spanned by vector $(x_k,\underline{y}_{k-1})$. By Lemma \ref{lem:smoothing},
\begin{align*}
&\mathbb{E}\Big[\big(x_{k-1}-\mathbb{E}[x_{k-1}\vert x_k,\underline{y}_{k-1},\theta_k]\big)\\
&\hspace{1cm}\times \big(x_{k-1}-\mathbb{E}[x_{k-1}\vert x_k,\underline{y}_{k-1},\theta_k]\big)^{\top}\big\vert\theta_k\Big]\\
&=\mathbb{E}\Big[\Big(x_{k-1}-\widehat{x}_{k-1\vert k-1} -\Sigma_{k\vert k-1} P_{k\vert k-1}^{-1}\big(x_k-\widehat{x}_{k\vert k-1}\big)\Big)\\
&\times \Big(x_{k-1}-\widehat{x}_{k-1\vert k-1} -\Sigma_{k\vert k-1} P_{k\vert k-1}^{-1}\big(x_k-\widehat{x}_{k\vert k-1}\big)\Big)^{\top}\big\vert \theta_k\Big]\\
&=\mathbb{E}\big[(x_{k-1}-\widehat{x}_{k-1\vert k-1})(x_{k-1}-\widehat{x}_{k-1\vert k-1})^{\top}\vert\theta_k\big]\\
&-\mathbb{E}\Big[(x_{k-1}-\widehat{x}_{k-1\vert k-1})(x_k-\widehat{x}_{k\vert k-1})^{\top}\big\vert \theta_k\Big]P_{k\vert k-1}^{-1}\Sigma_{k\vert k-1}^{\top}\\
&-\Sigma_{k\vert k-1}P_{k\vert k-1}^{-1}\mathbb{E}\Big[(x_k-\widehat{x}_{k\vert k-1})(x_{k-1}-\widehat{x}_{k-1\vert k-1})^{\top}\big\vert \theta_k\Big]\\
&\hspace{0.5cm}+\Sigma_{k\vert k-1}P_{k\vert k-1}^{-1}\mathbb{E}\Big[(x_k-\widehat{x}_{k\vert k-1})\\&\hspace{3cm}\times(x_{k}-\widehat{x}_{k\vert k-1})^{\top}\big\vert \theta_k\Big]P_{k\vert k-1}^{-1}\Sigma_{k\vert k-1}^{\top}\\
&=P_{k-1\vert k-1} - \Sigma_{k\vert k-1} P_{k\vert k-1}^{-1}\Sigma_{k\vert k-1}^{\top}.
\end{align*}
Finally, 
\begin{align*}
&\mathcal{S}(x_k,\underline{y}_k\vert\theta_k)\mathcal{S}^{\top}(x_k,\underline{y}_k\vert\theta_k) \\
&\hspace{1cm}-\mathbb{E}\Big[\frac{\partial \log f(z_k \vert \theta_k)}{\partial x_k} \frac{\partial \log f(z_k \vert \theta_k)}{\partial x_k^{\top}}    \Big\vert \xi_k,\theta_k \Big] \\
&=-Q_k^{-1}F_k\Big(P_{k-1\vert k-1} - \Sigma_{k\vert k-1} P_{k\vert k-1}^{-1}\Sigma_{k\vert k-1}^{\top}\Big)F_k^{\top} Q_k^{-1}\\
&=-Q_k^{-1}F_kP_{k-1\vert k-1}F_k^{\top} Q_k^{-1} +Q_k^{-1}F_k \Sigma_{k\vert k-1} P_{k\vert k-1}^{-1}F_k^{\top} Q_k^{-1}\\
&=-Q_k^{-1}\big(P_{k\vert k-1}-Q_k)Q_k^{-1} \\
&\hspace{1cm}+Q_k^{-1}\big(P_{k\vert k-1}-Q_k)P_{k\vert k-1}^{-1}\big(P_{k\vert k-1}-Q_k)Q_k^{-1}\\
&=P_{k\vert k-1}^{-1}-Q_k^{-1}.
\end{align*}
The proof is complete on account that $J_z(x_k,\underline{y}_k\vert\theta_k)=H_k^{\top}R_k^{-1}H_k + Q_k^{-1},$ see (\ref{eq:infomat}). The result equals to the one derived by directly taking derivative of $\mathcal{S}(x_k,\underline{y}_k\vert\theta_k)$, i.e.,
\begin{eqnarray*}
J_{\xi}(x_k,\underline{y}_k\vert\theta_k)=\frac{\partial \mathcal{S}(x_k,\underline{y}_k\vert\theta_k)}{\partial x_k}=H_k^{\top}R_k^{-1}H_k + P_{k\vert k-1}^{-1}. \square
\end{eqnarray*}

\subsection*{\textbf{G}: Proof of Lemma \ref{lem:fishermatrix2}}
Using the expression of score function $\mathcal{S}(x_k,\underline{y}_k\vert\theta_k)$,
\begin{align*}
&\mathbb{E}\big[\mathcal{S}(x_k,\underline{y}_k\vert\theta_k)\mathcal{S}^{\top}(x_k,\underline{y}_k\vert\theta_k)\big\vert\theta_k\big]\\
&\hspace{-0.25cm}=\mathbb{E}\Big[\Big(-(H_k^{\top}R_k^{-1}H_k + P_{k\vert k-1}^{-1})(x_k -\widehat{x}_{k\vert k-1}) + H_k^{\top}R_k^{-1}\overline{y}_k  \Big) \\
&\hspace{-0.25cm}\times \Big(-(H_k^{\top}R_k^{-1}H_k + P_{k\vert k-1}^{-1})(x_k -\widehat{x}_{k\vert k-1}) + H_k^{\top}R_k^{-1}\overline{y}_k  \Big)^{\top}   \Big\vert \theta_k\Big]\\
&\hspace{-0.25cm}=(H_k^{\top}R_k^{-1}H_k + P_{k\vert k-1}^{-1})\mathbb{E}\big[(x_k-\widehat{x}_{k\vert k-1})(x_k-\widehat{x}_{k\vert k-1})^{\top}\vert\theta_k\big]\\
&\hspace{0.25cm}\times (H_k^{\top}R_k^{-1}H_k + P_{k\vert k-1}^{-1})-(H_k^{\top}R_k^{-1}H_k + P_{k\vert k-1}^{-1})\\
&\hspace{0.25cm}\times\mathbb{E}\big[(x_k-\widehat{x}_{k\vert k-1})\overline{y}_k^{\top}\vert\theta_k\big] R_k^{-1}H_k -H_k^{\top}R_k^{-1}\\
&\hspace{0.25cm}\times \mathbb{E}\big[\overline{y}_k(x_k-\widehat{x}_{k\vert k-1})^{\top}\vert\theta_k\big](H_k^{\top}R_k^{-1}H_k + P_{k\vert k-1}^{-1})\\
&\hspace{0.25cm}+H_k^{\top} R_k^{-1}\mathbb{E}\big[\overline{y}_k \overline{y}_k^{\top}\vert\theta_k\big]R_k^{-1} H_k\\
&\hspace{-0.25cm}=H_k^{\top} R_k^{-1} H_k + P_{k\vert k-1}^{-1},
\end{align*}
where the final result was obtained after replacing 
\begin{align*}
&\mathbb{E}\big[(x_k-\widehat{x}_{k\vert k-1})(x_k-\widehat{x}_{k\vert k-1})^{\top}\vert\theta_k\big]=P_{k\vert k-1}\\
&\mathbb{E}\big[(x_k-\widehat{x}_{k\vert k-1})\overline{y}_k^{\top}\vert\theta_k\big]=P_{k\vert k-1} H_k^{\top}\\
&\mathbb{E}\big[\overline{y}_k \overline{y}_k^{\top}\vert\theta_k\big]=H_kP_{k\vert k-1}H_k^{\top} + R_k.
\end{align*}
The second identity follows from Proposition \ref{prop:propscore} and (\ref{eq:fisher}). $\square$

\subsection*{\textbf{H}: Proof of Lemma \ref{lem:inverse}}
The proof follows from applying Woodbury matrix formula, see e.g. p.258 in \cite{Higham}. To be more details,
\begin{align*}
&\hspace{-0.35cm}\big(H_k^{\top} R_k^{-1}H_k + P_{k\vert k-1}^{-1}\big)\\
&\times \big[P_{k\vert k-1} - P_{k\vert k-1} H_k^{\top}\big(H_kP_{k\vert k-1} H_k^{\top} + R_k\big)^{-1} H_k P_{k\vert k-1}\big]\\
=&I-H_k^{\top}\big(H_kP_{k\vert k-1} H_k^{\top} + R_k\big)^{-1} H_k P_{k\vert k-1} +H_k^{\top} R_k^{-1}H_k P_{k\vert k-1} \\
&- H_k^{\top} R_k^{-1}H_kP_{k\vert k-1} H_k^{\top}\big(H_kP_{k\vert k-1} H_k^{\top} + R_k\big)^{-1} H_k P_{k\vert k-1}\\
=&(I+H_k^{\top} R_k^{-1}H_k P_{k\vert k-1})-(H_k^{\top} + H_k^{\top}R_k^{-1}H_kP_{k\vert k-1}H_k^{\top})\\
&\hspace{1cm}\times (H_k P_{k\vert k-1} H_k^{\top}+R_k)^{-1}H_k P_{k\vert k-1}\\
=&(I+H_k^{\top} R_k^{-1}H_k P_{k\vert k-1})-H_k^{\top}R_k^{-1}(H_kP_{k\vert k-1}H_k^{\top}+R_k)\\
&\hspace{1cm}\times (H_k P_{k\vert k-1} H_k^{\top}+R_k)^{-1}H_k P_{k\vert k-1},
\end{align*}
from which the proof is complete on account that the matrix $H_k P_{k\vert k-1} H_k^{\top}+R_k $ is assumed to be invertible. The second matrix identity follows from applying the first identity. $\square$

\medskip

\begin{wrapfigure}{l}{20mm}
    \includegraphics[width=1in,height=1.25in,clip,keepaspectratio]{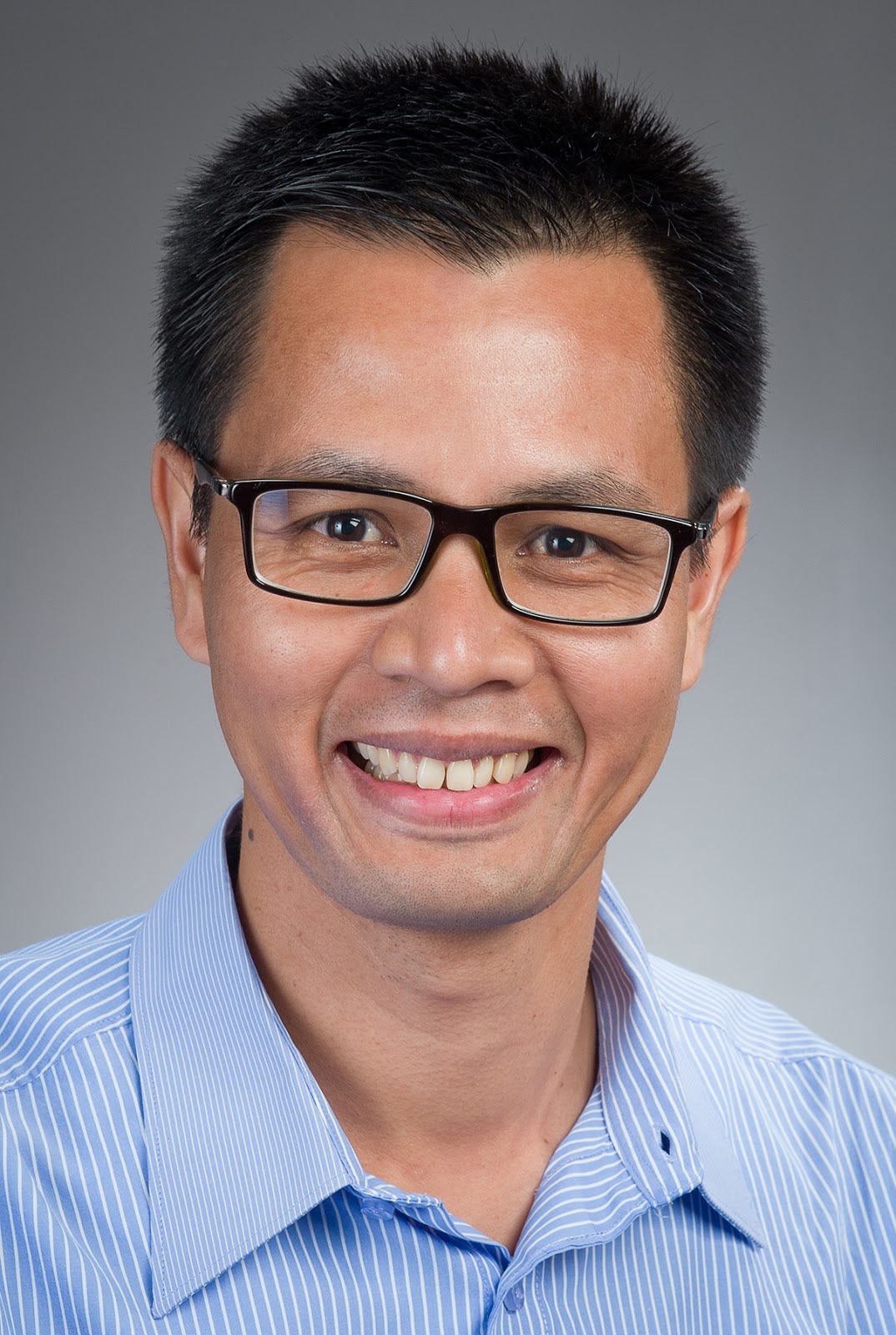}
  \end{wrapfigure}\par
  \textbf{Budhi Arta Surya} is a senior lecturer (equivalent to an associate professor) of statistics at the School of Mathematics and Statistics of Victoria University of Wellington, New Zealand. He received an ingenieur Ir. (equivalent to MSc) degree in applied mathematics (majoring in stochastic systems and control) from the University of Twente, Enschede, in August 2001 and a doctorate degree Dr. in mathematics (applied probability and stochastic process) in January 2007 from the University of Utrecht, both in Netherlands. His doctorate thesis was concerned with optimal stopping of L\'evy processes. Soon after graduating from Utrecht, he joined  Bank of America Corp. as a quantitative financial analyst, based mainly in Singapore, with direct reporting line to the Head of Quantitative Risk Management Group in Charlotte, USA. He was a visiting scholar to the Department of Industrial Engineering and Operations Research of Columbia University in May 2014 and to the Department of Technology, Operations and Statistics of Stern School of Business of New York University in September 2019.  \par
  

\begin{thebibliography}{99}

\bibitem[Anderson and Moore(1979)]{Anderson}
\textsc{Anderson, B.D.O. ad Moore, J.B.} (1979).
\emph{Optimal Filtering}. 
New Jersey: Prentice-Hall, Inc.

\bibitem[Astrom(1970)]{Astrom}
\textsc{Astrom, K. J.} (1970).
\emph{Introduction to Stochastic Control}. 
New York: Academic Press.

\bibitem[Bar-Shalom et al.(2001)]{Bar-Shalom}
\textsc{Bar-Shalom, Y., Li, X.R. and Kirubarajan, T.} (2001).
\emph{Estimation with Applications to Tracking and Navigation}. 
Wiley-Interscience.

\bibitem[Bergman(2001)]{Bergman} 
\textsc{Bergman, N.} (2001). 
Posterior Cram\'er-Rao lower bounds for sequential estimation. 
\emph{In: Doucet, A., de Freitas, N., Gordon, N. (eds) Sequential Monte Carlo Methods in Practice. Statistics for Engineering and Information Science}. Springer, New York.
321--338.

\bibitem[Berndt et al.(1974)]{BHHH}
\textsc{Berndt, E., Hall, B., Hall, R., and Hausman, J.} (1974).
Estimation and inference in nonlinear structural models.
\emph{Annals of Economic and Social Measurement}
\textbf{3}(4),
653--665.
 
\bibitem[Chui and Chen(2017)]{Chui}
\textsc{Chui, C.K. and Chen, G.} (2017).
\emph{Kalman Filtering with Real-Time Applications}, 5th Ed. 
Springer Nature.

\bibitem[Dempster et al.(1977)%
Dempster, Laird and Rubin]{Dempster}
\textsc{Dempster, A. P., Laird, N. M. and Rubin, D.B.} (1977).
Maximum likelihood from incomplete data via the EM algorithm (with
discussion).
\emph{Journal of the Royal Statistical Society Series B}
\textbf{39},
1--38.

\bibitem[Doucet et al.(2001)]{Doucet}
\textsc{Doucet, A., de Freitas, N. and Gordon, N.} (2001).
\emph{Sequential Monte Carlo Methods in Practice}.
Springer.

\bibitem[Doucet et al.(2000)]{Doucet2000}
\textsc{Doucet, A., Godsill, S. and Andrieu, C.} (2000).
On sequential Monte Carlo sampling methods for Bayesian filtering.
\emph{Statistics and Computing}
\textbf{10},
197--208.

\bibitem[Durbin and Koopman(2012)]{Durbin}
\textsc{Durbin, J. and Koopman, S.J.} (2012).
\emph{Time Series Analysis by State Space Methods}.
Oxford: Oxford University Press.

%\bibitem[Fisher(1925)]{Fisher} 
%\textsc{Fisher, R. A.} (1925). 
%Theory of statistical estimation. 
%\emph{Proceedings of the Cambridge Philosophical Society},
%\textbf{22}, 
%700--725. 

\bibitem[Freedman(2006)]{Freedman} 
\textsc{Freedman, D. A.} (2006). 
On the so-called ''Huber Sandwich Estimator'' and ''Robust Standard Errors''. 
\emph{The American Statistician}, 
\textbf{60}(4), 
299--302.

\bibitem[Gill and Levit(1995)]{Gill} 
\textsc{Gill, R.D. and Levit, B.Y.} (1995). 
Applications of the van Trees inequality: a Bayesian Cram\'er-Rao bound. 
\emph{Bernoulli}, 
\textbf{1}(1/2), 
59--79.

\bibitem[Godsill(2019)]{Godsill} 
\textsc{Godsill, S.} (2019). 
Particle filtering: the first 25 years and beyond. 
\emph{IEEE ICASSP},
7760--7764.

\bibitem[Gordon et al.(1993)]{Gordon} 
\textsc{Gordon, N., Salmond, D., and Smith, A.F.} (1993). 
Novel approach to nonlinear/non-Gaussian Bayesian state estimation. 
\emph{IEEE Proc. Radar Signal Process},
\textbf{140}(1/2), 
107--113.

\bibitem[Halmos and Savage(1949)%
Halmos and Savage]{Halmos}
\textsc{Halmos, P. R., and Savage, L. J.} (1949).
Application of the Radon-Nikodym theorem to the theory of sufficient statistics.
\emph{The Annals of Mathematical Statistics},
\textbf{20}, 
225--241.

\bibitem[Harvey(1989)]{Harvey}
\textsc{Harvey, A.C.} (1989).
\emph{Forecasting, structural time series models and the Kalman filter}.
Cambridge University Press.

\bibitem[Higham(2002)]{Higham}
\textsc{Higham, N.} (2002). 
\emph{Accuracy and Stability of Algorithms}, 2nd Ed.
Philadelpia: SIAM 

\bibitem[Horn and Johnson(2013)]{Horn} 
\textsc{Horn, R. A. and Johnson, C. R.} (2013). 
\emph{Matrix Analysis}. 
Cambridge: Cambridge University Press.

\bibitem[Jazwinski(1970)]{Jazwinski} 
\textsc{Jazwinski, A.H.} (1970). 
\emph{Stochastic Processes and Filtering Theory}. 
New York: Academic Press.


\bibitem[Kailath(1968)]{Kailath}
\textsc{Kailath, T. } (1968).
An innovations approach to least-squares estimation part I: linear filtering in additive white noise.
\emph{IEEE Transactions of Automatic Control},
\textbf{13}, 
646--655.

\bibitem[Kalman(1960)%
Kalman]{Kalman}
\textsc{Kalman, R. E} (1960).
A new approach to linear filtering and prediction problems. 
\emph{Trans. ASME}
\textbf{82},
p. 35.

\bibitem[Kelemen et. al(2006)]{Kelemen}
\textsc{Kelemen, J.Z., Kertesz-Farkas, A., Kocsor, A. and Puskas, L.G.} (2006).
Kalman filtering for disease-state estimation from microarray data. 
\emph{Bioinformatics}
\textbf{22},
3047--3053.

\bibitem[Kitagawa(2021)]{Kitagawa2021}
\textsc{Kitagawa, G.} (2021).
\emph{Introduction to Time Series Modeling with Applications in R}, 2nd Ed.
CRC Press.

\bibitem[Kitagawa(1996)]{Kitagawa}
\textsc{Kitagawa, G.} (1996).
Monte Carlo filter and smoother for non-Gaussian nonlinear state space models.
\emph{Journal of Computational and Graphical Statistics},
\textbf{5}(1), 
1--25.

\bibitem[Kitagawa(1993)]{Kitagawa93}
\textsc{Kitagawa, G.} (1993).
A Monte Carlo filtering and smoothing method for non-Gaussian nonlinear state space models.
\emph{Proceedings of the 2nd U.S.-Japan Joint Seminar on Statistical Time Series Analysis},
110--131.

\bibitem[Lange(1995)]{Lange}
\textsc{Lange, K.} (1995). 
A gradient algorithm locally equivalent to the EM algorithm. 
\emph{Journal of Royal Statistical Society: Series B (Statistical Methodology),} 
\textbf{57}(2), 
425--437.

\bibitem[Little and Rubin(2020)]{Little} 
\textsc{Little, R. J. A. and Rubin, D. B.} (2020). 
\emph{Statistical Analysis with Missing Data}. 
John Wiley \& Sons, Inc.

\bibitem[Louis(1982)%
Louis]{Louis}
\textsc{Louis, T. A.} (1982).
Finding the observed information matrix when using the EM algorithm. 
\emph{Journal of the Royal Statistical Society Series B}
\textbf{44},
226--233.

\bibitem[McLachlan and Krishnan(2008)%
McLachlan and Krishnan]{McLachlan} 
\textsc{McLachlan, G. J. and Krishnan, T.} (2008). 
\emph{The EM Algorithm and Extensions}. 
John Wiley \& Sons, Inc.

\bibitem[Musoff and Zarchan(2005)%
Musoff and Zarchan]{Musoff2005} 
\textsc{Musoff, H. and Zarchan, P.} (2005). 
\emph{Fundamentals of Kalman Filtering: A Practical Approach}, 2nd Ed.
American Institute of Aeronautics and Astronautics, Inc.

\bibitem[Orchard and Woodbury(1972)]{Orchard}
\textsc{Orchard, T. and Woodbury, M. A.} (1972). 
A missing information principle: Theory and applications. 
\emph{Proceedings of the 6th Berkeley Symposium on Mathematical Statistics and Probability}, 
\textbf{1}, Berkeley, CA: University of California, 
697--715.

\bibitem[Petersen and Pedersen (2012)]{Petersen}
\textsc{Petersen, K.B. and Pedersen, M.S.} (2012).
\emph{The Matrix Cookbook}.
\url{http://matrixcookbook.com}

\bibitem[Pit and Shephard(1999)]{Shephard}
\textsc{Pitt, M.K and Shephard, N.} (1999).
Filtering via simulation: Auxiliary particle filters.
\emph{Journal of the American Statistical Association}
\textbf{94}(446),
590--599.

\bibitem[RTeam(2013)]{RTeam}
\textsc{R Core Team.} (2013). 
\textbf{R}: A language and environment for statistical computing. R foundation for Statistical Computing, Vienna, Austria.  https://www.R-project.org

\bibitem[Ramadan and Bitmead(2022)]{Ramadan}
\textsc{Ramadan, M. S. and Bitmead, R. R.} (2022).
Maximum likelihood recursive state estimation using the expectation maximization algorithm.
\emph{Automatica}
\textbf{144},
110482.

\bibitem[Rauch et al.(1965)%
Rauch, Tung, and Striebel(1965)]{Rauch}
\textsc{Rauch, H.E., Tung, F. and Striebel, C.T.} (1965).
Maximum likelihood estimates of linear dynamical systems.
\emph{AIAA Journal}
\textbf{3}(8),
1445--1450.

\bibitem[Ristic et al.(2004)]{Ristic}
\textsc{Ristic, B., Arulampalam, S., and Gordon, N.} (2004).
\emph{Beyond the Kalman Filter: Particle Filters for Tracking Applications.}
Artech House Publishers.

\bibitem[Schervish(1995)Schervish]{Schervish}
\textsc{Schervish, M. J.} (1995).
\emph{Theory of Statistics}.
New York: Springer.

\bibitem[Schmidt(1966)]{Schmidt66}
\textsc{Schmidt, S.F.} (1966).
Application of state-space methods to navigation problems.
\emph{Advances in Control Systems}
\textbf{3},
293--340.

\bibitem[Sch\"on et al.(2011)]{Schon}
\textsc{Sch\"on, T.B., Wills, A. and Ninnes, B.} (2011).
System identification of nonlinear state-space models.
\emph{Automatica}
\textbf{47},
39--49.

\bibitem[Smith et al.(1962)]{Smith62}
\textsc{Smith, G.L., Schmidt, S.F., and McGee, L.A.} (1962).
Applications of statistical filter theory to the optimal estimation of position and velocity on board a circumlunar vehicle.
\emph{NASA TR R-135}.

\bibitem[Surya(2022)Surya]{Surya2022}
\textsc{Surya, B. A.} (2022).
Some results on maximum likelihood from incomplete data: finite sample properties and improved M-estimator for resampling.
\emph{Preprint}. \url{https://arxiv.org/abs/2108.01243}

\bibitem[Tichavsk\'y et al.(1998)]{Tichavsky}
\textsc{Tichavsk\'y, P., Muravchik, C.H., and Nehoral, A.} (1998).
Posterior Cram\'er-Rao bounds for discrete-time nonlinear filtering.
\emph{IEEE Transactions on Signal Processing}
\textbf{46}(5),
1386--1396.

\bibitem[Triantafyllopoulos(2021)]{Trianta}
\textsc{Triantafyllopoulos, K.} (2021).
\emph{Bayesian Inference of State Space Models: Kalman Filtering and Beyond}.
Cham: Springer.

\bibitem[Tripathi(1999)]{Tripathi}
\textsc{Tripathi, G.} (1999).
A matrix extension of the Cauchy-Schwarz inequality.
\emph{Economic Letters}
\textbf{63},
1--3.

\bibitem[Van Trees(1968)]{VanTrees68}
\textsc{Van Trees, H.} (1968).
Detection, Estimation, and Modulation Theory. Part I.
New York: \emph{Wiley}.

\bibitem[Van Trees et al.(2013)]{VanTrees}
\textsc{Van Trees, H., Bell, K.L., and Tian, Z.} (2013).
Detection, Estimation, and Modulation Theory. Part I: Detection, Estimation and Filtering Theory.
Hoboken: \emph{John Wiley \& Sons}.

\bibitem[Vaziri et al.(2018)]{Vaziri}
\textsc{Vaziri, S., Koehl, P., Aviran, S.} (2018).
Extracting information from RNA SHAPE data: Kalman filtering approach.
\emph{PLoS ONE}
\textbf{13}(11),
e0207029.

\bibitem[Zhao et al.(2019)]{Zhao}
\textsc{Zhao, Y., Fritsche, C., Hendeby, G., Yin, F., Chen, T., and Gunnarsson, F.} (2019).
Cram\'er-Rao bounds for filtering based on Gaussian process state-space models.
\emph{IEEE Transactions on Signal Processing}
\textbf{67}(23),
5936--5951.

%\bibitem[van der Vaart(2000)]{vanderVaart2000}
%\textsc{van der Vaart, A.W.} (2000).
%\emph{Asymptotic Statistics}.
%Cambridge: Cambridge University Press.

\end{thebibliography}
  \end{document}